\documentclass{amsart}

\usepackage{fullpage}
\usepackage{amsthm}
\usepackage[T1]{fontenc}
\usepackage[utf8]{inputenc}

\usepackage{moreverb}
\usepackage{float}
\usepackage[unicode,bookmarksopen=false,bookmarksdepth=5]{hyperref}

\usepackage{tikz}
\usetikzlibrary{decorations.pathreplacing,hobby,backgrounds,trees,arrows,automata,shapes,positioning,fit,patterns,calc,matrix}

\usepackage[english]{babel}

\usepackage{enumitem}
\setlist[1]{topsep=0.8ex,itemsep=.3ex}
\usepackage{xspace}

\newcommand{\upset}[2][\Cs]{\ensuremath{{\mathord\uparrow_{#1}\,#2}}}

\newcommand\vari{quotienting Boolean algebra\xspace}

\newcommand\varis{quotienting Boolean algebras\xspace}

\newcommand\pvari{quotienting lattice\xspace}

\newcommand\pvaris{quotienting lattices\xspace}

\newcommand\Bs{\ensuremath{\mathscr{B}}\xspace}
\newcommand\Cs{\ensuremath{\mathscr{C}}\xspace}
\newcommand\Ds{\ensuremath{\mathscr{D}}\xspace}

\newcommand\Hs{\ensuremath{\mathscr{H}}\xspace}

\newcommand\Ss{\ensuremath{\mathscr{S}}\xspace}

\newcommand\Xs{\ensuremath{\mathscr{X}}\xspace}

\newcommand\nfa{{{NFA}}\xspace}

\renewcommand\min{\ensuremath{\text{\scriptsize min}}\xspace}
\renewcommand\max{\ensuremath{\text{\scriptsize max}}\xspace}

\newcommand\sigenr{\ensuremath{<,+1,\min,\max,\varepsilon}\xspace}
\newcommand\sigen{\ensuremath{<,+1,\min,\max,\varepsilon}\xspace}

\newcommand\fo{\ensuremath{\textup{FO}}\xspace}
\newcommand\fow{\mbox{\ensuremath{\fo(<)}}\xspace}
\newcommand\fows{\mbox{\ensuremath{\fo(\sigenr)}}\xspace}

\newcommand\foeq{\mbox{\ensuremath{\fo(=)}}\xspace}

\newcommand\folab{\mbox{\ensuremath{\fo(\emptyset)}}\xspace}

\newcommand\reg{\ensuremath{\textup{REG}}\xspace}

\newcommand\sic[1]{\ensuremath{\Sigma_{#1}}\xspace}

\newcommand\sicu{\sic{1}}

\newcommand\sicd{\sic{2}}

\newcommand\sict{\sic{3}}

\newcommand\pic[1]{\ensuremath{\Pi_{#1}}\xspace}

\newcommand\picu{\pic{1}}

\newcommand\picd{\pic{2}}

\newcommand\pict{\pic{3}}

\newcommand\bsc[1]{\ensuremath{\Bs\Sigma_{#1}}\xspace}
\newcommand\bsw[1]{\ensuremath{\Bs\Sigma_{#1}(<)}\xspace}

\newcommand\bscu{\bsc{1}}

\newcommand\bscd{\bsc{2}}
\newcommand\bswd{\bsw{2}}

\newcommand\bsct{\bsc{3}}

\newcommand\pt{\ensuremath{\textup{PT}}\xspace}

\newcommand\wat{\ensuremath{\textup{WAT}}\xspace}
\newcommand\at{\ensuremath{\textup{AT}}\xspace}

\newcommand\sfr{\ensuremath{\textup{SF}}\xspace}

\newcommand\patt[1]{\ensuremath{\ensuremath{#1\textup{-ATT}}}\xspace}
\newcommand\datt{\patt{d}}
\newcommand\att{\ensuremath{\textup{ATT}}\xspace}

\newcommand\sttp[1]{\ensuremath{\textup{ST}[#1]}\xspace}

\newcommand\bool[1]{\ensuremath{Bool(#1)}\xspace}
\newcommand\pol[1]{\ensuremath{Pol(#1)}\xspace}
\newcommand\bpol[1]{\ensuremath{BPol(#1)}\xspace}

\newcommand\polp[2]{\ensuremath{Pol_{#2}(#1)}\xspace}
\newcommand\polk[1]{\polp{#1}{k}}

\newcommand\bpolp[2]{\ensuremath{BPol_{#2}(#1)}\xspace}
\newcommand\bpolk[1]{\bpolp{#1}{k}}

\tikzstyle{char}=[anchor=mid,inner sep=0pt]
\tikzstyle{pebb}=[line width=1.5pt,->]
\tikzstyle{fiar}=[shorten >= 1pt,thick,->]
\tikzstyle{bag}=[inner sep=1pt]

\tikzstyle{non}=[inner sep=1pt]
\tikzstyle{lbox}=[rounded corners=5pt,draw=black!60,very thick,align=center]
\tikzstyle{wbox}=[rounded corners=5pt,align=center,minimum height=0.65cm,inner ysep=0pt]
\tikzstyle{gbox}=[rounded corners=5pt,fill=green!20,align=center,minimum height=0.65cm,inner ysep=0pt]
\tikzstyle{bbox}=[rounded corners=5pt,fill=blue!20,align=center,minimum height=0.65cm,inner ysep=0pt]
\tikzstyle{rbox}=[rounded corners=5pt,fill=red!20,align=center,minimum height=0.65cm,inner ysep=0pt]
\tikzstyle{ledg}=[draw=black,very thick]

\tikzstyle{linc}=[fill=white,draw=black,rotate=90,inner sep=2pt,thick,circle]
\tikzstyle{linc2}=[fill=white,draw=black,rotate=0,inner sep=2pt,thick,circle]
\tikzstyle{leg}=[draw,minimum width = 1.0cm,minimum height = 0.5cm]
\tikzstyle{tag}=[draw,fill=white,sloped,circle,inner sep=1pt]

\tikzstyle{word}=[draw,very thick,rectangle,rounded corners=2pt,anchor=west,minimum height=0.5cm]
\tikzstyle{word2}=[rectangle,rounded corners=2pt,anchor=west,minimum height=0.5cm]

\tikzstyle{stmono}=[inner sep=2pt]

\tikzset{every state/.style={draw=blue!50!green,very thick,fill=blue!50!green!20}}
\tikzset{statesub/.style={state,minimum size=1.3cm,inner sep=1pt}}
\tikzset{pattstate/.style={state,draw=red!50!yellow,line width=2pt,fill=red!50!yellow!20}}
\tikzset{pdotstate/.style={state,minimum size=0.75cm,inner sep=0.5pt,draw=red!50!yellow,line
    width=2pt,dashed,fill=red!50!yellow!20}}
\tikzstyle{trans}=[shorten >= 1pt,thick,->]

\makeatletter
\tikzstyle{initial by arrow}=   [after node path=
{
  {
    [to path=
    {
      [->,double=none,shorten >= 1pt,thick,every initial by arrow]
      ([shift=(\tikz@initial@angle:\tikz@initial@distance)]\tikztostart.\tikz@initial@angle)
      node [shape=coordinate,anchor=\tikz@initial@anchor,draw] {\tikz@initial@text}
      -- (\tikztostart)}]
    edge ()
  }
}]
\makeatother

\makeatletter
\tikzstyle{accepting by arrow}=   [after node path=
{
  {
    [to path=
    {
      [->,double=none,shorten >= 1pt,thick,every accepting by arrow]
      (\tikztostart) --
      ([shift=(\tikz@accepting@angle:\tikz@accepting@distance)]\tikztostart.\tikz@accepting@angle)
      node [shape=coordinate,anchor=\tikz@accepting@anchor,draw] {\tikz@accepting@text}
    }
    ]
    edge ()
  }
}]
\makeatother

\newcommand\mhline[3][]{
  \pgfmathtruncatemacro\hc{#3-1}
  \node[fit=(#2-#3-1),inner sep=0pt](R){};
  \node[fit=(#2-\hc-1),inner sep=0pt](L){};
  \node (K) at ($(R)!0.5!(L)$) {};
  \draw[#1] (K -| #2.west) -- (K -| #2.east);
}

\newcommand\polrelp[1]{\ensuremath{\leqslant_{#1}}\xspace}
\newcommand\polrelk{\polrelp{k}}

\newcommand\bpolrelp[1]{\ensuremath{\simeq_{#1}}\xspace}
\newcommand\bpolrelk{\bpolrelp{k}}

\newcommand\atteq[1]{\ensuremath{\simeq^{#1}}\xspace}
\newcommand\datteq{\atteq{d}}

\newcommand{\dotdp}[1]{\ensuremath{\textup{DD}[#1]}\xspace}
\newcommand{\dotzer}{\dotdp{0}}
\newcommand{\dothone}{\dotdp{\frac{1}{2}}}
\newcommand{\dotone}{\dotdp{1}}
\newcommand{\dothtwo}{\dotdp{\frac{3}{2}}}
\newcommand{\dottwo}{\dotdp{2}}
\newcommand{\doththree}{\dotdp{\frac{5}{2}}}
\newcommand{\dotthree}{\dotdp{3}}
\newcommand{\dothfour}{\dotdp{\frac{7}{2}}}

\newcommand{\stzer}{\sttp{0}}
\newcommand{\sthone}{\sttp{\frac{1}{2}}}
\newcommand{\stone}{\sttp{1}}
\newcommand{\sthtwo}{\sttp{\frac{3}{2}}}
\newcommand{\sttwo}{\sttp{2}}
\newcommand{\sththree}{\sttp{\frac{5}{2}}}
\newcommand{\stthree}{\sttp{3}}
\newcommand{\sthfour}{\sttp{\frac{7}{2}}}

\tikzstyle{nor}=[minimum size=0.35cm,draw,rectangle,inner sep=2pt]
\tikzstyle{nod}=[minimum size=0.35cm,draw,circle,inner sep=2pt]
\tikzstyle{nof}=[minimum size=0.35cm,draw,circle,double,double distance=1pt]

\tikzstyle{nol}=[minimum size=0.35cm,draw,rectangle,inner sep=1pt,rotate=90]

\tikzstyle{ar}=[line width=0.5pt,->,double]
\tikzstyle{siar}=[line width=1.5pt,->]

\theoremstyle{plain}
\newtheorem{theorem}{Theorem}[section]
\newtheorem{proposition}[theorem]{Proposition}
\newtheorem{lemma}[theorem]{Lemma}
\newtheorem{corollary}[theorem]{Corollary}
\newtheorem{fact}[theorem]{Fact}
\newtheorem{problem}[theorem]{Problem}
\newtheorem{example}[theorem]{Example}
\newtheorem*{claim}{Claim}
\newtheorem*{remark}{Remark}
\theoremstyle{definition}

\usepackage{amsmath,amsfonts,amssymb}
\usepackage[mathscr]{euscript}

\newcommand\cont[1]{\ensuremath{{\mathord{\mathrm{alph}}}(#1)}\xspace}

\newcommand\nat{\ensuremath{\mathbb{N}}\xspace}

\definecolor{ocre}{RGB}{40,130,80}
\definecolor{bookgreen}{RGB}{40,130,80}
\definecolor{bookblue}{RGB}{50,110,150}
\definecolor{bookred}{RGB}{180,15,47}

\let\dropQED\relax

\usepackage{csquotes}
\usepackage{lmodern}
\usepackage[style=numeric,citestyle=numeric,sorting=nyt,sortcites=false,autopunct=true,autolang=hyphen,hyperref=true,giveninits=true,abbreviate=true,backend=bibtex]{biblatex}

\usepackage{calc}

\defbibheading{bibempty}{}

\title{Generic Results for Concatenation Hierarchies}

\author{Thomas Place}
\address{LaBRI, Bordeaux University, France}
\email{tplace@labri.fr}
\urladdr{www.labri.fr/perso/tplace}
\thanks{Funded by the DeLTA project (ANR-16-CE40-0007)}

\author{Marc Zeitoun}
\address{LaBRI, Bordeaux University, France}
\email{mz@labri.fr}
\urladdr{www.labri.fr/perso/zeitoun}
\thanks{Funded by the DeLTA project (ANR-16-CE40-0007)}

\begin{document}

\begin{abstract}
  In the theory of formal languages, the understanding of concatenation hierarchies of regular languages is one of the most fundamental and challenging  topic. In this paper, we survey progress made in the comprehension of this problem since 1971, and we establish new generic statements regarding this~problem.
\end{abstract}

\maketitle

\section{Introduction}\label{sec:introduction}
This paper has a dual vocation. The first is to outline progress seen during the last 50 years about \emph{concatenation hierarchies} of regular languages. The second is to provide \emph{generic statements} and \emph{elementary proofs} of some of the core results on this topic, which were obtained previously in restricted cases. In this introduction, we present the historical background, first highlighting the motivations and the key ideas that emerged since the mid 60s. In a second part, we describe the contributions of the paper, which are either new proofs of existing results or generalizations thereof.

\medskip\noindent
\textbf{Historical background: a short survey of 50 years of research.} Concatenation hierarchies were introduced in order to understand the interplay between two basic constructs used to build \emph{regular languages}: Boolean operations and concatenation. The story started in 1956 with Kleene's theorem~\cite{kleene}, one of the key results in automata theory. It states that languages of finite words recognized by finite automata are exactly the ones that can be described by regular expressions, \emph{i.e.}, are built from the singleton languages and the empty set using a finite number of times operations among three basic ones: union, concatenation, and iteration (a.k.a.\ Kleene~star).

As Kleene's theorem provides another syntax for regular languages, it makes it possible to classify them according to the hardness of describing a language by such an expression. The notion of star-height was designed for this purpose. The \emph{star-height} of a regular expression is its maximum number of nested Kleene stars. The \emph{star-height} of a regular language is the minimum among the star-heights of all regular expressions that define the language. Since there are languages of arbitrary star-height~\cite{eggan1963,DejeanSchutz:1966}, this makes the notion an appropriate complexity measure, and justifies the question of computing the star-height of a regular language, which was raised in 1963 by Eggan~\cite{eggan1963} (see also~Brzozowski~\cite{BrzozowskiOpen80}): ``Given a regular language and a natural number~$n$, is there an expression of star-height $n$ defining the language?''

This question, called the \emph{star-height problem}, is an instance of the \emph{membership problem}. Given a class $\Cs$ of regular
languages, the membership problem for \Cs simply asks whether \Cs is a decidable class, that is:

\smallskip
\begin{tabular}{rl}
  {\bf Input:\quad}  &  A regular language $L$. \\
  {\bf Output:\quad} &  Does $L$ belong to \Cs?
\end{tabular}

\smallskip Thus, the star-height problem asks whether membership is decidable for each class $\Hs_n$ consisting of languages of star-height~$n$. It was first solved in 1988 by Hashiguchi~\cite{Hashiguchi:1988}, but it took 17 more years to obtain simpler proofs, see~\cite{Kirsten:Distance-desert-automata-star:2005,BookSaka,mb-SHviagames}.

\smallskip
Kleene's theorem also implies that adding \emph{complement} to our set of basic operations does not make it possible to define more languages. Therefore, instead of just considering regular expressions, one may consider \emph{generalized} regular expressions, where complement is allowed (in addition to union, concatenation and Kleene star). This yields the notion of \emph{generalized star-height}, which is defined as the star-height, but replacing ``regular expression'' by ``generalized regular expression''. One may then ask the same question: is there an algorithm to compute the \emph{generalized star-height} of a regular language? Despite its simple statement, this question, also raised in 1980 by Brzozowski~\cite{BrzozowskiOpen80,brpbs80}, is still open. Even more, one does not know whether there exists a regular language of generalized star-height greater than~1. In other terms, membership is open for the class of languages of generalized star-height~1 (see~\cite{jep-openreg35} for a historical presentation).

\medskip This makes it relevant to already focus on languages of generalized star height~0, \emph{i.e.}, that can be described using only union, concatenation and Boolean operations (including complement), but \emph{without} the Kleene star. Such languages are called \emph{star-free}. It turns out that even this restricted problem is difficult. It was solved in 1965 by Schützenberger in a seminal paper.

\begin{theorem}[Schützenberger~\cite{schutzsf}]
  Membership is decidable for the class of star-free languages.
\end{theorem}

Star-free languages rose to prominence because of numerous characterizations, and in particular the logical one, due to McNaughton and Papert (1971). The key point is that one may describe languages with logical sentences: any word may be viewed as a logical structure made of a linearly ordered sequence of positions, each carrying  a label. In first-order logic over words (denoted by \fow), one may quantify these positions, compare them with a predicate ``$<$'' interpreted as the (strict) linear order, and check their labels (for any letter $a$, a unary ``label'' predicate selecting positions with label $a$ is available). Therefore, each \fow sentence states a property over words and defines the language of all words that satisfy it.

\begin{theorem}[McNaughton \& Papert~\cite{mnpfosf}]\label{thm:mnp}
  For a regular language $L$, the following properties are \hbox{equivalent}:
  \begin{itemize}
  \item $L$ is star-free.
  \item $L$ can be defined by an \fow sentence.
  \end{itemize}
\end{theorem}

Let us point out that this connection between star-free and first-order definable languages is rather intuitive. Indeed, there is a clear correspondence between union, intersection and complement for star-free languages with the Boolean connectives in \fow sentences. Moreover, concatenation corresponds to existential quantification.

\medskip
Just as the star-height measures how complex a regular language is, a natural complexity for star-free languages is the required number of \emph{alternations} between \emph{concatenation} and \emph{complement} operations for building a given star-free language from basic ones. This led Brzozowski and Cohen~\cite{BrzoDot} to introduce in 1971 a hierarchy of classes of regular languages, called the \emph{dot-depth hierarchy}. It classifies all star-free languages into full levels, indexed by natural numbers: 0, 1, 2,\ldots, and half levels, indexed by half natural numbers: $\frac12$,~$\frac32$,~$\frac52$, etc. Roughly speaking, levels count the number of alternations between concatenation and Boolean operations that are necessary to express a given star-free language.

More formally, the hierarchy is built by using, alternately, two closure operations starting from level 0: \emph{Boolean} and \emph{polynomial} closures.  Given a class of languages \Cs, its \emph{Boolean closure}, denoted \bool{\Cs}, is the smallest Boolean algebra containing \Cs. Polynomial closure is slightly more complicated as it involves \emph{marked concatenation}. Given two languages $K$ and $L$, a marked concatenation of $K$ with $L$ is a language of the form $KaL$ for some $a \in A$. The \emph{polynomial closure} of \Cs, denoted \pol{\Cs}, is the smallest class of languages containing \Cs and closed under union, intersection and marked concatenation (\emph{i.e.}, $K\cup L$, $K\cap L$ and $KaL$ belong to $\Cs$ for $K,L \in \Cs$,~$a \in A$).

\medskip\noindent
The dot-depth hierarchy is now defined as follows:
\begin{itemize}
\item Level 0 is the class $\{\emptyset,\{\varepsilon\},A^+,A^*\}$ (where $A$ is the working alphabet).
\item Each \emph{half level}  $n+\frac{1}{2}$ is the \emph{polynomial closure} of the previous full level~$n$.
\item Each \emph{full level} $n+1$ is the \emph{Boolean closure} of the previous half level  $n+\frac{1}{2}$.
\end{itemize}

A side remark is that this definition is not the original one. First, the historical definition of the dot-depth started from another class at level~0. However, both definitions coincide at level~1 and above. Second, the polynomial closure of a class~\Cs was defined as the smallest class containing~\Cs and closed under \emph{union} and \hbox{concatenation}. This definition is seemingly weaker, as it does not explicitly insist for \pol{\Cs} to be closed under intersection. However, Arfi~\cite{arfi87,arfi91} and Pin~\cite{jep-intersectPOL} showed that the two definitions are equivalent, provided that \Cs satisfies some mild closure~properties.

\medskip

The union of all levels in the dot-depth hierarchy is the whole class of star-free languages. Moreover, Brzozowski and Knast proved in 1978 that the dot-depth hierarchy is strict: any level contains strictly more languages than the previous~one.

\begin{theorem}[Brzozowski \& Knast~\cite{BroKnaStrict}]\label{thm:hintro:ddstrict}
  The dot-depth hierarchy is strict when the alphabet contains at least two letters.
\end{theorem}

This shows in particular that
classes built using Boolean and polynomial closures do not satisfy the same closure properties, in general. Typically, when \Cs is a class of languages, \pol{\Cs} is closed under marked concatenation but \textbf{\emph{not}} under complement, while \bool{\Cs} is closed under complement but {\bf\emph{not}} under marked concatenation. The fact that the hierarchy is strict motivates the investigation of the membership problem for all levels.

\begin{problem}[Membership for the dot-depth hierarchy]
  Given some level in the dot-depth hierarchy, is membership decidable for this level?
\end{problem}

Using the framework developed by Schützenberger in his proof for deciding whether a language is star-free, Knast proved in 1983 that level 1 enjoys decidable membership, via an intricate proof from the combinatorial point of view.
\begin{theorem}[Knast~\cite{knast83}]\label{thm:knast}
  Level 1 in the dot-depth has decidable membership.
\end{theorem}

The case of half levels required to adapt Schützenberger's approach, which was designed to deal with Boolean algebras only (recall that half levels are \emph{not} Boolean algebras, otherwise the hierarchy would collapse). In 1995, Pin~\cite{pinordered} modified the framework to handle half levels. Membership was then solved for level~$\frac12$ by Pin and Weil in 2002, as well as for level $\frac32$ by Gla\ss er and Schmitz in 2007.
\begin{theorem}[Pin \& Weil~\cite{pwdelta,pwdelta2,PinWeilVD}, Gla\ss er \& Schmitz~\cite{glasserdd}]\label{thm:dd32}
  Levels~$\frac12$ and $\frac32$ in the dot-depth hierarchy have decidable membership.
\end{theorem}

One may now wonder why level~0 in the dot-depth hierarchy is $\{\emptyset,\{\varepsilon\},A^+,A^*\}$. It would be natural to start from $\{\emptyset,A^*\}$, and to apply the very same construction for higher levels. This is exactly the definition of the Straubing-Thérien hierarchy, introduced independently in 1981 by Straubing~\cite{StrauConcat} and Thérien~\cite{TheConcat}. Its definition follows the same scheme as that of the dot-depth, except that level~0 is $\{\emptyset,A^*\}$.

Like the dot-depth, the Straubing-Thérien hierarchy is strict and spans the whole class of star-free languages. One can show this by proving that level $n$ in the dot-depth hierarchy sits between levels $n$ and $n+1$ in the Straubing-Thérien hierarchy. This makes membership a relevant problem for each level in this hierarchy as well.

\begin{problem}[Membership for the Straubing-Thérien hierarchy]
  Given some level in the Straubing-Thérien hierarchy, is membership decidable for this level?
\end{problem}

Just as for the dot-depth hierarchy, level 1 in the Straubing-Thérien hierarchy was shown to be decidable by Simon in 1972 (actually before the formal definition of the hierarchy itself). The first half levels were solved in 1987 by~Arfi who relied, for level $\frac32$, on a difficult result of Hashiguchi~\cite{Hashiguchi:1983}. In 1995, Pin and Weil presented a self-contained proof using the adaptation~\cite{pinordered} of the framework of Schützenberger to classes that are not closed under complement.
\begin{theorem}[Simon~\cite{simonphd,simonthm}]\label{thm:st1}
  Level 1 in the Straubing-Thérien hierarchy has decidable membership.
\end{theorem}

\begin{theorem}[Arfi~\cite{arfi87,arfi91}, Pin \& Weil~\cite{pwdelta,pwdelta2}]\label{thm:st32}
  Levels $\frac12$ and $\frac32$ in the Straubing-Thérien hierarchy have decidable membership.
\end{theorem}

In fact, the dot-depth and the Straubing-Thérien hierarchies are closely related. First, as already stated, they are interleaved. More importantly, Straubing proved in 1985 an effective reduction between the membership problems associated to their full levels, which Pin and Weil adapted to half levels in 2002.

\begin{theorem}[Straubing~\cite{StrauVD}, Pin \& Weil~\cite{PinWeilVD}]\label{STred}\label{thm:ddredst}
  Membership for a level in the dot-depth reduces to membership for the same level in the Straubing-Thérien~\hbox{hierarchy}.
\end{theorem}

This theorem is crucial. Indeed, from a combinatorial point of view, the Straubing-Thérien hierarchy is much simpler to deal with than the dot-depth. This is evidenced by all recent publications on the topic: most results for the dot-depth are obtained indirectly as corollaries of results for the Straubing-Thérien hierarchy via Theorem~\ref{thm:ddredst}. This is the case for the last results about membership that we state, which date back to 2014--2015, and conclude the state of the art about membership for both hierarchies.

\begin{theorem}[Place \& Zeitoun~\cite{pzqalt,pzboolpol}, Place~\cite{pseps3}]\label{thm:5272}
  Membership is decidable for levels 2, $\frac52$ and $\frac72$ in both the dot-depth and the Straubing-Thérien hierarchies.
\end{theorem}

Note that there is a gap between levels $\frac52$ and $\frac72$: it is unknown whether level~3 has decidable membership. This is because full levels are actually harder to cope with than half levels. Indeed, the framework that was developed recently to solve membership problems relies on the more general separation problem (and in fact, on an even more general problem called covering~\cite{pzcovering}). Furthermore, closure under concatenation product---which holds for half but not for full levels---is an essential ingredient in the methodology elaborated for solving separation and covering.

\medskip

Now that we have surveyed the most prominent results regarding membership for two concatenation hierarchies, let us explain an extra but important motivation for investigating this problem. Recall that the initial incentive was to understand the interplay between Boolean operations and concatenation, two operations at the heart of language theory. But additionally, Thomas discovered in 1982 a tight connection between the dot-depth and first-order logic, which suffices by itself to motivate an in-depth investigation of this hierarchy. The core idea is the following: since star-free languages are exactly those that one can define in first-order logic, it is desirable to refine this correspondence level by level, in each of the hierarchies considered so far. The beautiful result of Thomas establishes such a correspondence.

\smallskip
To present it, we first slightly extend the standard signature used in first-order logic over words. In addition to the linear order and the label predicates, we add:
\begin{itemize}
\item The binary \emph{successor} ``$+1$'', interpreted as the successor between positions.
\item The unary \emph{minimum} ``$\min$'', that selects the leftmost position of the word.
\item The unary \emph{maximum} ``$\max$'', that selects the rightmost position of the word.
\item The nullary \emph{empty} ``$\varepsilon$'' predicate, which holds for the empty word only.
\end{itemize}
We denote by \fows the resulting logic. Since these predicates are all definable in \fow, adding them in the signature does not increase the overall expressive power of first-order logic. In other words, \fow and \fows are equally expressive. However, this enriched signature makes it possible to define fragments of first-order logic that correspond to levels of the dot-depth hierarchy.

To this end, we classify $\fows$ sentences by counting their number of quantifier \emph{alternations}. Given a natural number $n$, a sentence is said to be  ``$\sic{n}(\sigen)$'' (resp.\ ``$\pic{n}(\sigen)$'') when it is a formula from \fows whose prenex normal form has either:
\begin{itemize}
\item \emph{Exactly} $n$ blocks of quantifiers, the leftmost being an~``$\exists$'' (resp.\ a~``$\forall$'') block, or
\item \emph{Strictly less} than $n$ blocks of quantifiers.
\end{itemize}
For example, a formula of \fows whose prenex normal form is:
\[
  \exists x_1 \exists x_2\; \forall x_3\; \exists x_4
  \ \varphi(x_1,x_2,x_3,x_4) \quad \text{(with $\varphi$ quantifier-free)},
\]
\noindent
is $\sict$. Observe that while $\fow$ and $\fows$ have the same expressiveness, the enriched signature
increases the expressive power of individual levels.

The negation of a $\sic{n}(\sigen)$ sentence is not a $\sic{n}(\sigen)$ sentence in general (it is a $\pic{n}(\sigen)$ sentence). Thus, the corresponding classes of languages are not closed under complement, which makes it meaningful to define $\bsc{n}(\sigen)$ sentences as finite Boolean combinations of \hbox{$\sic{n}(\sigen)$} and $\pic{n}(\sigen)$ sentences. This yields a strict hierarchy of classes of languages depicted in \figurename~\ref{fig:hiera}, where, slightly abusing notation, each level denotes the class of languages defined by the corresponding set of formulas.

\tikzstyle{non}=[inner sep=1pt]
\tikzstyle{tag}=[draw,fill=white,sloped,circle,inner sep=1pt]
\begin{figure}[!htb]
  \centering
  \begin{tikzpicture}

    \node[non] (b0) at (-0.2,0.0) {$\sic{0} = \pic{0} = \bsc{0}$ };

    \node[non] (s1) at (1.0,-0.8) {\sicu};
    \node[non] (p1) at (1.0,0.8) {\picu};
    \node[non] (b1) at (2.2,0.0) {\bscu};

    \node[non] (s2) at ($(b1)+(1.2,-0.8)$) {\sicd};
    \node[non] (p2) at ($(b1)+(1.2,0.8)$) {\picd};
    \node[non] (b2) at ($(b1)+(2.4,0.0)$) {\bscd};

    \node[non] (s3) at ($(b2)+(1.2,-0.8)$) {\sict};
    \node[non] (p3) at ($(b2)+(1.2,0.8)$) {\pict};

    \draw[thick] (b0.-60) to [out=-90,in=180] node[tag] {\scriptsize $\subsetneq$} (s1.west);
    \draw[thick] (b0.60) to [out=90,in=-180] node[tag] {\scriptsize $\subsetneq$} (p1.west);

    \draw[thick] (s1.east) to [out=0,in=-90] node[tag] {\scriptsize $\subsetneq$} (b1.-120);
    \draw[thick] (p1.east) to [out=0,in=90] node[tag] {\scriptsize $\subsetneq$} (b1.120);

    \draw[thick] (b1.-60) to [out=-90,in=180] node[tag] {\scriptsize $\subsetneq$} (s2.west);
    \draw[thick] (b1.60) to [out=90,in=-180] node[tag] {\scriptsize $\subsetneq$} (p2.west);
    \draw[thick] (s2.east) to [out=0,in=-90] node[tag] {\scriptsize $\subsetneq$} (b2.-120);
    \draw[thick] (p2.east) to [out=0,in=90] node[tag] {\scriptsize $\subsetneq$} (b2.120);

    \draw[thick] (b2.-60) to [out=-90,in=180] node[tag] {\scriptsize $\subsetneq$} (s3.west);
    \draw[thick] (b2.60) to [out=90,in=-180] node[tag] {\scriptsize $\subsetneq$} (p3.west);

    \draw[thick,dotted] ($(s3.east)+(0.1,0.0)$) to ($(s3.east)+(0.6,0.0)$);
    \draw[thick,dotted] ($(p3.east)+(0.1,0.0)$) to ($(p3.east)+(0.6,0.0)$);

  \end{tikzpicture}
  \caption{Quantifier Alternation Hierarchy}\label{fig:hiera}
\end{figure}
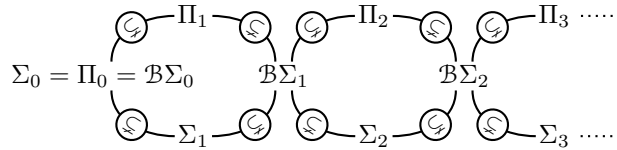
The correspondence discovered by Thomas relates levels of the dot-depth hierarchy with levels in the quantifier alternation hierarchy of enriched first-order logic.

\begin{theorem}[Thomas~\cite{ThomEqu}]\label{thm:thomas-citethomequ}
  For any alphabet $A$, any $n \in \nat$ and any language $L \subseteq A^*$, the two following properties hold:
  \begin{enumerate}
  \item $L$ has dot-depth $n$ iff $L$ belongs to $\bsc{n}(\sigen)$.
  \item $L$ has dot-depth $n + \frac{1}{2}$ iff $L$ belongs to $\sic{n+1}(\sigen)$.
  \end{enumerate}
\end{theorem}

Some years later in 1986, a similar correspondence was established between levels in the Straubing-Thérien hierarchy and in the quantifier alternation hierarchy over the signature consisting of the linear order and the label predicates. Such levels, denoted by $\bsc{n}(<)$ and $\sic{n}(<)$, are defined analogously as for the enriched~signature.

\begin{theorem}[Perrin \& Pin~\cite{PPOrder}]\label{thm:perr-pin-citepp}
  For any alphabet $A$, any $n \in \nat$ and any language $L \subseteq A^*$, the two following properties hold:
  \begin{enumerate}
  \item $L$ has level $n$ in the Straubing-Thérien hierarchy iff $L$ belongs to $\bsc{n}(<)$.
  \item $L$ has level $n + \frac{1}{2}$ in the Straubing-Thérien hierarchy iff $L$ belongs to $\sic{n+1}(<)$.
  \end{enumerate}
\end{theorem}

\smallskip\noindent\textbf{Contributions.} The line of research that we just surveyed spans over~45 years. This explains why results are scattered in the literature, and often tailored to one or the other of the two hierarchies. Moreover, their proofs often rely on involved tools, such as algebraic or topological ones, and sometimes use other hard results as black boxes. In this paper, we present a unified framework that capture them all. In the following, we call \emph{simple} a concatenation hierarchy whose level~0 is a \emph{finite} Boolean algebra closed under quotient (see Section~\ref{sec:preliminaries}). We present 6 theorems that suffice to recover \emph{\textbf{all}} known results that we presented so far, providing new proofs for 3 of them:
\begin{enumerate}
\item We give a new proof that the polynomial closure of a lattice of regular languages closed under quotient is also closed under intersection.
\item We generalize Theorem~\ref{thm:hintro:ddstrict}: \emph{any} simple hierarchy is \emph{strict}.
\item We state that levels $\frac12$, 1 and $\frac32$ of \emph{any} simple hierarchy have decidable separation, hence also decidable membership.
\item We state the following transfer result (even for non-finitely based hierarchies): if level $n-\frac12$ has decidable separation, then level $n+\frac12$ has decidable membership.
\item We generalize Theorem~\ref{thm:ddredst} to separation, with a language-theoretic formulation.
\item We generalize Theorems~\ref{thm:thomas-citethomequ} and \ref{thm:perr-pin-citepp} to \emph{any} hierarchy, by showing that one can describe any concatenation hierarchy by an associated logical fragment.
\end{enumerate}

We provide new proofs for Items 1, 2 and 6 (Theorems~\ref{thm:hintro:polc}, \ref{thm:hintro:strict} and \ref{thm:qalt:maintheo}). For Items 3, 4 and 5 (Theorems~\ref{thm:sep:hiera}, \ref{thm:sep:transfer} and \ref{thm:enrichment}), see \cite{pzcovering,pzsucc,pzboolpol} or the full papers~\cite{pzcoveringfull,pzsuccfull,pzqaltfull}.

\medskip\noindent\textbf{Organization.} In Section~\ref{sec:preliminaries}, we set up the notation and introduce a tool which will be useful for proving Item~2 above: stratifications of a class of languages (it is actually also important for proving Item~4, see~\cite{pzboolpol}). We define Boolean and polynomial closures in Section~\ref{sec:polyn-bool-clos}, where we also prove that closure under intersection for polynomial closure is implied from closure under union and marked concatenation, if the class we start from is closed under mild properties. In Section~\ref{chap:hieraintro}, we define generic concatenation hierarchies and we state their basic properties. In Section~\ref{sec:hintro:strictness}, we prove that concatenation hierarchies with a finite basis are strict. Section~\ref{sec:two-fund-conc} investigates the two historical hierarchies: the dot-depth and the Straubing-Thérien hierarchies. Finally, Section~\ref{sec:link-with-logic} presents the generic logical definition of concatenation hierarchies, thus generalizing the result of Thomas~\cite{ThomEqu}. This paper is the full version of~\cite{pzcsr17}.

\section{Preliminary definitions and tools}\label{sec:preliminaries}
In this section, we set up the definitions and the notation. We start with the classical notions of words, languages and classes of languages. We then present the two problems we are interested in: ``membership'' and ``separation''. At last, we introduce the notion of \emph{stratification} of a class of languages.

\medskip\noindent\textbf{Finite words and classes of languages.}
Throughout the paper, we fix a finite alphabet~$A$. We let $\varepsilon$ be the empty word. The set of all finite words over $A$ is denoted by $A^*$, and the set $A^{*}\setminus\{\varepsilon\}$ of all nonempty words over~$A$ is denoted by $A^+$.  If $w$ is a word, we denote its length by $|w|$, that is, its number of letters. If $|w|=n$, then $w=a_1\cdots a_n$ with $a_i\in A$, and the set of \emph{positions} of $w$ is $\{1,\ldots, n\}$. Moreover, for two positions $i,j$ we define $w]i,j[$ as the word $a_{i+1}\cdots a_{j-1}$ if $i+1\leq j-1$, and as $\varepsilon$ if $i+1> j-1$. We define similarly $w[i,j[$ as $a_i\ldots a_{j-1}$ if $i\leq j-1$ and as $\varepsilon$ otherwise. We define $w]i,j]$ symmetrically. Given $w \in A^*$, we let $\cont{w}$ be the set of letters appearing in $w$, that is, the smallest set $B\subseteq A$ such that $w\in B^*$. We say that $\cont{w}$ is the \emph{alphabet} of $w$. A \emph{language (over $A$)} is a subset of $A^*$. Finally, a \emph{class of languages} is a set of languages over $A$.

\begin{remark}
  Our definition of a class of languages is simpler than the usual one. When dealing with several alphabets, a class of languages is often defined as a \emph{function} mapping any finite alphabet $A$ to a set of languages $\Cs(A)$ over~$A$. Our simpler definition is justified by the fact that we mainly use one fixed alphabet in the~paper.
\end{remark}

\noindent
There are several fundamental operations on languages that we shall consider:
\begin{itemize}
\item Boolean operations (union, intersection and complement),
\item Left and right quotients. If $w$ is a word and $L$ is a language, then the left (resp.\ right) quotient $w^{-1}L$ (resp.\ $Lw^{-1}$) of $L$ by $w$ is the following language:
  \[
    w^{-1}L \stackrel{\text{def}}{=}\{v\in A^*\mid wv\in L\}, \qquad     Lw^{-1} \stackrel{\text{def}}{=}\{v\in A^*\mid vw\in L\}.
  \]
  Note that if $a$ is a letter and $w$ is a word, then for every language $L$, we have $(wa)^{-1}L=a^{-1}(w^{-1}L)$ and $L(aw)^{-1}=(Lw^{-1})a^{-1}$. Therefore, a class is closed under taking quotients iff it is closed under taking quotients by any \emph{letter}. We shall freely use this fact throughout the paper. Another basic fact that we shall use without further reference is that quotients commute with Boolean operations. For instance, $w^{-1}(K\cup L)=w^{-1}K\cup w^{-1}L$ and $w^{-1}(A^*\setminus L)=A^*\setminus(w^{-1}L)$.
\item The concatenation of two languages $K,L\subseteq A^*$ is defined as:
  \[
    KL \stackrel{\text{def}}{=} \{uv\mid u\in K\text{ and }v\in L\}.
  \]

\item Finally, given a letter $a\in A$, the \emph{marked concatenation} of $K$ and $L$ by $a$ is the language $K\{a\}L$, also written $KaL$.
\end{itemize}

\smallskip\noindent
All classes considered in this paper satisfy robust properties, which we present~now.
\begin{itemize}
\item A class of languages is a \emph{lattice} if it contains $\emptyset$ and $A^*$ and it is closed under union and intersection.
\item A \emph{Boolean algebra} is a lattice closed under complement.
\item A class of languages is \emph{quotienting} when it is closed under taking (left and right) quotients.
\end{itemize}

\begin{example}\label{ex:at}
  Let \at be the class of languages over $A$ consisting of all finite Boolean combinations of languages $A^*aA^*$, for $a\in A$. The name ``\at'' stands for ``alphabet testable'': a language belongs to \at when membership of a word in this language depends only on the set of letters occurring in the word. It is straightforward to verify that \at is a \emph{finite} \vari.\qed
\end{example}

We denote by \reg the class of all regular languages over $A$. All classes that we consider consist of regular languages only, \emph{i.e.}, are sub-classes of \reg. Recall that regular languages can be equivalently defined by nondeterministic finite automata, regular expressions, finite monoids or monadic second-order logic. In this paper, we assume a basic knowledge in automata theory, but we shall rather rely on the following characterization of regular languages, due to Myhill and~Nerode.
\begin{theorem}[Myhill and Nerode~\cite{nerode58}]\label{thm:auto:nerode}
  Let $L\subseteq A^*$ be a language. The following properties are equivalent:
  \begin{enumerate}
  \item $L$ is regular,
  \item $L$ has finitely many left quotients,
  \item $L$ has finitely many right quotients.
  \end{enumerate}
\end{theorem}

\medskip\noindent\textbf{Membership and separation.}
A class of regular languages is usually given by a syntax. For such a class~\Cs of languages, the most basic question is whether one can test membership of an input regular language in the class \Cs. In other words, we want to design an algorithm deciding whether an input language admits a description in the given syntax, or to prove that no such algorithm exists. The corresponding decision problem is called \emph{\Cs-membership} (or membership for \Cs).

\medskip
\noindent\textbf{Membership problem for \Cs:}

\smallskip\noindent
\begin{tabular}{ll}
  {\bf Input:}    & A regular language $L$.\\
  {\bf Question:} & Does $L$ belong to \Cs?
\end{tabular}

\smallskip
Recent solutions to the membership problem for specific classes actually consider a more general problem, the \emph{\Cs-separation problem} (or separation problem for \Cs), which is stated as follows:

\medskip
\noindent\textbf{Separation problem for \Cs:}

\smallskip\noindent
\begin{tabular}{ll}
  {\bf Input:}    & Two regular languages $L_1,L_2$.\\
  {\bf Question:} & Is there a language $K$ from \Cs such that $L_1\subseteq K$ and $K\cap L_2=\emptyset$?
\end{tabular}

\smallskip
We say that a language $K$ such that $L_1\subseteq K$ and $K\cap L_2=\emptyset$ is a \emph{separator} of $(L_1,L_2)$. Observe that since regular languages are closed under complement, there is a straightforward reduction from membership to separation. Indeed, an input language $L$ belongs to \Cs when it can be \Cs-separated from its complement.

\medskip\noindent\textbf{Finite lattices and canonical preorders.}
In this section, we present simple mathematical tools associated to any \emph{finite} class. Of course, most finite classes are not very interesting. The only example which is featured prominently in this paper is the class \at of alphabet testable languages. However, some of the decision problems can be lifted from finite to infinite classes thanks to stratifications.

\smallskip\noindent\textbf{Canonical preorders for finite lattices.} We fix an arbitrary finite lattice \Cs, to which we associate a \emph{canonical preorder relation over $A^*$} defined as follows. Given $w,w' \in A^*$, we write $w \leq_\Cs w'$ if and only if,
\[
  \text{For all $L \in \Cs$,} \quad w \in L \ \Rightarrow\ w' \in L.
\]
It is immediate from the definition that $\leq_\Cs$ is indeed transitive and reflexive.

We now present the applications of the relation $\leq_\Cs$. We start with an important lemma, which relies on the fact that \Cs is finite. We say that a language $L \subseteq A^*$ is an \emph{upper set} for $\leq_\Cs$ when for any two words $u,v \in A^*$, if $u \leq_\Cs v$ and $u \in L$, then $v \in L$. In other words, $L$ is an upper set when it coincides with its upwards closure $\upset L$ with respect to the preorder $\leq_{\Cs}$, where $\upset L$ is defined as the set of words that are above some word of $L$:
\[
  \upset L = \{u\in A^{*}\mid \exists w \in L,\ w\leq_{\Cs}u\}.
\]

\begin{lemma}\label{lem:metho:upper}
  Let $\Cs$ be a finite lattice. Then, for any word $w\in A^{*}$, we have:
  \[
    \upset w = \bigcap_{L\in\Cs\text{ and }w\in L}L.
  \]
  In particular, the canonical preorder $\leq_\Cs$ has finitely many upper sets.
\end{lemma}

\begin{proof}
  The equality follows directly from the definition of $\leq_{\Cs}$. We prove the second assertion. By definition, an upper set is a union of languages of the form $\upset w$. Hence, it suffices to prove that there are finitely many languages $\upset w$, which follows immediately from the equality of the lemma and the finiteness of \Cs.
\end{proof}

\medskip

We now prove the second important property of the preorder $\leq_\Cs$: we use the fact that \Cs is a lattice to characterize the languages belonging to \Cs. They are exactly the upper sets for $\leq_\Cs$.

\begin{lemma}\label{lem:metho:satur}
  Let $\Cs$ be a finite lattice of languages. Then,
  for any $L \subseteq A^*$, we have $L \in \Cs$ iff $L$ is an upper set for~$\leq_\Cs$.
\end{lemma}

\begin{proof}
  Assume first that $L \in \Cs$. Then for all $w \in L$ and all $w'$ such that $w \leq_\Cs w'$, we have $w' \in L$ by definition of $\leq_\Cs$. Hence, $L$ is an upper set.

  Assume now that $L$ is an upper set. Observe that since \Cs is finite and closed under intersection, for any word $w$, the upper set $\upset w$ belongs to \Cs by Lemma~\ref{lem:metho:upper}: it is the intersection of all languages in \Cs containing~$w$. Furthermore,
  $L = \bigcup_{w \in L} \upset w$.
  By Lemma~\ref{lem:metho:upper}, there are only finitely many sets of the form~$\upset w$. Since \Cs is closed under finite union, $L$ belongs to \Cs.
\end{proof}

While Lemma~\ref{lem:metho:satur} states an equivalence, we mainly use the left to right implication (or rather its contrapositive). It is useful for proving that a given language $L$ does not belong to \Cs, or that two languages $K,L$ are not \Cs-separable. We describe this application in the following corollary.

\begin{corollary}\label{cor:metho:satur}
  Let $\Cs$ be a finite lattice $K,L \subseteq A^*$ be two languages. Then, the following properties hold:
  \begin{enumerate}
  \item $L$ does {\bf not} belong to \Cs iff there exist $w \in L$ and $w' \not\in L$ such that $w \leq_\Cs w'$.
  \item $L$ is {\bf not} \Cs-separable from $K$ iff there exist $w \in L$ and $w' \in K$ such that $w \leq_\Cs w'$.
  \end{enumerate}
\end{corollary}

\begin{proof}
  The first item is the contrapositive of Item~1 in Lemma~\ref{lem:metho:satur}. We prove the second. Assume that there exist $w \in L$ and $w' \in K$ such that $w \leq_\Cs w'$. Hence, for any language $H$ separating $L$ from $K$, we have $w \in H$ and $w' \not\in H$. Since $w \leq_\Cs w'$, this means that $H \not\in \Cs$ by definition of $\leq_\Cs$. Hence, $L$ is not \Cs-separable from $K$.

  Conversely, assume $L$ is not \Cs-separable from $K$. For any $w \in A^*$, note that $\upset w\in \Cs$ by Lemma~\ref{lem:metho:satur}, since $\upset w$ is an upper set for $\leq_\Cs$. We define
  \[H = \bigcup_{w \in L} \upset w.\]
  This union is finite since $\leq_\Cs$ has finitely many upper sets. Hence, $H \in \Cs$. Moreover, $L \subseteq H$ by definition. Hence, since $L$ is not \Cs-separable from $K$, we have $H \cap K \neq \emptyset$. Let $w' \in H \cap K$. By definition $w' \in \upset w$ for some $w \in L$  which means that $w \leq_\Cs w'$. Thus, we have $w \in L$ and $w' \in K$ such that $w \leq_\Cs w'$.
\end{proof}

\begin{example}
  Let $L = a^*b^*$. Then $L \not\in \at$. Indeed, we have $ab \in L$ and $ba \not\in L$ while $\cont{ab} = \cont{ba}$. For the same reason, $a^+b^+$ is not \at-separable from $b^+a^+$.
\end{example}

\smallskip\noindent\textbf{Canonical preorders for \pvaris.} We now present additional properties of the canonical $\leq_\Cs$ preorder that hold when the finite lattice \Cs is closed under quotients. The key property is given in the following lemma: closure under quotients for \Cs corresponds to compatibility with word concatenation for~$\leq_\Cs$.

\begin{lemma}\label{lem:metho:quotients}
  A finite lattice~\Cs is closed under quotient if and only if its associated canonical preorder $\leq_\Cs$ is a precongruence for word concatenation. That is, for any words $u,v,u',v'$,
  \[
    u \leq_\Cs u' \quad \text{and} \quad v \leq_\Cs v' \quad \Rightarrow \quad uv \leq_\Cs u'v'.
  \]
\end{lemma}

\begin{proof}
  We do the proof for lattices (the result for Boolean algebras is an immediate consequence). First assume that \Cs is closed under quotients and let $u,u',v,v'$ be four words such that $u \leq_\Cs u'$ and $v \leq_\Cs v'$. We have to prove that $uv \leq_\Cs u'v'$. Let $L \in \Cs$ and assume that $uv \in L$.
  This means that $v \in u^{-1} \cdot L$. By closure under left quotient, we have $u^{-1}L \in \Cs$, hence, since $v \leq_\Cs v'$, we obtain that $v'\in u^{-1} \cdot L$ and therefore that $uv' \in L$. It now follows that $u \in L(v')^{-1}$. Using closure under right quotient, we obtain that $L(v')^{-1} \in \Cs$. Therefore, since $u \leq_\Cs u'$, we conclude that $u' \in L(v')^{-1}$ which means that $u'v' \in L$, as desired.

  Conversely, assume that $\leq_\Cs$ is a precongruence.	Let $L \in \Cs$ and $w \in A^*$, we prove that $w^{-1}L \in \Cs$ (the proof for right quotients is symmetrical). By Lemma~\ref{lem:metho:satur}, we have to prove that $w^{-1}L$ is an upper set. Let $u \in w^{-1}L$ and $u' \in A^*$ such that $u \leq_\Cs u'$. Since $\leq_\Cs$ is a precongruence, we have $wu \leq_\Cs wu'$. Hence, since $L$ is an upper set (it belongs to \Cs) and $wu \in L$, we have $wu' \in L$. We conclude that $u' \in w^{-1}L$, which terminates the proof.
\end{proof}

\noindent We finish with a useful consequence of Lemma~\ref{lem:metho:quotients}: finite \pvaris contain only regular languages.

\begin{lemma}\label{lem:metho:omegapower}
  Let \Cs be a finite \pvari. Then any language in \Cs is regular. Moreover, there exists a natural number $p \geq 1$ such that for any word $u \in A^*$ and natural numbers $m,m' \geq 1$, we have $u^{pm} \leq_\Cs u^{pm'}$.
\end{lemma}

\begin{proof}
  Let $\sim$ be the equivalence generated by $\leq_\Cs$. That is, for any $w,w' \in A^*$,
  \[
    w \sim w' \Longleftrightarrow w \leq_\Cs w' \text{ and } w' \leq_\Cs w.
  \]
  By Lemma~\ref{lem:metho:upper}, $\leq_\Cs$ has finitely upper sets and by Lemma~\ref{lem:metho:quotients}, it is a precongruence for concatenation. Therefore, $\sim$ is a congruence of finite index for concatenation. It follows that the quotient set ${A^*}/{\sim}$ is a finite monoid and that the morphism,
  \[
    \begin{array}{llll}
      \alpha:&  A^* & \to & {A^*}/{\sim} \\
             &  w   & \mapsto & [w]_{\sim}
    \end{array}
  \]
  which maps every word $w$ to its equivalence class $[w]_{\sim}$ recognizes any upper set $L$ with respect to $\leq_\Cs$. Indeed, $\alpha(u)=\alpha(v)$ implies $u\leq_{\Cs}v$, and therefore, if $u\in L$, then also $v\in L$ (since $L$ is an upper set). It follows from Lemma~\ref{lem:metho:satur} that $\alpha$ recognizes all languages in $\Cs$, whence all these languages are regular.

  For the second part of the lemma, it suffices to observe that since ${A^*}/{\sim}$ is a finite monoid, it has an idempotent power $\omega$ (\emph{i.e.}, $s^\omega=s^{2\omega}$ for all $s\in {A^*}/{\sim}$). It suffices to choose $p$ as this power.
\end{proof}

In particular, for any finite \pvari \Cs, we will call the minimal number $p \geq 1$ which satisfies the statement of Lemma~\ref{lem:metho:omegapower}, the \emph{period~of\/~$\Cs$}.

\medskip\noindent\textbf{Stratifications.}
We now turn to stratifications. We start with the definition and explain how to use them for lifting the methodology outlined in the previous section to infinite classes.

Let \Cs be a class of languages. A \emph{stratification of \Cs} is an infinite sequence of \emph{finite} classes $\Cs_0,\dots,\Cs_k,\dots$ that we call the \emph{strata}, which satisfy the following properties:
\[
  \Cs_{k} \subseteq \Cs_{k+1} \text{ for all $k \in \nat$} \quad \text{and} \quad \Cs = \bigcup_{k \in \nat} \Cs_k.
\]
When \Cs is infinite (which is the only case when there is a point to stratifying \Cs), it admits infinitely many stratifications. Of course, not all of them are relevant. The standard approach is to consider stratifications that are tied to a particular syntax which may be used to define languages in \Cs. An example is to stratify the regular languages by considering the size of regular expressions (the languages of level $k$ being those that are defined by a regular expression which is made of $k$ or less symbols).

Since we only work with classes that are either \varis or \pvaris, it is important that the strata are such classes as well. An nice consequence of the fact that we deal with classes of regular languages is that such a stratification always exists.

\begin{proposition}\label{prop:metho:alwaystrat}
  Let \Cs be a class of languages. The two following properties hold:
  \begin{enumerate}
  \item \Cs is a \pvari of regular languages if and only if there exists a stratification of \Cs in which each stratum is a \pvari.
  \item \Cs is a \vari of regular languages if and only if there exists a stratification of \Cs in which each stratum is a \vari.
  \end{enumerate}
\end{proposition}

\begin{proof}
  We only prove the first item, the proof of the second one is similar. Let us first assume that there exists a stratification of \Cs in which each stratum is a \pvari. It is immediate that \Cs is a \pvari as well since it is the union of all strata. Moreover, since each stratum is a finite \pvari by definition, it may only contains regular languages by Lemma~\ref{lem:metho:omegapower}. We conclude that \Cs is a \pvari of regular languages.

  Conversely, assume that \Cs is a \pvari of regular languages. For all $k \in \nat$, we let $\Ds_k$ be the class of all languages $L \in \Cs$ that are recognized by an \nfa with at most $k$ states. Clearly, all classes $\Ds_k$ are finite and $\Ds_k \subseteq \Ds_{k+1}$ for all $k \in \nat$. Moreover, since $\Cs$ contains only regular languages, every $L \in \Cs$ is recognized by some \nfa and it follows that:
  \[
    \Cs = \bigcup_{k \in \nat} \Ds_k.
  \]
  Hence, this is indeed a stratification of \Cs. However the classes $\Ds_k$ need not be \pvaris. For all $k \in \nat$, let $\Cs_k$ be the smallest \pvari containing~$\Ds_k$. We still have $\Cs_k \subseteq \Cs_{k+1}$ for all $k \in \nat$ and since \Cs is a \pvari, we have,
  \[
    \Cs = \bigcup_{k \in \nat} \Cs_k.
  \]
  Hence, it suffices to prove that the classes $\Cs_k$ are finite. Since quotients commute with Boolean operations, any language in $\Cs_k$ is a Boolean combination of quotients of languages in $L \in \Ds_k$. Finally, recall that by Myhill-Nerode theorem (Theorem~\ref{thm:auto:nerode}) a regular language has finitely many quotients. Hence, since $\Ds_k$ is a finite class of regular languages, $\Cs_k$ is finite as well. This terminates the proof.
\end{proof}

Proposition~\ref{prop:metho:alwaystrat} is useful when trying to prove generic results for an arbitrary \pvari \Cs. However, when working with a specific one, we always define a tailored stratification which has better properties than the generic one given by  Proposition~\ref{prop:metho:alwaystrat}. Finding natural stratifications is motivated by two objectives. First, they are used as classifications of the languages belonging to the class: the lowest level that contains a particular language serves as a complexity measure for this language. Second, we use stratifications to lift some results from finite to infinite~classes.

The main idea is that proving a property for \Cs can be reduced to proving that it is satisfied by all strata $\Cs_k$, which are finite. For example, for showing that a given language $L$ does not belong to $\Cs$, it suffices to prove that it does not belong to $\Cs_k$ for all $k \in \nat$.

We finish with an overview of the notions that we have introduced and outline the general approach that we will take when using them with specific classes. Assume that an (infinite) \pvari \Cs is fixed. We will always begin by presenting a stratification of \Cs.
\begin{itemize}
\item When \Cs is a \pvari, we want all strata $\Cs_k$ to be \pvaris as well.
\item When \Cs is a \vari, we want all strata $\Cs_k$ to be \varis as well.
\end{itemize}

\begin{example}\label{ex:metho:att}
  Consider the class of alphabet threshold testable languages (denoted by \att), which consists of all finite Boolean combinations of languages of the form $(A^*aA^*)^d$ for $a \in A$ and $d \geq 1$. One may observe that \att is a \vari. Obtaining a natural stratification of \att is simple: for any $d \geq 1$, it suffices to define \datt as the class of all finite Boolean combinations of languages of the form $(A^*aA^*)^{d'}$ for some $a \in A$ and $d' \leq d$. Clearly all classes \datt are finite. Furthermore, $\datt \subseteq \patt{(d+1)}$ for all $d \geq 1$, and $\att$ is the union of all classes \datt. Thus, this is indeed a stratification of \att.
\end{example}

Once we have our stratification $\Cs_0,\dots,\Cs_k,\dots$ of \Cs into \pvaris in hand, our next move is to consider the canonical preorder relations associated to each stratum. For this overview, we denote by $\leq_k$ the preorder associated to $\Cs_k$ for all $k \in \nat$. By definition, $w \leq_k w'$ if and only if,
\[
  \text{For all $L \in \Cs_k$} \quad w \in L \ \Rightarrow\ w' \in L.
\]

Note that these relations are connected. Indeed, by definition of stratifications, we have $\Cs_k \subseteq \Cs_{k+1}$ for all $k \geq 1$. Hence, the preorder $\leq_{k+1}$ refines the preorder $\leq_k$.

\begin{fact}\label{fct:metho:finer}
  For any $k \in \nat$ and any two words $w,w'$, the following implication holds:
  \[
    w \leq_{k+1} w'\ \Rightarrow\ w \leq_k w'.
  \]
\end{fact}

Moreover, since we choose all strata to be \pvaris, these preorder relations satisfy all generic properties proved above, which we summarize in the following lemma.

\begin{lemma}\label{lem:metho:stratlem}
  For any $k \in \nat$, the relation $\leq_k$ is a preorder over $A^*$ which has finitely many upper sets and is compatible with the concatenation operation. Moreover, for any language $L \subseteq A^*$,
  \begin{enumerate}
  \item $L \in \Cs_k$ if and only if $L$ is an upper set for $\leq_k$.
  \item $L \in \Cs$ if and only if there exists $k \in \nat$ such that $L$ is an upper set for $\leq_k$.
  \end{enumerate}
\end{lemma}

\begin{proof}
  Immediate from Lemma~\ref{lem:metho:satur}, Lemma~\ref{lem:metho:quotients} and the fact that $\Cs = \bigcup_{k \in \nat} \Cs_k$.
\end{proof}

Finally, as we already explained, the main application of these notions will be to prove negative properties. This is what we described in Corollary~\ref{cor:metho:satur}. Let us generalize the statement to stratifications.

\begin{corollary}\label{cor:metho:stratcor}
  Let $K,L \subseteq A^*$ be two languages. The two following properties hold:
  \begin{enumerate}
  \item $L$ does {\bf not} belong to \Cs iff for all $k \in \nat$ there exist $w \in L$ and $w' \not\in L$ such that $w \leq_k w'$.
  \item $L$ is {\bf not} \Cs-separable from $K$ iff for all $k \in \nat$ there exist $w \in L$ and $w' \in K$ such that $w \leq_k w'$.
  \end{enumerate}
\end{corollary}

\begin{example}
  We consider the alphabet threshold testable languages of Example~\ref{ex:metho:att}. Since the strata $\datt$ are \varis, one may verify that the associated preorders are actually equivalence relations which we denote by $\datteq$. One may also verify from the definition that given two words $w,w' \in A^*$, we have $w \datteq w'$ if and only if for each letter $a \in A$, $w$ and $w'$ contains the same amount of occurrences of $a$ up to threshold~$d$ (that is, either $w$ and $w'$ both contain more than $d$ copies of $a$, or they contain exactly the same number of copies).

  With this alternate definition, it is simple to see that for all $d \geq 1$, $ab \datteq ba$. Hence since $ab \in (ab)^*$ and $ba \not\in (ab)^*$, this proves that $(ab)^* \not\in \att$.
\end{example}

\section{Boolean and polynomial closures}\label{sec:polyn-bool-clos}
Classes of languages are often built from simpler ones by using \emph{closure operators}. Such a construction pattern permits to rank classes of languages. We are interested in two such operators: Boolean and polynomial closure. In this section, we define these operators and then describe some of their important properties.

\smallskip\noindent\textbf{Boolean closure.} Given a class of languages \Cs, the \emph{Boolean closure of \Cs}, denoted by \bool{\Cs}, is the smallest Boolean algebra containing \Cs, \emph{i.e.},  the smallest class of languages containing $\Cs$ closed under union, intersection and complement. Observe that by definition, the Boolean closure of a class is a Boolean algebra. In particular, it follows that Boolean closure is an idempotent operation: for any class \Cs, we have $\bool{\bool{\Cs}} = \bool{\Cs}$. Furthermore, Boolean closure preserves many properties of the input class~\Cs. In particular, it preserves closure under quotient.

\begin{proposition}\label{prop:hintro:boolc}
  For any \pvari of languages \Cs, the Boolean closure of \Cs is a \vari.
\end{proposition}

\begin{proof}
  Immediate since quotients commute with Boolean operations.
\end{proof}

\begin{remark}
  Not all closure properties are preserved under Boolean closure. The most significant example is closure under concatenation. In fact, all classes that we build with Boolean closure will \emph{not} be closed under concatenation. This is a problem, since our techniques for solving membership and separation rely on concatenation.
\end{remark}

\smallskip\noindent\textbf{Polynomial closure.} We turn to the second operation: \emph{polynomial closure}. Let \Cs be a class of languages. We say that a language $L \subseteq A^*$ is a \emph{\Cs-monomial} when there exists a natural number $n \in \nat$, $L_0,\cdots,L_n \in \Cs$ and $a_1,\dots,a_n \in A$ such that,
\[
  L = L_0a_1L_1a_2L_2 \cdots a_nL_n.
\]
The minimal integer $n$ for which this property holds is called the \emph{degree} of $L$. Finally, a \emph{\Cs-polynomial} is a finite union of \Cs-monomials. The degree of a \Cs-polynomial $L$ is the minimal integer $n$ such that $L$ is a union of \Cs-monomials having degree at most~$n$.

We may now define polynomial closure. For any class of languages \Cs, we call \emph{polynomial closure of \Cs} the class of all \Cs-polynomials. We denote it by $\pol{\Cs}$. Note that $\Cs \subseteq \pol{\Cs}$, since the languages in \Cs are the \Cs-monomials of degree $0$.

\begin{example}\label{ex:hintro:sig1}
  Consider $\Cs = \{\emptyset,A^*\}$. Then $\pol{\Cs}$ consists of all finite unions of languages of the form $A^*a_1A^* \cdots a_nA^*$ for $n \geq 0$ and $a_1,\dots,a_n \in A$.
\end{example}

\begin{example}\label{ex:hintro:sig2}
  Consider the class  \at of alphabet testable languages from Example~\ref{ex:at}. One may check that $\pol{\at}$ consists of all finite unions of languages of the~form,
  \[
    B_0^*a_1B_1^*a_2B_2^* \cdots a_nB_n^* \quad \text{with $a_1,\dots,a_n \in A$ and $B_0,\dots,B_n \subseteq A$}.
  \]
\end{example}

It is not immediate from the definition that classes built with polynomial closure have much structure. In fact, without any hypothesis on the input class \Cs, we are only able to prove two properties. First, \pol{\Cs} is closed under union, by definition.
Moreover, it is simple to prove that any polynomial closure is closed under \emph{marked concatenation}.

\begin{lemma}\label{lem:hintro:polmarked}
  For any class \Cs, the class \pol{\Cs} is closed under marked concatenation.
\end{lemma}

\begin{proof}
  Let $K,L$ over $A$ be two languages in \pol{\Cs}. Given $a \in A$, we prove that $KaL \in \pol{\Cs}$. By definition, there exists \Cs-monomials $K_1,\dots,K_m$ and $L_1,\dots,L_n$ such that $K = \bigcup_{i \leq m} K_i$ and $L = \bigcup_{i \leq n} L_i$. Hence, we have,
  \[
    KaL = \bigcup_{i \leq m} \bigcup_{j \leq n} K_iaL_j.
  \]
  Since $K_iaL_j$ is clearly a \Cs-monomial for all $i,j$, we conclude that $KaL \in \pol{\Cs}$.
\end{proof}

When \Cs contains the singleton language $\{\varepsilon\}$, one may prove a stronger variant of Lemma~\ref{lem:hintro:polmarked}. While simple, this observation is important as we will rely on it later.

\begin{lemma}\label{lem:hintro:epsilon}
  Let \Cs be a class of languages. Then, $\{\varepsilon\} \in \Cs$ if and only if $\{\varepsilon\} \in \pol{\Cs}$. Moreover, in that case, for any $w \in A^+$ and $K,L \in \pol{\Cs}$, we have:
  \[
    KwL \in \pol{\Cs} \quad wL \in \pol{\Cs} \quad Kw \in \pol{\Cs} \quad \{w\} \in \pol{\Cs}.
  \]
\end{lemma}

\begin{proof}
  Clearly, if $\{\varepsilon\} \in \Cs$, then $\{\varepsilon\} \in \pol{\Cs}$ since $\Cs \subseteq \pol{\Cs}$. Conversely, it suffices to observe that a \Cs-polynomial of degree $n \geq 1$ cannot be equal to $\{\varepsilon\}$. Hence, if $\{\varepsilon\} \in \pol{\Cs}$, then $\{\varepsilon\}$ is a \Cs-polynomial of degree $0$, \emph{i.e.}, an element of \Cs.

  We now assume that $\{\varepsilon\} \in \pol{\Cs}$. Consider $w \in A^+$ and $K,L \in \pol{\Cs}$. We have to prove that $KwL,wL,Kw,\{w\} \in \pol{\Cs}$. We present a proof for $KwL$ (the other cases follow by choosing $K$ or $L$ to be $\{\varepsilon\}$). Let $w = a_1 \cdots a_\ell$ with $a_1,\dots,a_\ell \in A$. It is immediate that $KwL = Ka_1\{\varepsilon\}a_2\{\varepsilon\} \cdots \{\varepsilon\}a_\ell L$.  Hence, $KwL \in \pol{\Cs}$ by closure under marked concatenation.
\end{proof}

A simple corollary of Lemma~\ref{lem:hintro:polmarked} is that there is a more elementary definition of \pol{\Cs}: it is the smallest class containing \Cs which is closed under both union and marked concatenation. Furthermore, a useful consequence of this observation is that polynomial closure is an idempotent operation: applying it to a class which is already closed under union and marked concatenation does not add any new language.

\begin{lemma}\label{lem:hintro:polidem}
  For any class of languages \Cs, we have $\pol{\pol{\Cs}} = \pol{\Cs}$.
\end{lemma}

Without any hypothesis on the input class \Cs, these are the only properties of \pol{\Cs} that one may prove. However, it was proved by Arfi~\cite{arfi87,arfi91} that when \Cs is a \pvari containing only regular languages, the polynomial closure \pol{\Cs} is a \pvari as well. The hard part is proving closure under intersection. In fact the proof of this single property depends on all properties of \pvaris of regular languages: the fact that \Cs is closed under quotient and contains only regular languages is crucial for showing that \pol{\Cs} is closed under intersection.

This result is important since it yields an alternate definition of \pol{\Cs} which is much simpler to manipulate: \pol{\Cs} is the smallest lattice containing \Cs which is closed under marked concatenation.

\begin{theorem}[Arfi~\cite{arfi87,arfi91}, Pin~\cite{jep-intersectPOL}]\label{thm:hintro:polc}t
  Let \Cs be a \pvari of regular languages. Then \pol{\Cs} is also a \pvari of regular languages. In particular, \pol{\Cs} is the smallest lattice containing \Cs and closed under marked concatenation.
\end{theorem}

\begin{proof}
  Let us fix a \pvari \Cs of regular languages. We prove that \pol{\Cs} is a \pvari of regular languages. We already know that \pol{\Cs} is closed under union be definition. Moreover, it is also clear that $\pol{\Cs}$ consists only of regular languages, since so does $\Cs$ and the class \reg is closed under both union and concatenation. Hence, we may concentrate on closure intersection and quotient. We begin with the latter, which is simple (for this, we only need \Cs to be closed under quotient itself).

  \medskip
  Given any $w \in A^*$ and any \Cs-polynomial $L$ over $A$, we have to prove that $w^{-1}L$ and $Lw^{-1}$ are \Cs-polynomials as well. We present a proof for $w^{-1}L$ (the argument for $Lw^{-1}$ is symmetrical). We may assume without loss of generality that $w = a \in A$, since $(ua)^{-1}L=a^{-1}(u^{-1}L)$. Moreover, since quotients commute with union, it suffices to consider the case when $L$ is a \Cs-monomial. Hence, there exist $L_0,\ldots,L_n \in \Cs$ and $a_1,\dots,a_n \in A$ such that,
  \[
    L = L_0a_1L_1a_2L_2 \cdots a_nL_n.
  \]
  There are two cases depending on whether $\varepsilon \in L_0$ and $a_1 = a$, or not. We have:
  \[
    a^{-1}L = \left\{
      \begin{array}{ll}
        (a^{-1}L_0)a_1L_1a_2L_2 \cdots a_nL_n \cup L_1a_2L_2 \cdots a_nL_n & \text{~if $\varepsilon \in L_0$ and $a_1 = a$,} \\
        (a^{-1}L_0)a_1L_1a_2L_2 \cdots a_nL_n & \text{~otherwise.}
      \end{array}
    \right.
  \]
  Since \Cs is closed under quotient by hypothesis, $w^{-1}L$ is a union of \Cs-monomials and is itself a \Cs-polynomial. This terminates the proof for closure under quotient.

  \medskip

  We now turn to closure under intersection, which is more involved. In particular, we use all hypotheses on~\Cs, including closure under quotient. Let $K$ and $L$ be two \Cs-polynomials and let $m,n \in \nat$ be their degrees. We have to prove that $K \cap L$ is a \Cs-polynomial as well.  Note that since intersection is distributive over union, we may assume without loss of generality that $K$ and $L$ are both \Cs-monomials. We begin by treating the special case when either $K$ or $L$ has degree $0$ using induction on the degree of the other one. We then use this special case to prove the general one.

  \medskip
  \noindent
  {\bf Special Case: Either $K$ or $L$ has degree $0$.} By symmetry, we may assume that $K$ has degree $0$ which means that $K \in \Cs$. We use induction on the degree $n$ of the \Cs-monomial $L$. If $n=0$, $K$ and $L$ both belong to \Cs, which is closed under intersection by hypothesis. Hence we conclude that $K \cap L \in \Cs \subseteq \pol{\Cs}$. Otherwise, $L$ has degree $n \geq 1$ and by definition, it can be decomposed as follows: $L = L_1bL_2$ where $L_1 \in \Cs$ and $L_2$ is a \Cs-monomial of degree at most $n-1$.

  Observe that a word $w$ belongs to $K \cap L$ when it belongs to $K$ and can be decomposed as $w = w_1bw_2$ with $w_1 \in L_1$ and $w_2 \in L_2$. Given any word $u \in A^*$, we let $Q_u$ be the set of all words $x \in A^*$ such that $u \in K(bx)^{-1}$. We claim that the following equality holds:
  \begin{equation}\label{eq:hintro:intersec}
    K \cap L = \bigcup_{u \in A^*} \left(L_1 \cap \bigcap_{x \in Q_u} K(bx)^{-1}\right) \cdot b \cdot (L_2 \cap (ub)^{-1}K)
  \end{equation}
  Before proving this claim, let us explain why it concludes the proof that $K \cap L \in \pol{\Cs}$. First observe that since $\Cs\subseteq\reg$, we know from Myhill-Nerode Theorem (Theorem~\ref{thm:auto:nerode}) that there are finitely many quotients of $K$. Hence, there finitely many languages of the form $\bigcap_{x \in Q_u} K(bx)^{-1}$ and $(ub)^{-1}K$ which means that the union over all $u \in A^*$ in~\eqref{eq:hintro:intersec} actually ranges over finitely many distinct languages. Therefore, since \pol{\Cs} is closed under finite union, it suffices to prove that for any $u \in A^*$,
  \[
    \left(L_1 \cap \bigcap_{x \in Q_u} K(bx)^{-1}\right) \cdot b \cdot (L_2 \cap (ub)^{-1}K) \text{~~~belongs to~} \pol{\Cs}
  \]
  Since \pol{\Cs} is closed under marked concatenation (see Lemma~\ref{lem:hintro:polmarked}), it suffices to prove that the two following properties hold:
  \[
    L_1 \cap \bigcap_{x \in Q_u} K(bx)^{-1} \in \pol{\Cs} \quad \text{~~and~~} \quad L_2 \cap (ub)^{-1}K \in \pol{\Cs}.
  \]
  For the first property, we know by hypothesis that $L_1 \in \Cs$. Moreover, for any $x \in Q_u$, $K(bx)^{-1} \in \Cs$ since it is a quotient of $K \in \Cs$. Furthermore, since $K$ is regular (as $K\in\Cs\subseteq\reg$), it has finitely many quotients, which means that $L_1 \cap \bigcap_{x \in Q_u} K(bx)^{-1}$ is a finite intersection of languages in \Cs and belongs to \Cs as well. Since $\Cs\subseteq\pol\Cs$, it belongs also to $\pol\Cs$. For the second property, note that $(ub)^{-1}K \in \Cs$ since \Cs is  closed under quotient. Therefore, since $L_2$ is a \Cs-polynomial of degree at most $n-1$ by hypothesis, we may use induction to conclude that $L_2 \cap (ub)^{-1}K \in \pol{\Cs}$.

  \smallskip

  It remains to prove that~\eqref{eq:hintro:intersec} holds. Assume first that $w \in K \cap L$. Then $w \in K$ and it can be decomposed as $w = w_1bw_2$ with $w_1 \in L_1$ and $w_2 \in L_2$. It follows that $w_2 \in L_2 \cap (w_1b)^{-1}K$ and by definition of $Q_{w_1}$, we have $w_1 \in L_1 \cap \bigcap_{x \in Q_{w_1}} K(bx)^{-1}$.
  This terminates the proof of the left to right inclusion.

  Conversely, assume that $w$ belongs to the right hand side of~\eqref{eq:hintro:intersec}. We have to prove that $w \in K \cap L$. It is immediate from the definition that $w \in L_1bL_2 = L$. It remains to prove that $w \in K$. By definition, there exists $u \in A^*$ such that $w$ can be decomposed as $w = w_1bw_2$ with $w_1 \in \bigcap_{x \in Q_u} K(bx)^{-1}$ and $w_2 \in (ub)^{-1}K$. Since $w_2 \in (ub)^{-1}K$, we have $u \in K(bw_2)^{-1}$ and therefore $w_2 \in Q_u$ by definition. We conclude that $w_1 \in K(bw_2)^{-1}$, which exactly means that $w = w_1bw_2 \in K$.

  \medskip
  \noindent
  {\bf General case.} We now assume that $K$ and $L$ both have arbitrary degrees $m$ and~$n$. To prove that $K \cap L \in \pol{\Cs}$, we proceed by induction on the sum $m+n$ of the degrees. The argument is similar to the one above. If $m = 0$ or $n = 0$, this is exactly the special case. Hence, we assume that $m,n \geq 1$: $K$ and $L$ may be decomposed as,
  \[
    K = K_1aK_2 \quad \text{and} \quad L= L_1bL_2,
  \]
  where $K_1,K_2 \in \Cs$ and $L_1,L_2$ are \Cs-monomials of degree at most $m-1$ and $n-1$. Observe that a word~$w$ belongs to $K \cap L$ if and only if it admits two decompositions witnessing its membership in $K_1aK_2$ and $L_1bL_2$, respectively. We use this observation to break down $K \cap L$ as the union of two languages (or three depending on whether $a=b$ or not). Consider the three following languages:
  \[
    \begin{array}{lll}
      H_\ell & = & \{w_1aw_2bw_3 \mid w_1 \in K_1, w_2bw_3 \in K_2, w_1aw_2 \in L_1 \text{ and } w_3 \in L_2\},\\
      H_r & = & \{w_1bw_2aw_3 \mid w_1 \in L_1, w_2aw_3 \in L_2, w_1bw_2 \in K_1 \text{ and } w_3 \in K_2\},\\
      H_c & = & \{w_1aw_2 \mid w_1 \in L_1 \cap K_1 \text{ and } w_2 \in L_2 \cap K_2\}.
    \end{array}
  \]
  It is simple to verify that $K \cap L = H_\ell \cup H_r$ when $a \neq b$ and $K \cap L = H_\ell \cup H_r \cup H_c$ when $a = b$. Hence, it suffices to prove that $H_\ell,H_r$ and $H_c$ are \Cs-polynomials. Since the proof is similar for all three cases, we concentrate on $H_\ell$.

  \medskip

  Given any word $u \in A^*$, we write $P_u$ for the set of all words $x \in A^*$ such that $u \in L_1(ax)^{-1}$, \emph{i.e.}, $P_u=(ua)^{-1}L_1$. We claim that the following equality holds:
  \begin{equation}\label{eq:hintro:interfinal}
    H_\ell = \bigcup_{u \in A^*} \left(K_1 \cap \bigcap_{x \in P_u} L_1(ax)^{-1}\right) \cdot a \cdot (K_2 \cap ((ua)^{-1}L_1)bL_2)
  \end{equation}
  Before establishing~\eqref{eq:hintro:interfinal}, let us use it to prove that $H_\ell \in \pol{\Cs}$.	Since $\Cs \subseteq \reg$, we know from Myhill-Nerode Theorem (Theorem~\ref{thm:auto:nerode}) that there are finitely many quotients of $L_1$. Hence, there are finitely many languages of the form $\bigcap_{x \in P_u} L_1(ax)^{-1}$ and $((ua)^{-1}L_1)bL_2$ for $u \in A^*$. Since \pol{\Cs} is closed under finite union, we obtain from~\eqref{eq:hintro:interfinal} that in order to show $H_\ell\in\pol\Cs$, it suffices to prove that for all $u \in A^*$,
  \[
    \left(K_1 \cap \bigcap_{x \in P_u} L_1(ax)^{-1}\right) \cdot a \cdot (K_2 \cap ((ua)^{-1}L_1)bL_2) \text{~~belongs to} \pol{\Cs}.
  \]
  Let $u \in A^*$. Since \pol{\Cs} is closed under marked concatenation (see Lemma~\ref{lem:hintro:polmarked}) it suffices to prove that the two following properties hold:
  \[
    K_1 \cap \bigcap_{x \in P_u} L_1(ax)^{-1} \in  \pol{\Cs} \quad \text{~~and~~} \quad K_2 \cap ((ua)^{-1}L_1)bL_2 \in \pol{\Cs}.
  \]
  The argument for the first property is identical to the special case that we treated above. Indeed, we know that $K_1 \in \Cs$ by hypothesis. Moreover, for any $x \in P_u$, $L_1(ax)^{-1} \in \Cs$ (it is a quotient of $L_1 \in \Cs$). Finally, $L_1 \in \Cs\subseteq\reg$ has finitely many quotients and $K_1 \cap \bigcap_{x \in P_u} L_1(ax)^{-1}$ is a finite intersection of languages in \Cs and belongs to \Cs as well, whence to $\pol\Cs$. For the other property, we have $(ua)^{-1}L_1 \in \Cs$ since \Cs is closed under quotient. Moreover, since $L_2$ is a \Cs-monomial of degree at most $n-1$, this entails that $((ua)^{-1}L_1)bL_2$ is a \Cs-monomial of degree at most $n$. Finally, $K_2$ is a \Cs-monomial of degree at most $m-1$. Therefore, we may use induction on the sum of the degrees to conclude that $K_2 \cap ((ua)^{-1}L_1)bL_2 \in \pol{\Cs}$.

  \smallskip

  It remains to prove that~\eqref{eq:hintro:interfinal} holds. Assume first that $w \in H_\ell$. Then $w = w_1aw_2bw_3$ with $w_1 \in K_1$, $w_2bw_3 \in K_2$, $w_1aw_2 \in L_1$ and $w_3 \in L_2$.	It follows that $w_2 \in (w_1a)^{-1}L_1$ and therefore, $w_2bw_3 \in K_2 \cap ((w_1a)^{-1}L_1)bL_2$. Finally, by definition of $P_{w_1}$, we have $w_1 \in K_1 \cap \bigcap_{x \in P_{w_1}} L_1(ax)^{-1}$. Therefore, we conclude that,
  \[
    w \in (K_1 \cap \bigcap_{x \in P_{w_1}} L_1(ax)^{-1}) \cdot a \cdot (K_2 \cap ((w_1a)^{-1}L_1)bL_2).
  \]
  This terminates the proof of the left to right inclusion. Conversely, assume that $w$ belongs to the right hand side of~\eqref{eq:hintro:interfinal}. We have to prove that $w \in H_\ell$. By definition, there exists $u \in A^*$ such that $w = w_1ay$ with,
  \[
    w_1 \in K_1 \cap \bigcap_{x \in P_u} L_1(ax)^{-1} \quad \text{and} \quad y \in K_2 \cap ((ua)^{-1}L_1)bL_2.
  \]
  Therefore, $y$ can be decomposed as $y = w_2bw_3$ with $w_2 \in (ua)^{-1}L_1$ and $w_3 \in L_2$. Hence, we have $w = w_1aw_2bw_3$ with $w_1 \in K_1$, $w_2bw_3 \in K_2$ and $w_3 \in L_2$. To prove that $w \in H_\ell$, it remains to prove that $w_1aw_2 \in L_1$. Since $w_2 \in (ua)^{-1}L_1$, we have $u \in L_1(aw_2)^{-1}$ and therefore $w_2 \in P_u$ by definition. We conclude that $w_1 \in L_1(aw_2)^{-1}$, which exactly means that $w = w_1aw_2 \in L_1$, as desired.
\end{proof}

We finish the presentation of the closure properties of \pol{\Cs} with a few additional results that also require \Cs to be a \pvari of regular languages. The first one is closure concatenation.

\begin{lemma}\label{lem:hintro:polconcat}
  Let \Cs be a \pvari of regular languages. Then \pol{\Cs} is closed under concatenation.
\end{lemma}

\begin{proof}
  We use closure under marked concatenations and quotient (see Lemma~\ref{lem:hintro:polmarked} and Theorem~\ref{thm:hintro:polc}). Let $K$ and $L$ be two languages in \pol{\Cs} over some alphabet $A$. If $\varepsilon \in L$, one may verify that:
  \[
    KL =
    \begin{cases}
      \bigcup_{a \in A} Ka(a^{-1}L) \cup K&\quad\text{if $\varepsilon \in L$},\\
      \bigcup_{a \in A} Ka(a^{-1}L) &\quad\text{if $\varepsilon \not\in L$}.
    \end{cases}
  \]
  In either case, we obtain a language from $\pol{\Cs}$, which terminates the proof.
\end{proof}

\smallskip\noindent\textbf{Additional operations.} We end this section by looking at two additional operations that are built from Boolean and polynomial closures. The first one is simply the composition of the two: given any class \Cs, we write \bpol{\Cs} for the class \bool{\pol{\Cs}}. Combining the results of the previous subsections, we obtain the following result.

\begin{proposition}\label{prop:hintro:bpolc}
  For any \pvari of regular languages \Cs, the class \bpol{\Cs} is a \vari.
\end{proposition}

The second operation is motivated by a simple observation about polynomial closure. We proved in Theorem~\ref{thm:hintro:polc} that \pol{\Cs} is a lattice provided that \Cs is a \pvari of regular languages. However, it is simple to verify that in general, \pol{\Cs} is \emph{not} closed under complement, even when \Cs is.

\begin{example}\label{ex:hintro:copol}
  Consider the class $\Cs = \{\emptyset,A^*\}$. It turns out that $\pol{\Cs}$ is not closed under complement. Indeed, it is clear that $A^+= \bigcup_{a \in A} A^*aA^* \in \pol{\Cs}$. However, it follows from Lemma~\ref{lem:hintro:epsilon} that its complement $\{\varepsilon\}$ does not belong to \pol{\Cs}.
\end{example}

When dealing with a lattice \Ds which is not closed under complement, it makes sense to consider the \emph{complement class} which we denote by $\overline{\Ds}$. More precisely, $\overline{\Ds}$ contains all complements of languages in \Ds: for any language $L \subseteq A^*$, we have $L \in \overline{\Ds}$ if and only if $A^* \setminus L \in \Ds$. Note that since \Ds is lattice, so is $\overline{\Ds}$ by DeMorgan's laws. We shall often consider the complement of classes which are polynomial closures: given a class \Cs, we consider $\overline{\pol{\Cs}}$.

It is simple to transfer most of the closure properties of \pol{\Cs} to its complement class $\overline{\pol{\Cs}}$. This yields the following proposition.

\begin{proposition}\label{prop:hintro:copolc}
  For any \pvari of regular languages \Cs, the class $\overline{\pol{\Cs}}$ is a \pvari.
\end{proposition}

\begin{proof}
  The properties of \pol{\Cs} can be easily transferred to $\overline{\pol{\Cs}}$ using DeMorgan's laws and the fact quotients commute with Boolean operations.
\end{proof}

Note that in contrast to \pol{\Cs}, in general, $\overline{\pol{\Cs}}$ is \emph{\textbf{not}} closed under (marked) concatenation, as shown by the example $\Cs=\{\emptyset,A^*\}$. However, this is less problematic than it is for Boolean closure. By definition, $\overline{\pol{\Cs}}$ and \pol{\Cs} are dual and the associated membership and separation problems are inter-reducible. Indeed, $L \in \overline{\pol{\Cs}}$ if and only if $A^* \setminus L \in \pol{\Cs}$ and $L_1$ is $\overline{\pol{\Cs}}$-separable from $L_2$ if and only if $L_2$ is \pol{\Cs}-separable from $L_1$. Hence, when considering these problems, one may simply work with \pol{\Cs} instead of $\overline{\pol{\Cs}}$.

In view of these observations, one might wonder about our motivation for considering the complement of polynomial closure. Indeed, we just explained that $\overline{\pol{\Cs}}$ is less robust than \pol{\Cs}, while the associated decision problems are symmetrical with the ones associated to \pol{\Cs}. Our motivation is explained by the next lemma, which can be used to bypass Boolean closure in some cases.

\begin{lemma}\label{lem:hintro:elimbool}
  Let \Cs be a \pvari of regular languages. Then,
  \begin{equation}\label{eq:polbp-=-polov}
    \pol{\bpol{\Cs}} = \pol{\overline{\pol{\Cs}}}.
  \end{equation}
\end{lemma}

\begin{proof}
  It is clear that $\overline{\pol{\Cs}} \subseteq \bpol{\Cs}$, whence $\pol{\overline{\pol{\Cs}}} \subseteq \pol{\bpol{\Cs}}$. We show that $\bpol{\Cs} \subseteq \pol{\overline{\pol{\Cs}}}$. Since polynomial closure is an idempotent operation by Lemma~\ref{lem:hintro:polidem}, the other inclusion will follow. Let $L$ be a language in \bpol{\Cs}. By definition, $L$ is a Boolean combination of \Cs-polynomials. Furthermore, using DeMorgan's laws, we obtain that $L$ is built by applying unions and intersections to languages that are either \Cs-polynomials (\emph{i.e.}, elements of \pol{\Cs}) or complements of \Cs-polynomials (\emph{i.e.}, elements of $\overline{\pol{\Cs}}$). It follows that $L \in \pol{\overline{\pol{\Cs}}}$, since $\pol{\Cs} \subseteq \pol{\overline{\pol{\Cs}}}$, $\overline{\pol{\Cs}} \subseteq \pol{\overline{\pol{\Cs}}}$ and $\pol{\overline{\pol{\Cs}}}$ is closed under union and intersection by Theorem~\ref{thm:hintro:polc}.
\end{proof}

Let us finish the section with a recap of all closure properties that we proved for the four operations on classes that we defined. We present it in \figurename~\ref{fig:hintro:props}.
\begin{figure}[!hbt]
  \centering
  \begin{tikzpicture}
    \matrix (M) [matrix of nodes, column  sep=0mm,row  sep=.5mm,draw,very thick,
    rounded corners=3pt,nodes={anchor=center,align=center,text width = 1.95cm, minimum height=.2cm}]
    {
      {} & Intersection and Union  & Complement & Quotient & Concat. & {Marked\\ concat.}      \\
      \bpol{\Cs} & \textcolor{bookgreen}{Y} & \textcolor{bookgreen}{Y} & \textcolor{bookgreen}{Y} & \textcolor{bookred}{N} & \textcolor{bookred}{N} \\
      \bool{\Cs} & \textcolor{bookgreen}{Y} & \textcolor{bookgreen}{Y} & \textcolor{bookgreen}{Y} & \textcolor{bookred}{N} & \textcolor{bookred}{N} \\
      \pol{\Cs}   & \textcolor{bookgreen}{Y} & \textcolor{bookred}{N} & \textcolor{bookgreen}{Y} & \textcolor{bookgreen}{Y} & \textcolor{bookgreen}{Y} \\
      $\overline{\pol{\Cs}}$  & \textcolor{bookgreen}{Y} & \textcolor{bookred}{N} & \textcolor{bookgreen}{Y} & \textcolor{bookred}{N} & \textcolor{bookred}{N} \\
    };

    \foreach \row in {2,3,4,5} {
      \mhline[thick]{M}{\row}
    }

  \end{tikzpicture}
  \caption{Closure properties satisfied for any \pvari of regular languages \Cs}
 \label{fig:hintro:props}
\end{figure}

\section{Concatenation hierarchies: definition and generic results}\label{chap:hieraintro}
We may now move to the main topic of this paper: concatenation hierarchies. As explained in the introduction, a natural complexity measure for star-free languages is the number of alternations between concatenation and complement that are required to build a given language from basic star-free languages. Generalizing this idea leads to the notion of concatenation hierarchy, which is meant to classify languages according to such a complexity measure.

The section is organized as follows. We first define concatenation hierarchies. Then, we present a stratification of polynomial closures. Finally, we prove that any concatenation hierarchy with a finite basis is~strict.

\medskip\noindent\textbf{Concatenation hierarchies: Definition.} A particular hierarchy is built from a starting class of languages~\Cs, which is called its \emph{basis}. In order to get nice properties, we restrict \Cs to be a \vari of regular languages. This is the only parameter of the construction, meaning that once the basis has been chosen, the construction is entirely generic: each new level is built from the previous one by applying a closure operation: either \emph{Boolean closure}, or \emph{polynomial closure}. We speak of the ``(concatenation) hierarchy of basis~\Cs''.

\smallskip
In the concatenation hierarchy of basis \Cs, languages are classified into levels of two distinct kinds: full levels (denoted by $0,1,2,3,\dots$) and half levels (denoted by $\frac{1}{2},\frac{3}{2}, \frac{5}{2},\dots$). The definition is as follows:
\begin{itemize}
\item Level $0$ is the basis (\emph{i.e.}, our parameter class \Cs). We denote it by $\Cs[0]$.
\item \emph{Half levels} are the \emph{polynomial closure} of the previous full level: for any $n \in \nat$, level $n+\frac{1}{2}$ is the polynomial closure of level $n$. We denote it by~$\Cs[n+\frac{1}{2}]$.
\item \emph{Integer levels} are the \emph{Boolean closure} of the previous half level: for any $n \in \nat$, level $n+1$ is the Boolean closure of level $n+\frac{1}{2}$. We denote it by~$\Cs[n+1]$.
\end{itemize}
We give a graphical representation of the construction process of a concatenation hierarchy in \figurename~\ref{fig:hintro:concat} below.

\begin{figure}[!htb]
   \begin{center}
    \begin{tikzpicture}
      \node[anchor=east] (l00) at (0.0,0.0) {{\large $0$}};
      \node[anchor=north] at ($(l00)-(0,0.25)$)  {(basis)};
      \node[anchor=east] (l12) at (1.5,0.0) {\large $\frac{1}{2}$};
      \node[anchor=east] (l11) at (3.0,0.0) {\large $1$};
      \node[anchor=east] (l32) at (4.5,0.0) {\large $\frac{3}{2}$};
      \node[anchor=east] (l22) at (6.0,0.0) {\large $2$};
      \node[anchor=east] (l52) at (7.5,0.0) {\large $\frac{5}{2}$};
      \node[anchor=east] (l33) at (9.0,0.0) {\large $3$};
      \node[anchor=east] (l72) at (10.5,0.0) {\large $\frac{7}{2}$};

      \draw[very thick,->] (l00) to node[above] {$Pol$} (l12);
      \draw[very thick,->] (l12) to node[below] {$Bool$} (l11);

      \draw[very thick,->] (l11) to node[above] {$Pol$} (l32);
      \draw[very thick,->] (l32) to node[below] {$Bool$} (l22);
      \draw[very thick,->] (l22) to node[above] {$Pol$} (l52);
      \draw[very thick,->] (l52) to node[below] {$Bool$} (l33);
      \draw[very thick,->] (l33) to node[above] {$Pol$} (l72);

      \draw[very thick,dotted] (l72) to ($(l72)+(1.0,0.0)$);

    \end{tikzpicture}

  \end{center}
  \caption{A concatenation hierarchy}
  \label{fig:hintro:concat}
\end{figure}

Observe that by definition, for any $n \in \nat$, we have $\Cs[n]\subseteq\Cs[n+\frac{1}{2}]\subseteq\Cs[n+1]$. However, these inclusions need not be strict. For instance, if the basis is closed under Boolean operations and marked concatenation (such as the class of star-free languages, or the whole class \reg), the associated hierarchy collapses at level~$0$. Of course, the interesting hierarchies are the ones that \emph{are} strict.

What we gain by imposing that the basis must be a \vari of regular languages are the following properties, obtained as an immediate corollary of Theorem~\ref{thm:hintro:polc} and Proposition~\ref{prop:hintro:boolc}.

\begin{proposition}\label{prop:hintro:concatvari}
  Let \Cs be a \vari of regular languages and consider the concatenation hierarchy of basis \Cs. Then, all half levels are \pvaris of regular languages and all full levels are \varis of regular languages.
\end{proposition}

The half levels are even more robust: since they are polynomial closures, they are closed under concatenation and marked concatenation (by Lemmas~\ref{lem:hintro:polconcat} and~\ref{lem:hintro:polmarked}).

\begin{proposition}\label{prop:hintro:closconcat}
  Let \Cs be a \vari of regular languages. Then, all half levels in the concatenation hierarchy of basis \Cs are closed under concatenation and marked concatenation.
\end{proposition}

In contrast, it is important to note that for a hierarchy to be strict, half levels must not be closed under complement and full levels must not be closed under marked concatenation. Indeed, a half level that is closed under complement would be equal to its Boolean closure and the hierarchy would collapse at this level. Similarly, a full level that is closed under marked concatenation would be equal to its polynomial closure and the hierarchy would collapse as well.

\begin{proposition}\label{prop:hintro:notclosed}
  Let \Cs be a \vari of regular languages. The following properties are equivalent:
  \begin{enumerate}
  \item The concatenation hierarchy of basis \Cs is strict.
  \item No half level in the hierarchy of basis \Cs is closed under complement.
  \item No full level in the hierarchy of basis \Cs is closed under marked concatenation.
  \end{enumerate}
\end{proposition}

\begin{remark}\label{rem:hintro:approach}
  Propositions~\ref{prop:hintro:closconcat} and~\ref{prop:hintro:notclosed} are simple consequences of the definitions. However, they are important for understanding our approach when considering membership and separation for levels within a concatenation hierarchy. Indeed, our techniques for solving these questions rely heavily on the concatenation operation. Hence, we prefer to work with classes that are closed under concatenation. This excludes full levels which cannot be closed under marked concatenation in a strict hierarchy. This will be reflected by our approach: all results---even those that apply to full levels---are based on the investigation of a half level.
\end{remark}

Remark~\ref{rem:hintro:approach} is complemented by the following useful observation. When only considering the half levels, one may bypass the full levels in the definition by applying polynomial closure to the complements of half levels. This trick is based on Lemma~\ref{lem:hintro:elimbool}. Let \Cs be a basis. Observe that by definition, for any $n \geq 1$, level $\Cs[n + \frac{1}{2}]$ is defined~as,
\[
  \Cs[n + \tfrac{1}{2}] = \pol{\Cs[n]} = \pol{\bpol{\Cs[n-1]}}.
\]
Applying Lemma~\ref{lem:hintro:elimbool}, we obtain the following alternate definition of level $\Cs[n + \frac{1}{2}]$:
\[
  \Cs[n + \tfrac{1}{2}] = \pol{\bpol{\Cs[n-1]}} = \pol{\overline{\pol{\Cs[n-1]}}} = \pol{\overline{\Cs[n-\tfrac{1}{2}]}}.
\]
The important point here is that the class $\overline{\Cs[n-\frac{1}{2}]}$ is much simpler to manipulate than $\Cs[n]$. Indeed, the associated membership and separation problems are dual with the ones for the class $\Cs[n-\frac{1}{2}]$, which is closed under concatenation and marked concatenation. Altogether, we obtain the following proposition.

\begin{proposition}\label{prop:hintro:bypassfull}
  Let \Cs be a \vari of regular languages and consider the associated concatenation hierarchy. Then, for any natural number $n \geq 1$,
  \[
    \Cs[n+\tfrac{1}{2}] = \pol{\overline{\Cs[n-\tfrac{1}{2}]}}.
  \]
\end{proposition}

\noindent
In view of Proposition~\ref{prop:hintro:bypassfull}, we update the construction process of a concatenation hierarchy in \figurename~\ref{fig:hintro:concat2}.

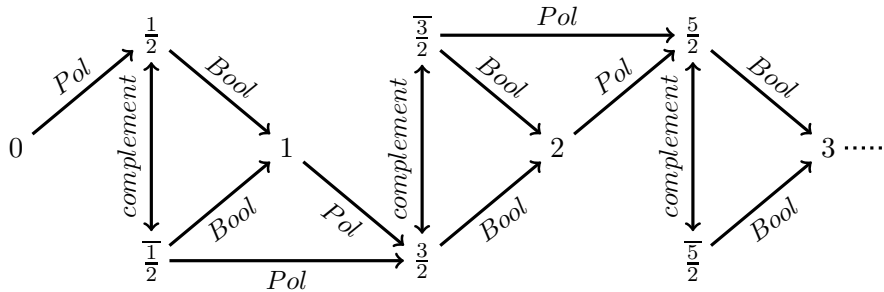
\begin{figure}[!htb]
  \begin{center}
    \begin{tikzpicture}[xscale=.9]
      \node (l00) at (0.0,0.0) {{\large $0$}};

      \node (l12) at (2.0,1.5) {\large $\frac{1}{2}$};
      \node (l12c) at (2.0,-1.5) {\large $\overline{\frac{1}{2}}$};

      \node (l11) at (4.0,0.0) {\large $1$};

      \node (l32c) at (6.0,1.5) {\large $\overline{\frac{3}{2}}$};
      \node (l32) at (6.0,-1.5) {\large $\frac{3}{2}$};

      \node (l22) at (8.0,0.0) {\large $2$};

      \node (l52) at (10.0,1.5) {\large $\frac{5}{2}$};
      \node (l52c) at (10.0,-1.5) {\large $\overline{\frac{5}{2}}$};

      \node (l33) at (12.0,0.0) {\large $3$};

      \draw[very thick,->] (l00) to node[above,sloped] {$Pol$} (l12);
      \draw[very thick,->] (l12) to node[above,sloped] {$Bool$} (l11);
      \draw[very thick,<->] (l12c) to node[above,sloped] {$complement$} (l12);
      \draw[very thick,->] (l12c) to node[below,sloped] {$Bool$} (l11);

      \draw[very thick,->] (l11) to node[below,sloped] {$Pol$} (l32);
      \draw[very thick,->] (l32) to node[below,sloped] {$Bool$} (l22);
      \draw[very thick,<->] (l32) to node[above,sloped] {$complement$} (l32c);
      \draw[very thick,->] (l32c) to node[above,sloped] {$Bool$} (l22);

      \draw[very thick,->] (l22) to node[above,sloped] {$Pol$} (l52);
      \draw[very thick,->] (l52) to node[above,sloped] {$Bool$} (l33);
      \draw[very thick,<->] (l52c) to node[above,sloped] {$complement$} (l52);
      \draw[very thick,->] (l52c) to node[below,sloped] {$Bool$} (l33);

      \draw[very thick,->] (l32c) to node[above,sloped] {$Pol$} (l52);

      \draw[very thick,->] (l12c) to node[below,sloped] {$Pol$} (l32);

      \draw[very thick,dotted] (l33) to ($(l33)+(0.8,0.0)$);
    \end{tikzpicture}

  \end{center}
  \caption{A concatenation hierarchy with complement levels}
 \label{fig:hintro:concat2}
\end{figure}

\medskip\noindent\textbf{Stratifying polynomial closures.}
In this section, we define a generic method for stratifying the class \pol{\Cs} when \Cs is a {\bf finite} \pvari, which will play a crucial role in many proofs. We shall use it to prove that any concatenation hierarchy with a finite basis is strict.

We assume that an arbitrary \emph{finite} \pvari of regular languages \Cs is fixed for the whole section. All definitions and results that we present now are parameterized by \Cs. We begin by presenting our stratification of \pol{\Cs}. Then, we introduce the preorder relations associated to each stratum and analyze their properties. Finally, we illustrate the definitions with a few examples.

\medskip\noindent\textbf{Definition.} We present a stratification of the class \pol{\Cs}. More precisely, given any $k \in \nat$, we define a \emph{finite} class \polk{\Cs} such that,
\begin{equation}\label{eq:strat-polc}
  \text{For all $k \in \nat$, } \polk{\Cs} \subseteq \polp{\Cs}{k+1} \quad \text{and} \quad \pol{\Cs} = \bigcup_{k \in \nat} \polk{\Cs}.
\end{equation}

Intuitively, the definition counts the number of marked concatenations that are necessary to define a particular language in \pol{\Cs}. We use induction on $k$.
\begin{itemize}
\item When $k = 0$, we simply define $\polp{\Cs}{0} = \Cs$.
\item When $k \geq 1$, \polk{\Cs} is the smallest lattice such that:
  \begin{enumerate}
  \item $\polp{\Cs}{k-1} \subseteq \polk{\Cs}$.
  \item For any $a \in A$ and $L_1,L_2 \in \polp{\Cs}{k-1}$, we have $L_1 a L_2 \in \polk{\Cs}$.
  \end{enumerate}
\end{itemize}

This concludes the definition. Since \Cs is a finite lattice, it is immediate that all classes \polk{\Cs} are finite lattices as well. Moreover,~\eqref{eq:strat-polc} indeed holds.
Hence, we did define a stratification of \pol{\Cs}. Let us prove that the classes \polk{\Cs} are not only lattices but \pvaris as well.

\begin{lemma}\label{lem:hintro:stratl1vari}
  For any $k \in \nat$, \polk{\Cs} is a finite \pvari.
\end{lemma}

\begin{proof}
  It is clear from the definition that for any $k \in \nat$, \polk{\Cs} is a lattice. Hence, it suffices to prove closure under quotients. We use induction on $k$. When $k = 0$, then $\polp{\Cs}{0} = \Cs$ which is a \pvari by hypothesis. Assume now that $k \geq 1$. Let $L \in \polk{\Cs}$ and $w \in A^*$, we prove that $Lw^{-1} \in \polk{\Cs}$ (as usual, the argument is symmetrical for left quotients). Since for any $b\in A$ and $u\in A^*$, we have $L(bu)^{-1}=(Lu^{-1})b^{-1}$, we may assume without loss of generality that $w$ is a letter, say $w=b\in A$. Finally, by definition of \polk{\Cs} and since quotients commute with unions and intersections, we only have two cases to consider:
  \begin{enumerate}
  \item $L \in \polp{\Cs}{k-1}$.
  \item $L = L_1aL_2$ with $L_1,L_2 \in \polp{\Cs}{k-1}$ and $a \in A$.
  \end{enumerate}
  The first case is immediate by induction on $k$. For the second, observe that we have:
  \[
    Lb^{-1} = \left\{
      \begin{array}{ll}
        L_1a(L_2 b^{-1}) \cup L_1 & \quad\text{if $a = b$ and $\varepsilon \in L_2$}, \\
        L_1a(L_2 b^{-1})          & \quad\text{otherwise}.
      \end{array}\right.
  \]
  An immediate induction on $k$ yields $L_2 b^{-1}\in \polp{\Cs}{k-1}$. Hence, we conclude that $Lb^{-1}$ is a union of languages in \polk{\Cs} and thus belongs to \polk{\Cs} itself.
\end{proof}

Finally, a useful observation is that this stratification may be lifted to the Boolean closure \bpol{\Cs}. Indeed, it suffices to choose the strata as the classes $\bpolk{\Cs} = \bool{\polk{\Cs}}$ for all $k \in \nat$. It is clear that these classes are finite and we have:
\[
  \text{For all $k \in \nat$, } \bpolk{\Cs} \subseteq \bpolp{\Cs}{k+1} \quad \text{and} \quad \bpol{\Cs} = \bigcup_{k \in \nat} \bpolk{\Cs}.
\]
In particular, since we proved in Lemma~\ref{lem:hintro:stratl1vari} that the classes \polk{\Cs} are \pvaris, we obtain the following result from Proposition~\ref{prop:hintro:boolc}.

\begin{lemma}\label{lem:hintro:stratbool}
  For any $k \in \nat$, \bpolk{\Cs} is a finite \vari.
\end{lemma}

\smallskip\noindent\textbf{Canonical relations.} Now that we have a stratification for \pol{\Cs}, we consider the canonical relations associated to the strata. For $k \in \nat$, let \polrelk~be the preorder associated to \polk{\Cs}.
Recall from Fact~\ref{fct:metho:finer} that, for all $k \in \nat$, $w_1 \polrelp{k+1} w_2  \Rightarrow  w_1 \polrelk w_2$.
Moreover, since all classes \polk{\Cs} are \pvaris, Lemma~\ref{lem:metho:stratlem}  yields the following lemma.

\begin{lemma}\label{lem:hintro:stratpol}
  For any $k \in \nat$, \polrelk is a precongruence with finitely many upper sets. Similarly, \bpolrelk is a congruence of finite index. Moreover, for any language $L \subseteq A^*$,
  \begin{enumerate}
  \item $L \in \polk{\Cs}$ if and only if $L$ is an upper set for \polrelk.
  \item $L \in \pol{\Cs}$ if and only if there exists $k \in \nat$ such that $L$ is an upper set for \polrelk.
  \end{enumerate}
\end{lemma}

We shall use Lemma~\ref{lem:hintro:stratpol} to prove that languages do not belong to \pol{\Cs}.

\begin{corollary}\label{cor:hintro:stratpol}
  Let $K,L \subseteq A^*$ be two languages. The following properties hold:
  \begin{enumerate}
  \item $L$ does {\bf not} belong to \pol{\Cs} iff for all $k \in \nat$ there exists $w \in L$ and $w' \not\in L$ such that $w \polrelk w'$.
  \item $L$ is {\bf not} \pol{\Cs}-separable from $K$ iff for all $k \in \nat$ there exists $w \in L$ and $w' \in K$ such that $w \polrelk w'$.
  \end{enumerate}
\end{corollary}

\medskip\noindent\textbf{Properties.} We now present specific properties of the preorders \polrelk. We start with an alternate definition of \polrelk which is easier to manipulate for proving these properties. Recall that $\leq_\Cs$ is the canonical preorder associated to the finite \pvari \Cs.

\begin{lemma}\label{lem:hintro:preoinduc}
  Let  $k$ be a natural number. For any two words $w,w' \in A^*$, we have $w \polrelk w'$ if and only if the two following properties hold:
  \begin{enumerate}
  \item $w \leq_\Cs w'$
  \item If $k > 0$, for any decomposition $w = uav$ with $u,v \in A^*$ and $a \in A$, there exist $u',v' \in A^*$ such that $w' = u'av'$, $u \polrelp{k-1} u'$ and $v \polrelp{k-1} v'$.
  \end{enumerate}
\end{lemma}

\begin{proof}
  Assume first that $w \polrelk w'$. We have to prove that the two items in the lemma hold. For the first item, observe that by definition, $\Cs \subseteq \polk{\Cs}$. Therefore, $w \polrelk w' \Rightarrow w \leq_\Cs w'$. We turn to the second item. Assume that $k > 0$ and consider a decomposition $w = uav$ of $w$. We have to find an appropriate decomposition of~$w'$. Let $K_u$ and $K_v$ be the upper sets of $u$ and $v$ for \polrelp{k-1}. By Lemma~\ref{lem:hintro:stratpol}, we know that $K_u,K_v \in \polp{\Cs}{k-1}$. Hence, $K_u aK_v \in \polk{\Cs}$ by definition. Moreover, since $w = uav \in  K_u aK_v$ and $w \polrelk w'$, it follows that $w' \in K_u a K_v$. Therefore, we obtain $u' \in K_u$ and $v' \in K_v$ such that $w' = u'av'$. It is then immediate by definition of $K_u$ and $K_v$ that $u \polrelp{k-1} u'$ and $v \polrelp{k-1} v'$.

  Conversely, assume that the two items in the lemma hold. We prove that $w \polrelk w'$. When $k = 0$, this is immediate since $\polp{\Cs}{0} = \Cs$ by definition. Therefore, $\polrelp{0}$ and $\leq_\Cs$ are the same relation, and the first item says that $w \leq_\Cs w'$. Assume now that $k>0$. Let $L \in \polk{\Cs}$, we have to prove that $w \in L \Rightarrow w' \in L$. By definition, $L$ is constructed by applying finitely many unions and intersections to the two following kinds of languages:
  \begin{enumerate}
  \item Languages in $\Cs$.
  \item Languages of the form $L_1aL_2$ with $L_1,L_2 \in \polp{\Cs}{k-1}$.
  \end{enumerate}
  The proof is by induction on this construction.
  \begin{itemize}
  \item When $L \in \Cs$ the implication is immediate since $w \leq_\Cs w'$ by the first item.
  \item Assume now that $L = L_1aL_2$ with $L_1,L_2 \in \polp{\Cs}{k-1}$. If $w \in L = L_1aL_2$, then it admits a decomposition $w = uav$ with $u \in L_1$ and $v \in L_2$. By the second item, we obtain $u',v' \in A^*$ such that $w' = u'av'$, $u \polrelp{k-1} u'$ and $v \polrelp{k-1} v'$. In particular, since $L_1,L_2 \in \polp{\Cs}{k-1}$, it follows by definition of $\polrelp{k-1}$ that $u' \in L_1$ and $v' \in L_2$, \emph{i.e.}, that $w' =u'av' \in L_1aL_2 = L$.
  \item Finally, if $L = L_1 \cup L_2$ or $L = L_1 \cap L_2$, we obtain inductively that $w \in L_1 \Rightarrow w' \in L_1$ and $w \in L_2 \Rightarrow w' \in L_2$ and therefore, $w \in L \Rightarrow w' \in L$.
  \end{itemize}
  This terminates the proof of Lemma~\ref{lem:hintro:preoinduc}.
\end{proof}

We now use Lemma~\ref{lem:hintro:preoinduc} to present and prove two characteristic properties of the relations \polrelk. Recall that since \Cs is a finite \pvari, we know from Lemma~\ref{lem:metho:omegapower} that there exists a natural number $p \geq 1$ called the \emph{period of $\Cs$} such that for any word $u$ and any $m,m' \geq 1$, we have:
\begin{equation}\label{eq:period}
  u^{pm} \leq_\Cs u^{pm'}.
\end{equation}

The two properties that we state now depend on this important parameter of~\Cs.

\begin{lemma}\label{lem:hintro:propreo1}
  Let $p$ be the period of $\Cs$. Consider some natural number $k \in \nat$. Then, for any $m,m' \geq 2^{k+1}-1$ and any word $u \in A^*$, we have:
  \[
    u^{pm} \polrelk u^{pm'}.
  \]
\end{lemma}

\begin{proof}
  Let $m,m' \geq 2^{k+1}-1$ and let $u$ be some word. We prove that $u^{pm} \polrelk u^{pm'}$. This amounts to proving that the two items in Lemma~\ref{lem:hintro:preoinduc} hold. We argue by induction on $k$. For the first item, it suffices to prove that $u^{pm} \leq_\Cs u^{pm'}$. This is immediate since the period $p$ of~\Cs satisfies~\eqref{eq:period} by Lemma~\ref{lem:metho:omegapower}. This concludes the case $k = 0$.

  When $k\geq1$, it remains to prove Item~2 of Lemma~\ref{lem:hintro:preoinduc}. Given a decomposition $u^{pm}  = w_1aw_2$, we have to decompose $u^{pm'} = w'_1aw'_2$ so that $w_1 \polrelp{k-1} w'_1$ and $w_2 \polrelp{k-1} w'_2$. By definition, the letter $a$ in the decomposition $u^{pm}  = w_1aw_2$ falls within some factor $u^p$ of $u^{pm}$. Let us refine the decomposition to isolate this factor. We have $u^{pm}  = u^{pm_1}v_1av_2u^{pm_2}$ where,
  \begin{itemize}
  \item $m = m_1+1+m_2$
  \item $v_1av_2 = u^p$.
  \item $u^{pm_1}v_1 = w_1$ and $v_2u^{pm_2} = w_2$.
  \end{itemize}
  Since $m \geq 2^{k+1}-1$ by hypothesis and $m = m_1+1+m_2$, either $m_1 \geq 2^{k}-1$ or $m_2 \geq 2^{k}-1$ (possibly both). By symmetry, let us assume that $m_1 \geq 2^{k} - 1$. We use the following claim.

  \begin{claim}
    There exist $m'_1,m'_2 \geq 1$ such that $m' = m'_1+1+m'_2$, $u^{pm_1} \polrelp{k-1} u^{pm'_1}$ and $u^{pm_2} \polrelp{k-1} u^{pm'_2}$.
  \end{claim}

  \begin{proof}
    There are two cases depending on whether $m_2 \geq 2^{k}-1$ or not. Assume first that $m_2 \geq 2^{k}-1$. Since $m' \geq 2^{k+1}-1$, we may choose $m'_1,m'_2 \geq 2^k-1$ such that $m' = m'_1 + 1 + m'_2$. It is now immediate by induction on $k$ that $u^{pm_1} \polrelp{k-1} u^{pm'_1}$ and $u^{pm_2} \polrelp{k-1} u^{pm'_2}$. Otherwise, $m_2 < 2^{k}-1$. We let $m'_2 = m_2$ and $m'_1 = m' - 1 - m'_2$. Clearly, $m'_1 \geq 2^{k} -1$ since $m' \geq 2^{k+1}-1$. Hence, we get $u^{pm_1} \polrelp{k-1} u^{pm'_1}$ by induction on $k$. Furthermore, $u^{pm_2} \polrelp{k-1} u^{pm'_2}$ is immediate since $m_2 = m'_2$ by definition.
  \end{proof}

  We may now finish the proof of Item~2. Let $m'_1,m'_2 \geq 1$ be as defined in the claim. We let $w'_1 = u^{pm'_1}v_1$ and $w'_2 = v_2u^{pm'_2}$. Clearly, $w'_1 aw'_2 = u^{pm'}$ since $v_1 av_2 = u^p$ and $m' = m'_1+1+m'_2$. Moreover, since \polrelp{k-1} is compatible with multiplication, we have
  \[
    \begin{array}{rll}
      w_1 = u^{pm_1}v_1 & \polrelp{k-1} & u^{pm'_1}v_1 = w'_1 \\
      w_2 = v_2u^{pm_2} & \polrelp{k-1} & v_2u^{pm'_2} = w'_2
    \end{array}
  \]
  This terminates the proof of Item~2 of Lemma~\ref{lem:hintro:preoinduc}.
\end{proof}

\noindent We turn to the second property which will be crucial to establish the strictness of finitely based hierarchies.

\begin{lemma}\label{lem:hintro:propreo2}
  Let $p$ be the period of $\Cs$. Let $k\in\nat$ and let $u,v \in A^*$ be two words such that $u^p \leq_\Cs v$. Then, for any $m,m'_1,m'_2 \geq 2^{k+1}-1$, we have:
  \[
    u^{pm} \polrelk u^{pm'_1}vu^{pm'_2}.
  \]
\end{lemma}

\begin{proof}
  The proof is similar to that of Lemma~\ref{lem:hintro:propreo1}. Let $k\geq0$, let $u,v$ satisfying $u^p \leq_\Cs v$, and let $m,m'_1,m'_2 \geq 2^{k+1}-1$. We prove that $u^{pm} \polrelk u^{pm'_1}vu^{pm'_2}$. This amounts to proving the two items in Lemma~\ref{lem:hintro:preoinduc}. The argument is an induction on $k$.

  For Item~1, we prove that $u^{pm} \leq_\Cs u^{pm'_1}vu^{pm'_2}$. By hypothesis on $u,v$, we know that $u^p \leq_\Cs v$. Hence, since $\leq_\Cs$ is compatible with concatenation, it suffices to prove that $u^{pm} \leq_\Cs u^{p(m'_1+ 1 + m'_2)}$. This is immediate by choice of $p$ in Lemma~\ref{lem:metho:omegapower}. This proves Item~1, and the case $k = 0$.

  When $k\geq1$, it remains to prove Item~2 of Lemma~\ref{lem:hintro:preoinduc}. Consider a decomposition $u^{pm}  = w_1aw_2$. We have to find a decomposition $u^{pm'_1}vu^{pm'_2} = w'_1aw'_2$ such that $w_1 \polrelp{k-1} w'_1$ and $w_2 \polrelp{k-1} w'_2$. By definition, the letter~$a$ in the decomposition $u^{pm}  = w_1aw_2$ falls within some factor $u^p$ of $u^{pm}$. Let us refine the decomposition to isolate this factor. We have $u^{pm}  = u^{pm_1}v_1av_2u^{pm_2}$ where,
  \begin{itemize}
  \item $m = m_1+1+m_2$
  \item $v_1av_2 = u^p$.
  \item $u^{pm_1}v_1 = w_1$ and $v_2u^{pm_2} = w_2$.
  \end{itemize}
  Since $m \geq 2^{k+1}-1$ by hypothesis and $m = m_1+1+m_2$, either $m_1 \geq 2^{k}-1$ or $m_2 \geq 2^{k}-1$ (possibly both). By symmetry, let us assume that $m_1 \geq 2^{k} - 1$. We use the following claim.

  \begin{claim}
    There exist $\ell'_1,\ell'_2 \in \nat$ such that $m'_2 = \ell'_1+1+\ell'_2$, $u^{pm_1} \polrelp{k-1} u^{pm'_1}vu^{p\ell'_1}$ and $u^{pm_2} \polrelp{k-1} u^{p\ell'_2}$.
  \end{claim}

  \begin{proof}
    There are two cases depending on whether $m_2 \geq 2^{k}-1$ or not. Assume first that $m_2 \geq 2^{k}-1$. Since $m'_2 \geq 2^{k+1}-1$, we may choose $\ell'_1,\ell'_2 \geq 2^k-1$ such that $m'_2 = \ell'_1 + 1 + \ell'_2$. That $u^{pm_1} \polrelp{k-1} u^{pm'_1}vu^{p\ell'_1}$ follows from induction on $k$. Moreover, that $u^{pm'_2} \polrelp{k-1} u^{p\ell'_2}$ follows from Lemma~\ref{lem:hintro:propreo1}.

    Otherwise, $m_2 < 2^{k}-1$. We let $\ell'_2 = m_2$ and $\ell'_1 = m'_2 - 1 - \ell'_2$. Clearly, $\ell'_1 \geq 2^{k} -1$ since $m'_2 \geq 2^{k+1}-1$. Hence, we get $u^{pm_1} \polrelp{k-1} u^{pm'_1}vu^{p\ell'_1}$ from induction on $k$. Furthermore, $u^{pm_2} \polrelp{k-1} u^{p\ell'_2}$ is immediate since $m_2 = \ell'_2$ by definition.
  \end{proof}

  We may now finish the proof of Item~2. Let $\ell'_1,\ell'_2 \geq 1$ be as defined in the claim. We let $w'_1 = u^{pm'_1}vu^{p\ell'_1}v_1$ and $w'_2 = v_2u^{p\ell'_2}$. Clearly, $w'_1 aw'_2 = u^{pm'_1}vu^{pm'_2}$ since $v_1 av_2 = u^p$ and $m_2' = \ell'_1+1+\ell'_2$. Moreover, since \polrelp{k-1} is compatible with multiplication, we obtain:
  \[
    \begin{array}{rll}
      w_1 = u^{pm_1}v_1 & \polrelp{k-1} & u^{pm'_1}vu^{p\ell'_1}v_1 = w'_1, \\
      w_2 = v_2u^{pm_2} & \polrelp{k-1} & v_2u^{pm'_2} = w'_2.
    \end{array}
  \]
  This terminates the proof of Item~2.
\end{proof}

\section{Strictness of finitely based hierarchies}\label{sec:hintro:strictness}
As explained in the introduction, concatenation hierarchies first appeared in the literature with two specific hierarchies: the dot-depth was introduced in 1971~\cite{BrzoDot} and the Straubing-Thérien hierarchy ten years later~\cite{StrauConcat,TheConcat}. Although both of them were investigated intensively, their understanding is still far from being satisfactory. For instance, membership algorithms are known only for the lower levels in both hierarchies. A common feature to these two hierarchies is that their bases are \emph{finite}.

In this section, we look at \emph{finitely based} hierarchies. We prove that \emph{any} such hierarchy is strict for alphabets of size~2 or more, meaning that any half or full level contains strictly more languages than the preceding ones. Moreover, this holds as soon as the alphabet contains at least two letters. The condition that the basis is finite may seem to be very restrictive, but it is already very challenging and it covers the two classical cases (namely, the dot-depth and Straubing-Thérien hierarchies).

\begin{theorem}\label{thm:hintro:strict}
  Let \Cs be a finite \vari of languages. Then, the concatenation hierarchy of basis~\Cs is strict for any alphabet of size at least two.
\end{theorem}

The remainder of the section is devoted to proving Theorem~\ref{thm:hintro:strict}. Let us fix an alphabet $A$ containing at least two distinct letters $a$ and $b$. Our objective is to prove that for any finite \vari of (regular) languages \Cs, the associated concatenation hierarchy is strict, that is, for any~$n \in \nat$:
\[
  \Cs[n] \subsetneq \Cs[n+ \tfrac{1}{2}] \subsetneq \Cs[n+1].
\]
We prove this result as the corollary of a more general one. Let us first introduce some terminology that we require in order to state this result. We call \emph{unambiguous family} an infinite language $U \subseteq A^+$ satisfying the two following conditions:
\begin{enumerate}
\item For any $u \in U^+$, the decomposition $u = u_1 \cdots u_n$ with $u_1,\dots,u_n \in U$ witnessing membership in $U^+$ is \emph{unique}.
\item Moreover, if $v \in U$ is an infix of $u$, then $v = u_i$ for some $i \leq n$.
\end{enumerate}

\begin{example}\label{ex:hintro:unambig}
  A typical example of unambiguous family is $U = \{a{b^n}a \mid n \geq 1\}$. In fact, this is exactly the family that we use below to prove Theorem~\ref{thm:hintro:strict}.
\end{example}

Consider an arbitrary (possibly infinite) \vari of regular languages \Cs and an unambiguous family $U$. We say that $\Cs$ is \emph{non-separating for $U$} when there exist a language $L$ and $V \subseteq U$ satisfying the four following conditions:
\begin{equation}\label{eq:ns}
  \left\{
    \begin{aligned}
      &U \setminus V\text{ is infinite.}\\
      &A^*LA^* = L.\\
      &L \in \pol{\Cs}.\\
      &(A^* \setminus L) \cap V^+\text{ is {\bf not} \pol{\Cs}-separable from }L \cap V^+.
    \end{aligned}
  \right.
\end{equation}

We may now state our general result. Any concatenation hierarchy (even with an infinite basis) which is non-separating for some unambiguous family $U$, is strict.

\begin{proposition}\label{prop:hintro:genstrict}
  Let \Cs be a \vari of regular languages. Assume that there exists an unambiguous family $U \subseteq A^*$ such that \Cs is non-separating for $U$. Then, the concatenation hierarchy of basis \Cs over $A$ is strict.
\end{proposition}

We divide the proof in two parts: first, we explain how Proposition~\ref{prop:hintro:genstrict} can be used to prove Theorem~\ref{thm:hintro:strict}. Then, we present the proof of Proposition~\ref{prop:hintro:genstrict} itself.

\begin{proof}[Proof of Theorem~\ref{thm:hintro:strict}, assuming Proposition~\ref{prop:hintro:genstrict}]
  Our objective is to show that the concatenation hierarchy of basis \Cs is strict for $A$. We first prove that we may assume without loss of generality that for $\{\varepsilon\} \in \Cs$.

  \begin{lemma}\label{lem:hintro:epsilonok}
    There exists a finite \vari \Ds such that $\{\varepsilon\} \in \Ds$ and the concatenation hierarchy of basis \Cs is strict iff the one of basis \Ds is strict.
  \end{lemma}

  \begin{proof}
    We define $\Ds$ as the smallest Boolean algebra containing $\Cs$ and such that $\{\varepsilon\} \in \Ds$. By definition, \Ds is also finite and it is a Boolean algebra. Moreover, since quotients commute with Boolean operations, since \Cs is closed under quotient and since the only quotients of $\{\varepsilon\}$ are $\{\varepsilon\}$ and $\emptyset$, \Ds is closed under quotient as~well.

    It remains to verify that the concatenation hierarchy of basis \Cs is strict if and only if the one of basis \Ds is strict. We prove that for any $n \in \nat$, $\Cs[n] \subseteq \Ds[n] \subseteq \Cs[n+1]$. The result will then be immediate. By definition of concatenation hierarchies, it suffices to verify that these inclusions hold for $n = 0$, \emph{i.e.}, $\Cs \subseteq \Ds \subseteq \Cs[1]$. Clearly, we have $\Cs \subseteq \Ds$. For the other inclusion, we have $\Cs \subseteq \Cs[1]$ and $\Cs[1]$ is a Boolean algebra. Hence, by definition of \Ds, it suffices to prove that $\{\varepsilon\} \in \Cs[1]$ to conclude that $\Ds \subseteq \Cs[1]$. This is immediate, since $A^+ = \bigcup_{a \in A} A^*aA^* \in \Cs[\tfrac{1}{2}]$.
    Therefore, $\{\varepsilon\} = A^* \setminus A^+ \in \Cs[1]$, which terminates the proof.
  \end{proof}

  In view of Lemma~\ref{lem:hintro:epsilonok}, we now assume that $\{\varepsilon\} \in \Cs$. We first show how to use Proposition~\ref{prop:hintro:genstrict} to prove that the concatenation hierarchy of basis \Cs is strict. Let $U = \{a{b^n}a \mid n \geq 1\}$. Clearly, $U$ is unambiguous. If we prove that \Cs is non-separating for $U$, it will follow from Proposition~\ref{prop:hintro:genstrict} that the concatenation hierarchy of basis \Cs is strict. Our objective is therefore to exhibit $L\subseteq \{a,b\}^*$ and $V \subseteq U$ satisfying~\eqref{eq:ns}.

  Recall that since \Cs is a finite \vari, Lemma~\ref{lem:metho:omegapower} yields a period $p \geq 1$ such that for any $w \in A^*$ and any $m,m' \geq 1$,
  \[
    w^{pm} \leq_\Cs w^{pm'}.
  \]
  We define
  \[
    L = A^*ab^{2p}aA^*.
  \]
  Note that since $A^* \in \Cs$ (as \Cs is a \vari) and $\{\varepsilon\} \in \Cs$ by hypothesis, it is immediate from Lemma~\ref{lem:hintro:epsilon} that $L \in \pol{\Cs}$. Moreover, $A^*LA^* = L$ by definition. Finally, we define $V = \{a{b^p}a,a{b^{2p}}a\} \subseteq U$. Since $V$ is finite, $U \setminus V$ is infinite. It remains to show that $(A^* \setminus L) \cap V^+$ is not \pol{\Cs}-separable from $L \cap V^+$.

  We use our generic stratification for polynomial closures of finite classes. Since~\Cs is a finite \vari, we have defined a stratification of \pol{\Cs} (see Equation~\eqref{eq:strat-polc} p.~\pageref{eq:strat-polc}). For all $k \in \nat$, we let $\polrelk$ be the canonical preorder associated to the stratum \polk{\Cs}. By Corollary~\ref{cor:hintro:stratpol}, proving that $(A^* \setminus L) \cap V^+$ is not \pol{\Cs}-separable from $L \cap V^+$ amounts to showing that for any $k \in \nat$, there exist $u_k,v_k \in A^*$ such that $u_k \in (A^* \setminus L) \cap V^+$, $v_k \in L \cap V^+$ and $u_k \polrelk v_k$. For $k \in \nat$, we define,
  \[
    \begin{array}{lll}
      u_{k} & = & (ab^{p}a)^{p2^{k+1}},\\
      v_{k} & = & (ab^{p}a)^{p2^{k+1}} \cdot (ab^{2p}a)^{p} \cdot (ab^{p}a)^{p2^{k+1}}.\\
    \end{array}
  \]
  Clearly, by definition of $L$ and $V$, we have $u_k \in (A^* \setminus L) \cap V^+$ and $v_k \in L \cap V^+$. It remains to prove that $u_k \polrelk v_k$. Recall that we chose $p$ as the period of $\Cs$ given by Lemma~\ref{lem:metho:omegapower} for the \vari \Cs. In particular, it follows that $b^p \leq_\Cs b^{2p}$. Moreover, since \Cs is closed under quotient, it follows from Lemma~\ref{lem:metho:quotients} that $\leq_\Cs$ is a congruence and we conclude that,
  \[
    (ab^{p}a)^{p} \leq_\Cs (ab^{2p}a)^{p}.
  \]
  It is now immediate from Lemma~\ref{lem:hintro:propreo2} and the definition of $u_k$ and $v_k$ that $u_k \polrelk v_k$, which concludes the proof of Theorem~\ref{thm:hintro:strict}.
\end{proof}

\noindent
We now prove Proposition~\ref{prop:hintro:genstrict}, as a consequence of two lemmas that we present now.

\begin{lemma}\label{lem:hintro:genstrict1}
  Let \Cs be a \vari of regular languages. Assume that there exists an unambiguous family $U \subseteq A^*$ such that \Cs is non-separating for~$U$. Then, $\pol{\Cs}$ is not closed under complement.
\end{lemma}

\begin{lemma}\label{lem:hintro:genstrict2}
  Let \Cs be a \vari of regular languages. Assume that there exists an unambiguous family $U \subseteq A^*$ such that \Cs is non-separating for~$U$. Then, \bpol{\Cs} is non-separating for $U$.
\end{lemma}

\begin{proof}[Proof of Proposition~\ref{prop:hintro:genstrict} assuming Lemmas~\ref{lem:hintro:genstrict1} and~\ref{lem:hintro:genstrict2}]
  Combining Lemmas~\ref{lem:hintro:genstrict1} and~\ref{lem:hintro:genstrict2} yields that for any \vari of regular languages~\Cs which is non-separating for some unambiguous family, all half levels in the associated concatenation hierarchy are not closed under complement. Proposition~\ref{prop:hintro:notclosed} entails that the concatenation hierarchy of basis \Cs is strict over~$A$. Thus, Proposition~\ref{prop:hintro:genstrict} is proved.
\end{proof}

\noindent
To conclude the proof of Theorem~\ref{thm:hintro:strict}, it remains to prove Lemmas~\ref{lem:hintro:genstrict1} and~\ref{lem:hintro:genstrict2}.
\begin{proof}[Proof of Lemma~\ref{lem:hintro:genstrict1}]
  It follows from our hypothesis that we have $L \in \pol{\Cs}$, and $V\subseteq U$ such that $(A^* \setminus L) \cap V^+$ is {\bf not} \pol{\Cs}-separable from $L \cap V^+$. Observe that $A^* \setminus L$ clearly separates from $(A^* \setminus L) \cap V^+$ from $L \cap V^+$. Hence, $A^* \setminus L \not\in \pol{\Cs}$ by hypothesis. Since $L \in \pol{\Cs}$ by~\eqref{eq:ns}, we conclude that \pol{\Cs} is not closed under complement, which terminates the proof of Lemma~\ref{lem:hintro:genstrict1}.
\end{proof}

\noindent We turn to Lemma~\ref{lem:hintro:genstrict2} whose proof is more involved.

\begin{proof}[Proof of Lemma~\ref{lem:hintro:genstrict2}]
  By hypothesis, \Cs is non-separating for $U$, \emph{i.e.}, we have $L$ and $V \subseteq U$ satisfying~\eqref{eq:ns}. We need to prove that \bpol{\Cs} is non-separating for $U$ as well. By definition, this amounts to finding $K \in \pol{\bpol{\Cs}}$ and $W \subseteq U$ satisfying the appropriate properties. We first build $K$ and $W$. Since $U \setminus V$ is infinite, it is in particular nonempty, so that we may choose some word $w \in U \setminus V$. Let us define:
  \[
    \left\{
      \begin{array}{ll}
        K &= A^*w(A^+ \setminus L)wA^*,\\
        W &= V\cup \{w\}.
      \end{array}
    \right.
  \]
  Observe that since $U \setminus V$ is infinite, so is $U \setminus W$. Furthermore, $A^*KA^*=K$. Let us verify that $K \in \pol{\bpol{\Cs}}$. First, $A^+ = \bigcup_{a \in A}A^*aA^* \in \pol{\Cs}$, and since $L \in \pol{\Cs}$ as well, we have $A^+ \setminus L \in \bpol{\Cs}$.  Moreover, observe that $\{\varepsilon\}=A^* \setminus A^+\in \bpol{\Cs}$. Hence, Lemma~\ref{lem:hintro:epsilon} shows that $K$ belongs to $\pol{\bpol{\Cs}}$.

  \smallskip

  What remains to show is that $(A^* \setminus K) \cap W^+$ is {\bf not} \pol{\bpol{\Cs}}-separable from $K \cap W^+$ (see~\eqref{eq:ns}). We first define a stratification of \pol{\bpol{\Cs}} which we will use to prove this result. Recall that by Lemma~\ref{lem:hintro:elimbool}, we know that,
  \[
    \pol{\bpol{\Cs}} = \pol{\overline{\pol{\Cs}}}.
  \]
  Intuitively, we want to stratify $\pol{\overline{\pol{\Cs}}}$ with our generic stratification for polynomial closures. However, this is not possible since $\overline{\pol{\Cs}}$ may {\bf not} be finite. To solve this issue, we first consider an arbitrary stratification of $\overline{\pol{\Cs}}$.

  Let $\Ds = \overline{\pol{\Cs}}$. By definition, \Ds is a \pvari. Hence, Proposition~\ref{prop:metho:alwaystrat} yields a stratification of \Ds into finite \pvaris $\Ds_0,\dots,\Ds_k,\dots$. For any $k \in \nat$, we denote by $\preceq_k$ the canonical preorder associated to $\Ds_k$. Moreover, since all $\Ds_k$ are finite, for any $k \in \nat$, Lemma~\ref{lem:metho:omegapower} yields a period $q_k \geq 1$ for $\Ds_k$ such that for any  $w \in A^*$ and any $m,m' \geq 1$, we have
  \[
    w^{q_{k}m} \preceq_k w^{q_{k}m'}.
  \]
  Finally, since $(A^* \setminus L) \cap V^+$ is {\bf not} \pol{\Cs}-separable from $L \cap V^+$ by hypothesis, we get the following important fact about the relations $\preceq_k$.

  \begin{fact}\label{fct:hintro:induc}
    For any $k \in \nat$, there exist $x_k,y_k \in A^*$ such that $x_k \in L \cap V^+$, $y_k \in (A^* \setminus L) \cap V^+$ and $x_k \preceq_k y_k$.
  \end{fact}

  \begin{proof}
    Since $(A^* \setminus L) \cap V^+$ is not \pol{\Cs}-separable from $L \cap V^+$, it follows that $L \cap V^+$ is not $\overline{\pol{\Cs}}$-separable from $(A^* \setminus L) \cap V^+$. Since $\Ds = \overline{\pol{\Cs}}$, the fact is now immediate from Corollary~\ref{cor:metho:stratcor}.
  \end{proof}

  We are now ready to stratify $\pol{\bpol{\Cs}} = \pol{\overline{\pol{\Cs}}}$. For all $k \in \nat$ we consider the class $\polk{\Ds_k}$ (\emph{i.e.}, the stratum $k$ in our generic stratification of $\pol{\Ds_k}$). Since the classes $\Ds_k$ are \pvaris, the classes $\polk{\Ds_k}$ are all \pvaris as well (by Lemma~\ref{lem:hintro:stratl1vari}). Moreover, since $\Ds = \overline{\pol{\Cs}}$, one may verify that,
  \[
    \text{For all $k \in \nat$, } \polk{\Ds_k} \subseteq \polp{\Ds_{k+1}}{k+1} \quad \text{and} \quad \pol{\overline{\pol{\Cs}}} = \bigcup_{k \in \nat} \polk{\Ds_k}.
  \]
  In summary, we now have a stratification of $\pol{\bpol{\Cs}} = \pol{\overline{\pol{\Cs}}}$. For any $k \in \nat$, we denote by $\polrelk$ the canonical preorder associated to $\polk{\Ds_k}$.

  Recall that our goal is to prove that $(A^* \setminus K) \cap W^+$ is {\bf not} \pol{\bpol{\Cs}}-separable from $K \cap W^+$. Now that we have a stratification of \pol{\bpol{\Cs}}, we know from Corollary~\ref{cor:metho:stratcor} that this amounts to showing that for any $k \in \nat$, there exist $u_k,v_k \in A^*$ such that $u_k \in (A^* \setminus K) \cap W^+$, $v_k \in K \cap W^+$ and $u_k \polrelk v_k$. Let $k \in \nat$, we define,
  \[
    \begin{array}{lll}
      u_{k} & = & (w \cdot x_k \cdot w)^{q_k2^{k+1}}, \\
      v_{k} & = & (w \cdot x_k \cdot w)^{q_k2^{k+1}}\cdot (w \cdot y_k \cdot w)^{q_k} \cdot (w \cdot x_k \cdot w)^{q_k2^{k+1}}.
    \end{array}
  \]
  It remains to prove that $u_k$ and $v_k$ satisfy the appropriate properties. We begin with $u_k \in (A^* \setminus K) \cap W^+$. Note that this is where we use the fact that $U$ is an unambiguous family.
  Since $w \in W$ and $x_k \in V^+ \subseteq W^+$, it follows that $u_k \in W^+$. It remains to prove that $u_k \in A^* \setminus K$, \emph{i.e.}, $u_k \not\in K$. Since $K = A^*w (A^+ \setminus L) w A^*$, this amounts to proving that for any infix of the form $wzw$ in $u_k$ with $z\not=\varepsilon$, we have $z \in L$. Consider such an infix. Since $w \in U$, $u_k\in U^+$ and $U$ is unambiguous, it is immediate from the definition of $u_k$ that $z$ must
  contain $x_k \in L$ as an infix. Since $L = A^*LA^*$ by hypothesis, it follows that $z \in L$.

  \smallskip

  We now prove that $v_k \in K \cap W^+$. Since $w\in W$ and $x_k,y_k\in V^+\subseteq W^+$, we have indeed $v_k\in W^+$. Furthermore, since $y_k \in A^* \setminus L\cap V^+\subseteq A^+\setminus L$ by definition, we have $v_k\in K = A^* w (A^+ \setminus L) w A^*$.  We conclude that $v_k \in K \cap W^+$ as desired.

  \smallskip

  We finish with $u_k \polrelk v_k$. Recall that by definition, $x_k \preceq_k y_k$. Moreover, since $\Ds_k$ is a \pvari, $\preceq_k$ is compatible with concatenation by Lemma~\ref{lem:metho:quotients} and so:
  \[
    (w \cdot x_k \cdot w)^{q_k} \preceq_k (w \cdot y_k \cdot w)^{q_k}
  \]
  Since $\polrelk$ is the canonical preorder associated to \polk{\Ds_k}, it follows from Lemma~\ref{lem:hintro:propreo2} and our choice of $q_k$ that $u_k \polrelk v_k$, finishing the proof of Lemma~\ref{lem:hintro:genstrict2}.
\end{proof}

\section{Membership and separation}\label{sec:memb-separ-conc}
Now that we know that finitely based concatenation hierarchies are strict, we focus on membership and separation for each individual level in such a hierarchy. We present an exhaustive and \emph{generic} state of the art regarding these problems in this section and the following. However, note that presenting the algorithms themselves would require introducing too much material. For this reason, we shall simply state the problems which are known to be decidable, without describing the algorithms. Both problems are unexpectedly hard, despite their straightforward formulations. The overall state of the art consists in only three theorems. We state two of them in this section and the last one in Section~\ref{sec:two-fund-conc}.

\noindent
The first result is that separation (hence also membership) is decidable up to level~$\frac32$.

\begin{theorem}[Place \& Zeitoun~\cite{pzboolpol}]\label{thm:sep:hiera}
  If \Cs is a finite \vari, separation and membership are decidable for the following classes:
  \begin{enumerate}
  \item \pol\Cs,
  \item \bpol\Cs,
  \item \pol{\bpol\Cs}.
  \end{enumerate}
  In other words, in any finitely based concatenation hierarchy, levels $\frac12$, $1$ and $\frac32$ have decidable separation.
\end{theorem}

Theorem~\ref{thm:sep:hiera} applies in particular to the dot-depth and the Straubing-Thérien hierarchies, since their bases are finite. Therefore, several classical results that we presented in Section~\ref{sec:introduction}, namely Theorems~\ref{thm:knast}, \ref{thm:dd32}, \ref{thm:st1}, \ref{thm:st32} and \ref{thm:ddredst}, are immediate corollaries of Theorem~\ref{thm:sep:hiera}. Note however that we do not recover Theorem~\ref {thm:5272} yet. Nevertheless, we will see in the next section that this result is also a corollary of Theorem~\ref{thm:sep:hiera}.

\begin{remark}\label{rem:covering}
  Theorem~\ref{thm:sep:hiera} is obtained by investigating a more general problem, \emph{covering}~
  \cite{pzcoveringfull,pzcovering}: any finitely based concatenation hierarchy has decidable covering. It turns out that algorithms for separation are byproducts of covering algorithms.
\end{remark}
The second generic result reduces membership for \pol\Cs to separation for~\Cs.

\begin{theorem}[Place \& Zeitoun~\cite{pzboolpol}]\label{thm:sep:transfer}
  For any quotienting lattice \Cs, $\pol\Cs$-membership reduces to \Cs-separation.
\end{theorem}

This result has the following corollary:

\begin{corollary}
  In any concatenation hierarchy of basis~\Cs:
  \begin{enumerate}
  \item If level $n$ has decidable separation, then level $n+\frac12$ has decidable membership.
  \item If level $n-\frac12$ has decidable separation, then level $n+\frac12$ has decidable membership.
  \end{enumerate}
  In particular, if \Cs is finite, then level $\frac52$ has decidable membership.
\end{corollary}

\begin{proof}
  By Proposition~\ref{prop:hintro:boolc} and Theorem~\ref{thm:hintro:polc}, all levels in the concatenation hierarchy of basis \Cs are \pvaris. Therefore, we may apply Theorem~\ref{thm:sep:transfer} to any such level. The first item in the corollary comes directly from the definition of level $n+\frac12$. For Item~2, let $\Ds=\Cs[n-1]$ and assume that $\pol\Ds=\Cs[n-\frac12]$ has decidable separation. Observe that $\overline{\pol\Ds}$ has also decidable separation, since $(K,L)$ are \pol\Ds-separable iff $(L,K)$ are $\overline{\pol\Ds}$-separable. Moreover, by Lemma~\ref{lem:hintro:elimbool}, we have $\Cs[n+\frac12]=\pol{\bpol{\Ds}} = \pol{\overline{\pol{\Ds}}}$. Therefore, it suffices to apply Theorem~\ref{thm:sep:transfer} to conclude the proof of Item~2. Finally, it follows from Item~3 in Theorem~\ref{thm:sep:hiera} that $\Cs[\frac52]$ has decidable membership if \Cs is finite.
\end{proof}

\section{Two fundamental concatenation hierarchies}\label{sec:two-fund-conc}
This section is devoted to the dot-depth and Straubing-Thérien hierarchies. The dot-depth is the concatenation hierarchy whose basis is:
\[
  \dotdp{0} = \{\emptyset,\{\varepsilon\},A^+,A^*\},
\]
while the Straubing-Thérien hierarchy is the concatenation hierarchy whose basis is:
\[
  \stzer = \{\emptyset,A^*\}.
\]
For $q \in \nat$ or $q\in \frac12+\nat$, we denote by \dotdp{q} level $q$ in the dot-depth hierarchy, and by $\sttp{q}$ level $q$ in the Straubing-Th\'erien hierarchy.
It is easy to see both hierarchies classify the star-free languages:
\[
  \sfr  = \bigcup_{n \geq 0} \dotdp{n} = \bigcup_{n \geq 0} \sttp{n}.
\]
This was the original motivation of Brzozowski and Cohen for considering the dot-depth hierarchy: for each language, one counts the number of alternations between complement and concatenation that are required to define it.

The Straubing-Thérien hierarchy quickly gained attention in the mid 80s after it was observed to be ``more fundamental'' than the dot-depth. This informal claim is motivated by the two following properties:
\begin{enumerate}
\item Straubing~\cite{StrauVD} showed that any full level in the dot-depth hierarchy is obtained by applying a uniform operation to the corresponding level in the Straubing-Thérien hierarchy. This result makes it possible to reduce membership for a level in the dot-depth to the same problem for corresponding level the Straubing-Thérien hierarchy. This was pushed later to half-levels~\cite{PinWeilVD} and to separation~\cite{pzsucc,pzsuccfull}.

\item An important result that we already stated in Section~\ref{sec:memb-separ-conc} is that separation is decidable up to level~$\frac32$ in any hierarchy with a finite basis. In the particular case of the Straubing-Thérien hierarchy, it follows from a theorem of Pin and Straubing~\cite{pin-straubing:upper} that the levels $\frac{3}{2}$ and above are also the levels $\frac{1}{2}$ and above in another hierarchy whose basis is also finite. While simple, this result is crucial, as it allows us to lift the decidability results from Section~\ref{sec:memb-separ-conc} up to level $\frac{5}{2}$ in the Straubing-Thérien~hierarchy (and therefore in the dot-depth as well by the first item above).
\end{enumerate}

Thanks to the generic analysis carried out in previous  sections, we know that levels of both hierarchies satisfy robust properties: since their bases are \varis, it follows from Proposition~\ref{prop:hintro:concatvari} that all half-levels are \pvaris and all full levels are \varis.
Moreover, it follows from Proposition~\ref{prop:hintro:closconcat} that all half-levels are closed under concatenation and marked concatenation.

In fact the Straubing-Thérien hierarchy is closely related to the dot-depth: the two hierarchies are interleaved as we state in the next proposition.
\begin{proposition}\label{prop:hintro:inclus}
  For any $n \in \nat$, the following inclusions hold:
  \[
    \sttp{n} \subseteq \dotdp{n} \subseteq \sttp{n+1} \quad \text{and} \quad \sttp{n + \tfrac{1}{2}} \subseteq \dotdp{n + \tfrac{1}{2}} \subseteq \sttp{n + \tfrac{3}{2}}.
  \]
\end{proposition}

\begin{proof}
  The inclusions $\sttp{n} \subseteq \dotdp{n}$ and $\sttp{n + \tfrac{1}{2}} \subseteq \dotdp{n + \tfrac{1}{2}}$ are immediate since it is clear that $\stzer \subseteq \dotzer$. For the other inclusions it suffices to observe that $\dotzer \subseteq \stone$. This holds since $\dotzer = \{\emptyset,\{\varepsilon\},A^+,A^*\}$, and we have $\emptyset,A^* \in \stzer\subseteq\stone$ and $A^+ = \bigcup_{a \in A} A^*aA^* \in \sthone \subseteq \stone$. Finally, $\{\varepsilon\} = A^* \setminus A^+ \in \stone$.
\end{proof}

Theorem~\ref{thm:hintro:strict} shows that over an alphabet of size at least 2, the dot-depth hierarchy is strict (as its basis \dotdp{0} is finite). Thus, Theorem~\ref{thm:hintro:ddstrict} is a simple corollary of Theorem~\ref{thm:hintro:strict}.

An immediate consequence of Proposition~\ref{prop:hintro:inclus} is the strictness of the Straubing-Thérien hierarchy, which follows from the strictness of the dot-depth (Theorem~\ref{thm:hintro:ddstrict}). Of course, this is also a consequence of Theorem~\ref{thm:hintro:strict} since the basis \stzer is finite.

\begin{remark}
  We proved Theorem~\ref{thm:hintro:ddstrict} for the dot-depth as the corollary of a more general theorem. However, there exist many specific proofs. This includes the original one by Brzozowski and Knast~\cite{BroKnaStrict} who exhibit languages $L_n$ for all $n \geq 1$ such that $L_n$ has dot-depth $n$ but not dot-depth $n-1$. The definition is as follows $L_1 = (ab)^*$, and for $n \geq 2$, $L_n = (aL_{n-1}b)^*$. As expected, the difficulty is proving that $L_n$ does \emph{not} have dot-depth $n-1$. This proof has often been revisited. Let us mention the game theoretic proofs of Thomas~\cite{ThomStrict2,ThomStrict} or the one by Thérien~\cite{therien:powersurvey}.

  Let us also mention the proof of Straubing~\cite{StrauConcat} which relies on a different approach based on algebra. Instead of working with the classes \dotdp{n}, this argument proves strict inclusions between associated algebraic classes (namely, semigroup varieties).
\end{remark}

\medskip\noindent\textbf{Examples.}
Let us present some typical examples of languages that belong or not to the first levels of the dot-depth hierarchy.
\begin{example}[Languages of dot-depth 1/2]
  Let $\Cs=\dotdp{0}$. Observe that $w \sim_\Cs w'$ if and only if $w,w'$ are either both empty or both nonempty. In particular, the period of \Cs is $p = 1$ for any alphabet.
Let $A = \{a,b\}$ and consider $\dotdp{\frac12}= \pol{\Cs}$. We show that the language $L = a^*b^*$ does not belong to $\dotdp{\frac12}$. Indeed, for any $k \in \nat$, consider the two following words:
  \[
    u_k = a^{2^{k+1}}b \quad \text{and} \quad v_k = a^{2^{k+1}}b^{2^{k+1}}a^{2^{k+1}}b.
  \]
  Clearly, we have $u_k \in L$ and $v_k \not\in L$ for any $k \in \nat$. Since the period of \dotdp0 is $1$ and $a \sim_\Cs b^{2^{k+1}}$ (both words are nonempty), we obtain from Lemma~\ref{lem:hintro:propreo2} that,
  \[
    u_k \polrelk v_k \quad \text{where $\polrelk$ is the canonical preorder associated to \polk{\Cs}}.
  \]
  Thus $L$ is not definable in $\pol\Cs=\dotdp{\frac12}$.
\end{example}

\begin{example}[Languages of dot-depth one]
  Consider the alphabet $A = \{a,b\}$. The typical example of language having dot-depth one is $(ab)^*$. Indeed, we have $(ab)^* = A^* \setminus (bA^* \cup A^*a \cup A^*aaA^* \cup A^*bbA^*)$. Hence, $(ab)^*$ has dot-depth one.
\end{example}

\begin{example} \label{ex:hintro:dd1}
  Let again $\Cs=\dotdp0$.  We prove that the language $K = (a(ab)^*b)^*$ does not belong to $\dotdp1=\bpol{\Cs}$. For any $k \in \nat$ consider the three following words,
  \[
    w_k = (ab)^{2^{k+1}} \qquad x_k = (aw_kbw_k)^{2^{k+1}} \qquad y_k = (aw_kaw_kbw_k)^{2^{k+1}}.
  \]
  Observe that for any $k \in \nat$, $x_k \in K$ and $y_k \not\in K$. We now prove that $x_k \bpolrelk y_k$, where $\bpolrelk$ is the canonical preorder relation associated to \bpolk{\Cs} (because \bpolk{\Cs} is a Boolean algebra, one may verify $\bpolrelk$ is in fact an equivalence relation). It will then follow that $K \not\in \bpol{\Cs}$.

  This amounts to proving that $x_k \polrelk y_k$ and $y_k \polrelk x_k$. Observe that $ab \sim_\Cs a \sim_\Cs b$ (all these words are nonempty). By definition of $w_k$, we obtain from Lemma~\ref{lem:hintro:propreo2} that,
  \begin{equation} \label{eq:hintro:notdd1}
    \begin{array}{lll}
      w_k \polrelk w_kaw_k, \\
      w_k \polrelk w_kbw_k.
    \end{array}
  \end{equation}
  Since \polrelk is compatible with concatenation, it immediately follows from the first inequality in~\eqref{eq:hintro:notdd1} that $x_k \polrelk y_k$. Indeed, we have,
  \[
    x_k = (aw_kbw_k)^{2^{k+1}} \polrelk (aw_kaw_kbw_k)^{2^{k+1}} = y_k.
  \]
  Conversely, observe that using compatibility with concatenation again and the second item in~\eqref{eq:hintro:notdd1}, we get,
  \[
    aw_kaw_kbw_k \polrelk aw_kbw_kaw_kbw_k = (aw_kbw_k)^2.
  \]
  Therefore, we have $y_k \polrelk (aw_kbw_k)^{2 \times 2^{k+1}}$. Finally, using Lemma~\ref{lem:hintro:propreo1}, we obtain:
  \[
    (aw_kbw_k)^{2 \times 2^{k+1}} \polrelk (aw_kbw_k)^{2^{k+1}} = x_k.
  \]
  Hence, we conclude that $y_k \polrelk x_k$ which terminates the proof.
\end{example}

\begin{example}
  Let $A = \{a,b\}$. One can verify that $(ab)^*$ does not belong to \stone. However, it belongs to \sttwo. Indeed, observe that the singleton language $\{\varepsilon\}$ belongs to \stone ($\{\varepsilon\} = A^* \setminus \left(\bigcup_{a \in A} A^*aA^*\right)$). Thus, $bA^*,A^*a,A^*aaA^*$ and $A^*bbA^*$ belong to \sthtwo and we may use the usual approach: $(ab)^*$ is the complement of $bA^*\cup A^*a \cup A^*aaA^* \cup A^*bbA^*$ and therefore belongs to \sttwo.
\end{example}

\begin{example}[Languages of dot-depth two]
  Consider the alphabet $A = \{a,b\}$. The language $(a(ab)^*b)^*$ has dot-depth two. Indeed, one may verify that it is the complement of the following language:
  \[
    (ab)^*bA^* + A^*aa(ba)^*aA^* + A^*b(ba)^*bbA^* + A^*a(ab)^*.
  \]
  Clearly, the above language has dot-depth $\frac{3}{2}$ since $(ab)^*$ and $(ba)^*$ have dot-depth one by the previous example. Hence, $(a(ab)^*b)^*$ has dot-depth two.
\end{example}

\medskip\noindent{\textbf{Membership and Separation.}} Theorem~\ref{thm:ddredst} shows that if membership of some level in the Straubing-Thérien hierarchy is decidable, then so is the corresponding level in the dot-depth. Actually, the state of the art with respect to membership is the same for the Straubing-Thérien hierarchy as the one for the dot-depth. In~\cite{pzsucc,pzsuccfull}, we generalized Theorem~\ref{thm:ddredst} to cope with separation and covering, by defining a generic operation on classes of languages, that maps each level of the Straubing-Thérien hierarchy on the corresponding level in the dot-depth hierarchy, and which preserves decidability of covering and separation. We outline the state of the art in \figurename~\ref{fig:hintro:strauther}.

\begin{theorem}[Place \& Zeitoun~\cite{pzsucc,pzsuccfull}]\label{thm:enrichment}
  For any level $q\in\nat$ or $q\in\frac12+\nat$, if\/$\sttp{q}$ has decidable separation (resp.\ covering), then so has $\dotdp{q}$.
\end{theorem}

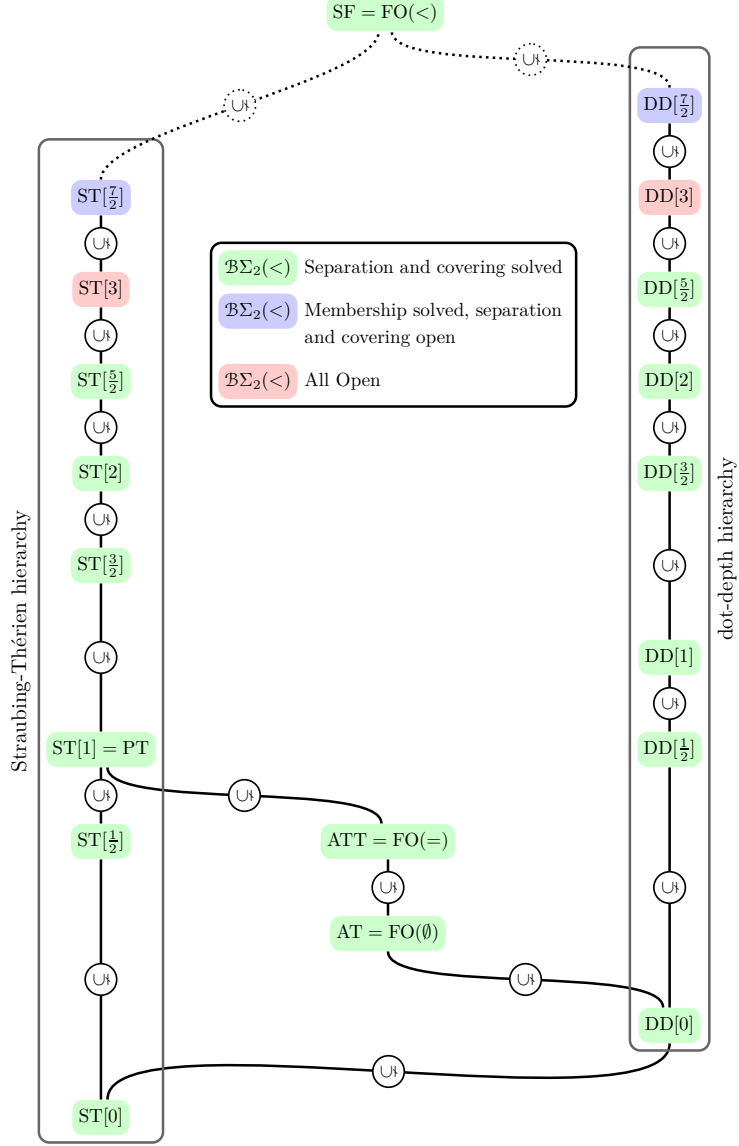
\begin{figure}[!ht]
  \centering
  \scalebox{.8}{
    \begin{tikzpicture}

      \def\step{1.7}
      \def\separ{1.5}
      \def\scal{0.9}
      \def\decal{0.5}

      \begin{scope}[scale=\scal,every node/.style={scale=\scal}]

        \node[gbox] (b0w) at (-3.5,-1*\step) {\stzer};
        \node[gbox] (s1w) at ($(b0w)+(0.0,3*\step)$) {\sthone};
        \node[gbox] (b1w) at ($(b0w)+(0.0,4*\step)$) {$\stone = \pt$};

        \node[gbox] (s2w) at ($(b1w)+(0.0,2*\step)$) {\sthtwo};
        \node[gbox] (b2w) at ($(b1w)+(0.0,3*\step)$) {\sttwo};

        \node[gbox] (s3w) at ($(b2w)+(0.0,\step)$) {\sththree};
        \node[rbox] (b3w) at ($(b2w)+(0.0,2*\step)$) {\stthree};

        \node[bbox] (s4w) at ($(b3w)+(0.0,\step)$) {\sthfour};

        \node[gbox] (folab) at (1.8,\step) {$\at = \folab$};
        \node[gbox] (foeq) at ($(folab)+(0.0,\step)$) {$\att = \foeq$};
        \draw[ledg] (folab) to node[linc] {$\subsetneq$} (foeq);

        \draw[ledg] (foeq.110) to [out=90,in=-90,looseness=0.4] node[linc] {$\subsetneq$} (b1w.-70);

        \node[gbox] (b0ws) at (7.0,0) {\dotzer};
        \node[gbox] (s1ws) at ($(b0ws)+(0.0,3*\step)$) {\dothone};
        \node[gbox] (b1ws) at ($(b0ws)+(0.0,4*\step)$) {\dotone};

        \node[gbox] (s2ws) at ($(b1ws)+(0.0,2*\step)$) {\dothtwo};
        \node[gbox] (b2ws) at ($(b1ws)+(0.0,3*\step)$) {\dottwo};

        \node[gbox] (s3ws) at ($(b2ws)+(0.0,\step)$) {\doththree};
        \node[rbox] (b3ws) at ($(b2ws)+(0.0,2*\step)$) {\dotthree};

        \node[bbox] (s4ws) at ($(b3ws)+(0.0,\step)$) {\dothfour};

        \draw[ledg] (b0w.north) to [out=90,in=-90,looseness=0.4] node[linc] {$\subsetneq$} (s1w);

        \draw[ledg] (b0w.70) to [out=90,in=-90,looseness=0.4] node[linc] {$\subsetneq$} (b0ws.south);

        \draw[ledg] (s1w.north) to [out=90,in=-90,looseness=0.4] node[linc] {$\subsetneq$} (b1w);

        \draw[ledg] (b1w.north) to [out=90,in=-90,looseness=0.4] node[linc] {$\subsetneq$} (s2w.south);

        \draw[ledg] (s2w.north) to [out=90,in=-90,looseness=0.4] node[linc] {$\subsetneq$} (b2w.south);

        \draw[ledg] (b2w.north) to [out=90,in=-90,looseness=0.4] node[linc] {$\subsetneq$} (s3w.south);
        \draw[ledg] (s3w.north) to [out=90,in=-90,looseness=0.4] node[linc] {$\subsetneq$} (b3w.south);
        \draw[ledg] (b3w.north) to [out=90,in=-90,looseness=0.4] node[linc] {$\subsetneq$} (s4w.south);

        \draw[ledg] (b0ws.north) to [out=90,in=-90,looseness=0.4] node[linc] {$\subsetneq$} (s1ws);

        \draw[ledg] (b0ws.110) to [out=90,in=-90,looseness=0.4] node[linc] {$\subsetneq$} (folab.south);

        \draw[ledg] (s1ws.north) to [out=90,in=-90,looseness=0.4] node[linc] {$\subsetneq$} (b1ws);
        \draw[ledg] (b1ws.north) to [out=90,in=-90,looseness=0.4] node[linc] {$\subsetneq$} (s2ws);

        \draw[ledg] (s2ws.north) to [out=90,in=-90,looseness=0.4] node[linc]
        {$\subsetneq$} (b2ws.south);

        \draw[ledg] (b2ws.north) to [out=90,in=-90,looseness=0.4] node[linc]
        {$\subsetneq$} (s3ws.south);
        \draw[ledg] (s3ws.north) to [out=90,in=-90,looseness=0.4] node[linc]
        {$\subsetneq$} (b3ws.south);
        \draw[ledg] (b3ws.north) to [out=90,in=-90,looseness=0.4] node[linc]
        {$\subsetneq$} (s4ws.south);
      \end{scope}

      \node[lbox,fit={(b0ws) (s1ws) (s2ws) ($(s4ws)+(0.0,0.5*\step)$)}] (bobox) {};

      \node[anchor=north,rotate=90,align=center] (fosbox) at (bobox.east)
      {dot-depth hierarchy};

      \node[lbox,fit={(b0w) (b1w) (s2w) ($(s4w)+(0.0,0.5*\step)$)}] (bobox2) {};

      \node[anchor=south,rotate=90,align=center] (fosbox) at (bobox2.west)
      {Straubing-Thérien hierarchy};

      \begin{scope}[scale=\scal,every node/.style={scale=\scal}]
        \node[gbox] (fow) at (1.75,11*\step) {$\sfr = \fow$};

        \draw[ledg,dotted] (s4w.north) to [out=90,in=-90,looseness=0.4] node[linc,dotted]
        {$\subsetneq$} (fow.-110);

        \draw[ledg,dotted] (s4ws.north) to [out=90,in=-90,looseness=0.4] node[linc,dotted]
        {$\subsetneq$} (fow.-70);
      \end{scope}

      \begin{scope}[scale=\scal,every node/.style={scale=\scal}]

        \node[gbox,text opacity=0] (gr) at (-0.6,14.0) {\bswd};
        \node[bbox,text opacity=0] (bl) at (-0.6,13.2) {\bswd};
        \node[rbox,text opacity=0] (re) at (-0.6,11.9) {\bswd};
        \node[anchor=west] (gr2) at (gr.east) {Separation and covering solved};
        \node[anchor=west] (bl2) at (bl.east) {Membership solved, separation};

        \node[anchor=west] at ($(bl.east)-(0,0.5)$) {and covering open};

        \node[anchor=west] (re2) at (re.east) {All Open};

      \end{scope}

      \node[lbox,draw=black,fit={(gr) (re) (gr2) (bl2)}] {};

    \end{tikzpicture}
  }
  \caption{Overview of classes. For the sake of avoiding clutter, inclusions between levels in the Straubing-Thérien and dot-depth hierarchies are omitted (see Proposition~\ref{prop:hintro:inclus}).}
 \label{fig:hintro:strauther}
\end{figure}

\medskip\noindent\textbf{The alphabet trick.} We now connect the Straubing-Thérien hierarchy with the concatenation hierarchy whose basis is the class \at of alphabet testable languages. While simple, this theorem has important consequences.

Let us briefly recall the definition of the alphabet testable languages. We shall connect two classes to the Straubing-Thérien hierarchy: \at itself and a weaker one which we denote by \wat.  For any alphabet $A$, recall that $\at$ consists of all Boolean combinations of languages of the form,
$A^*aA^*$ for $a \in A$.
Moreover, $\wat$ contains all unions of languages $B^*$ for $B \subseteq A$. We already know that \at is a \vari of regular languages and one may verify that \wat is a \pvari (closure under intersection follows from the fact that $B^* \cap C^* = (B \cap C)^*$).

It was proved by Pin and Straubing~\cite{pin-straubing:upper} that the level $\frac{3}{2}$ in the Straubing-Thérien hierarchy\footnote{In fact, the original formulation of Pin and Straubing considers level $2$ in the Straubing-Thérien hierarchy and not level $\frac{3}{2}$.} is also the class \pol{\wat}.

\begin{remark}
  Another formulation (which is essentially the original one of Pin and Straubing) is to say that \sthtwo contains exactly the unions of languages of the form,
  \[
    B_0^*a_1B_1^*a_2B_2^* \cdots a_n B_n^* \quad \text{with $B_0,\dots,B_n \subseteq A$}. \dropQED
  \]
\end{remark}

We reformulate this result in the following crucial theorem.

\begin{theorem}[Pin and Straubing~\cite{pin-straubing:upper}]\label{thm:hintro:alphatrick}
  Level $\frac{3}{2}$ in the Straubing-Thérien hierarchy satisfies the following property:
  \[
    \sttp{\tfrac{3}{2}} = \pol{\wat} = \pol{\at}.
  \]
  In particular, any level $n \geq \frac{3}{2}$ (half or full) in the Straubing-Thérien hierarchy corresponds exactly to level $n-1$ in the concatenation hierarchy of basis \at.
\end{theorem}

The important point here is that while \at is more involved than \stzer as a basis, it remains \emph{finite}. Therefore, Theorem~\ref{thm:hintro:alphatrick} states that any level $n \geq \frac{3}{2}$ in the Straubing-Thérien hierarchy is also level $n-1$ in another hierarchy whose basis remains finite. This result implies that we ``gain'' one level for the decidability results, therefore yielding Theorem~\ref{thm:5272}.

Indeed, we know that separation is decidable for levels $\frac{1}{2}$, $1$ and $\frac{3}{2}$ of \emph{any} concatenation hierarchy whose basis is finite. This of course applies to the Straubing-Thérien hierarchy since \stzer is clearly finite. However, Theorem~\ref{thm:hintro:alphatrick} allows us to go one step further and to lift these results to levels $2$ and $\frac{5}{2}$ in the particular case of the Straubing-Thérien hierarchy. Indeed, they are also levels $1$ and $\frac{3}{2}$ in the hierarchy of finite basis \at.

\medskip

We now prove Theorem~\ref{thm:hintro:alphatrick}. Since it clear that $\wat \subseteq \at$, the inclusion $\pol{\wat} \subseteq \pol{\at}$ is trivial. We show that $\pol{\at} \subseteq \sthtwo \subseteq \pol{\wat}$.

The inclusion $\pol{\at} \subseteq \sthtwo$ is simple. Indeed, we know from the definition that $\sthtwo = \pol{\stone}$. Hence, it suffices to prove that $\at \subseteq \stone$. Let $L \in \at$. By definition, $L$ is the Boolean combination of languages $A^*aA^*$ for some $a \in A$. Clearly, $A^*aA^* \in \sthone$ for any $a \in A$. Hence, $L \in \stone = \bool{\sthone}$.

The inclusion $\sthtwo \subseteq \pol{\wat}$ is more involved. We first reduce the proof to that of a simpler inclusion. Recall that we showed in Proposition~\ref{prop:hintro:bypassfull} that,
\[
  \sttp{\tfrac{3}{2}} = \pol{\overline{\sttp{\tfrac{1}{2}}}}
\]
Therefore, in order to prove that $\sthtwo \subseteq \pol{\wat}$, it suffices to show the following inclusion:
\begin{equation}\label{eq:hintro:inclus}
  \overline{\sttp{\tfrac{1}{2}}} \subseteq \pol{\wat}
\end{equation}
Indeed, it will follow that $\pol{\overline{\sttp{\tfrac{1}{2}}}} \subseteq \pol{\pol{\wat}} = \pol{\wat}$ since the polynomial closure operation is idempotent by Lemma~\ref{lem:hintro:polidem}. We now concentrate on proving~\eqref{eq:hintro:inclus}. This is a consequence of the following lemma.

\begin{lemma}\label{lem:hintro:atmainlem}
  For any $a_1,\dots,a_n \in A$, we have $A^* \setminus A^*a_1A^* \cdots a_n A^* \in \pol{\wat}$.
\end{lemma}

Before we prove Lemma~\ref{lem:hintro:atmainlem}, let us use it to show that the inclusion~\ref{eq:hintro:inclus} holds. By definition, any language $L \in  \overline{\sthone}$ is the complement of another language in the class $\sthone = \pol{\stzer}$. Hence, by definition of polynomial closure, there exist \stzer-monomials $K_1,\dots,K_m$ such that,
\[
  L = A^* \setminus \left(\bigcup_{i \leq m} K_i\right) = \bigcap_{i \leq m} A^* \setminus K_i
\]
Since $\stzer = \{\emptyset,A^*\}$ all \stzer-monomials $K_1,\dots,K_m$ are of the form $A^*a_1A^* \cdots a_n A^*$. Thus, it follows from Lemma~\ref{lem:hintro:atmainlem} that $A^* \setminus K_i \in \pol{\wat}$ for all $i \leq m$. Finally, since \pol{\wat} is closed under intersection, we conclude that $L \in \pol{\wat}$ which terminates the proof.

\begin{remark}
  While the above argument may seem simple, let us point out that we implicitly used Theorem~\ref{thm:hintro:polc} which is an involved result. On one hand, we used the original definition of polynomial closures for \pol{\stzer} (\emph{i.e.}, it contains the unions of \stzer-monomials). On the other hand, we used the fact that \pol{\wat} is closed under intersection which is not immediate from the definition: this is where we need Theorem~\ref{thm:hintro:polc}.
\end{remark}

It remains to prove Lemma~\ref{lem:hintro:atmainlem}. Consider $n$ letters $a_1,\dots,a_n \in A$. By a \emph{piece} of a word, we mean a scattered subword. Our objective is to show that $A^* \setminus A^*a_1A^* \cdots a_n A^* \in \pol{\wat}$. For all $k \leq n$, we let $L_k = A^* \setminus A^*a_1A^* \cdots a_kA^*$. Observe that by definition, for any $k \leq n$, $L_k$ contains all words $w$ such that $a_1 \cdots a_k$ is {\bf not} a piece of $w$. We prove by induction on $k$ that $L_k \in \pol{\wat}$ for all $k \leq n$.

When $k = 1$, this is immediate since $L_1 = (A \setminus \{a_1\})^*$ which belongs to \wat (and therefore to \pol{\wat}) by definition. We now assume that $k \geq 2$. Consider the following language $H_k$:
\[
  H_k = (A \setminus \{a_k\})^* \cup L_{k-1}a_k(A \setminus \{a_k\})^*.
\]
By induction hypothesis, we have $L_{k-1} \in \pol{\wat}$. Moreover, it is immediate from the definition of \wat that  $(A \setminus \{a_k\})^* \in \wat \subseteq \pol{\wat}$. Hence, we conclude that $H_k \in \pol{\wat}$ using closure under marked concatenation and union. We now show that $L_k = H_k$ which terminates the proof.

\medskip

We begin with $L_k \subseteq H_k$. Let $w \in L_k$. We consider two cases depending on whether $w$ contains the letter $a_k$ or not. If $a_k \not\in \cont{w}$, then $w \in (A \setminus \{a_k\})^*$ which is a subset of $H_k$ by definition. Hence, $w \in H_k$. Otherwise, $a_k \in \cont{w}$. We decompose $w$ by looking at the rightmost copy of the letter $a_k$: $w = ua_kv$ with $u \in A^*$ and $v \in (A \setminus \{a_k\})^*$. It is now immediate that $u \in L_{k-1}$. Indeed, otherwise the word $a_1 \cdots a_{k-1}$ would be a piece of $u$ and therefore $a_1 \cdots a_k$ would be a piece of $w$ which is not possible since $w \in L_k$ by hypothesis. We conclude that $w = ua_kv \in L_{k-1}a_k(A \setminus \{a_k\})^*$ which a subset of $H_k$ by definition. Thus, $w \in H_k$.

\medskip

We turn to the second inclusion: $H_k \subseteq L_k$. Let $w \in H_k$. If $w \in (A \setminus \{a_k\})^*$, then it is clear that $a_1 \cdots a_k$ is not a piece of $w$ which means that $w \in L_k$. Otherwise, $w \in L_{k-1}a_k(A \setminus \{a_k\})^*$. Thus, $w = ua_kv$ with $u \in L_{k-1}$ and $a_k \not\in \cont{v}$. By contradiction, assume that $a_1 \cdots a_k$ is a piece of $w$. Since $a_k \not\in \cont{v}$, it follows that $a_1 \cdots a_{k-1}$ must be a piece of $u$ which is impossible since $u \in L_{k-1}$. We conclude that $a_1 \cdots a_k$ is not a piece of $w$ which means that $w \in L_k$. This terminates the proof of Lemma~\ref{lem:hintro:atmainlem}.

\section{The link with logic}\label{sec:link-with-logic}
In this section, we present \emph{quantifier alternations hierarchies}, whose levels are defined by fragments of first-order logic. Such hierarchies classify languages according to the type of sentences needed to define them: the classifying parameter is the number of alternations between $\exists$ and $\forall$ quantifiers that are necessary to define a language.

The main theorem here is a \emph{generic} correspondence between concatenation and quantifier alternation hierarchies. For any basis \Cs, the concatenation hierarchy of basis \Cs corresponds exactly to the quantifier alternation hierarchy within a well-chosen variant of first-order logic. This generic connection was originally observed by Thomas~\cite{ThomEqu} who obtained it in a specific case. He showed that the dot-depth hierarchy corresponds exactly to the quantifier alternation hierarchy within the variant $\fo(\sigenr)$ of first-order~logic.

We first present first order logic and quantifier alternation hierarchies.
Then, we state and prove the main theorem of Section~\ref{sec:link-with-logic}: the correspondence between concatenation and quantifier alternation hierarchies. Finally, we instantiate this result on the dot-depth and Straubing-Thérien hierarchies.

\medskip\noindent\textbf{Quantifier alternation hierarchies.}
For defining languages with first-order logic, we view words as relational structures: a word of length~$n$ is a sequence of positions $\{1,\ldots,n\}$ labeled over alphabet~$A$. A \emph{signature} is a set of predicate symbols, each of them having an arity. Given a word of length $n$, a predicate of arity $k$ is interpreted as a $k$-ary relation on the set $\{1,\ldots,n\}$ of positions of the~word.
Important examples of predicates are the following:
\begin{itemize}
\item $\varepsilon$, the nullary ``empty'' predicate, which holds when the word is empty. That is, given a word $w$, the predicate $\varepsilon$ holds when $w=\varepsilon$.
\item For each $a\in A$, a unary ``label'' predicate, also denoted by $a$. Given a word $w$ and a position $i$ in $w$, $a(i)$ holds when position~$i$ in $w$ carries letter~$a$.
\item $\min(x)$, the unary ``minimum'' predicate, which selects the first position of a~word.
\item $\max(x)$, the unary ``maximum'' predicate, which selects the last position.
\item $<$, the binary ``order'' predicate, interpreted as the linear order on positions.
\item $+1$, the binary ``successor'' predicate, interpreted as the successor relation.
\end{itemize}

Each signature $\Ss$ defines a variant of first-order logic, which we denote by $\fo(\Ss)$. For concrete signatures, we will not write the label predicates, \emph{i.e.}, they will be always understood. For instance, $\fow$ denotes the variant of first-order logic for the signature consisting of the order predicate \emph{and} all label predicates.

For a given signature $\Ss$, we define the semantics of $\fo(\Ss)$ of first-order logic as follows: one may quantify over positions of a word, use Boolean connectives as well as the $\top$ (true) and $\bot$ (false) formulas, and test properties of the quantified positions using the predicate symbols from~$\Ss$. Each first-order sentence of $\fo(\Ss)$ therefore defines a language over $A^*$.

More formally, let $w = b_1 \cdots b_n \in A^*$ be a word and $\Xs$ be some finite set of first-order variables, an \emph{assignment of \Xs in $w$} is a map $\mu$ from \Xs to the set of positions of $w$ (i.e., $\mu: \Xs \to \{1,\dots,n\}$). In particular, if $\mu$ is an assignment of \Xs in $w$, $x$ a variable (not necessarily in $\Xs$) and $i$ a position in $w$, we will denote by $\mu[x \mapsto i] : (\Xs \cup \{x\}) \to \{1,\dots,n\}$, the assignment of $\Xs \cup \{x\}$ in $w$ that is identical to $\mu$ except that it maps $x$ to $i$. We can now define the semantic of a first-order formula.

Let $\varphi$ be a first-order formula and assume \Xs contains all free variables of $\varphi$. Then, for any word $w = b_1 \cdots b_n \in A^*$ and any assignment $\mu$ of \Xs in $w$, we say that $w$ \emph{satisfies} $\varphi$ under $\mu$, written $w,\mu \models \varphi$, when one the following properties hold:
\begin{itemize}
\item $\varphi :=$ ``$\top$''.
\item $\varphi :=$ ``$P(x_1,\dots,x_k)$'' for some predicate $P \in \Ss$ and $P(\mu(x_1),\dots,\mu(x_k))$ holds.
\item $\varphi :=$ ``$\exists x\ \Psi$'' and there exists a position $i \in \{1,\dots,n\}$ such that $w,\mu[x \mapsto i] \models \Psi$.
\item $\varphi :=$ ``$\forall x\ \Psi$'' and for any position $i \in \{1,\dots,n\}$, we have $w,\mu[x \mapsto i] \models \Psi$.
\item $\varphi :=$ ``$\Psi \vee \Gamma$'' and $w,\mu \models \Psi$ or $w,\mu \models \Gamma$.
\item $\varphi :=$ ``$\Psi \wedge \Gamma$'' and $w,\mu \models \Psi$ and $w,\mu \models \Gamma$.
\item $\varphi :=$ ``$\neg \Psi$'' and $w,\mu \not\models \Psi$ ($w$ does not satisfy $\Psi$ under $\mu$).
\end{itemize}
When $\varphi$ is a sentence, whether $w, \mu \models \varphi$ does not depend on $\mu$. In that case, we simply write $w, \mu \models \varphi$.  Any sentence $\varphi$ defines the language $\{w\in A^*\mid w\models\varphi\}$.

We now define a hierarchy of fragments within $\fo(\Ss)$ by classifying all $\fo(\Ss)$ sentences according to the number of quantifier alternations within their parse trees. For $i \in \nat$, a formula is $\sic{i}$(\Ss) (resp.~$\pic{i}(\Ss)$) if its prenex normal form has $(i-1)$ quantifier alternations (\emph{i.e.}, $i$ blocks of quantifiers) and starts with an $\exists$ (resp.\ a~$\forall$) quantification. For example, a formula whose prenex normal form is
\[
\forall x_1 \forall x_2\; \exists x_3\; \forall x_4
\ \varphi(x_1,x_2,x_3,x_4) \quad \text{(with $\varphi$ quantifier-free)}
\]
\noindent
is $\pic3(\Ss)$. Observe that a $\pic{i}(\Ss)$ formula is the negation of a $\sic{i}(\Ss)$ formula. Finally, a $\bsc{i}(\Ss)$ formula is a Boolean combination of $\sic{i}(\Ss)$ formulas. Note that by definition, we have $\bsc{i-1}(\Ss) \subseteq \sic{i}(\Ss) \subseteq \bsc{i}(\Ss)$ and $\bsc{i-1}(\Ss) \subseteq \pic{i}(\Ss) \subseteq \bsc{i}(\Ss)$ for any $i\geq1$. It is also clear that any $\fo(\Ss)$ formula belongs to some of these~classes.

We lift this syntactic definition to the semantic level: for $X = \fo(\Ss)$, $\sic{i}(\Ss)$, $\pic{i}(\Ss)$ or $\bsc{i}(\Ss)$, we say that a language $L$ is $X$-definable if it can be defined by an
$X$-formula. Abusing notation, we also denote by $X$ the class of $X$-definable languages. This gives us a hierarchy of languages depicted in \figurename~\ref{fig:qalt:hiera}.

\begin{figure}[H]
  \centering
  \begin{tikzpicture}
    \node[inner sep=1pt] (b0) at (0.0,0.0) {$\sic{0} = \pic{0} = \bsc{0}$ };

    \node[inner sep=1pt] (s1) at (1.5,-0.8) {\sicu};
    \node[inner sep=1pt] (p1) at (1.5,0.8) {\picu};
    \node[inner sep=1pt] (b1) at (3.0,0.0) {\bscu};

    \node[inner sep=1pt] (s2) at ($(b1)+(1.5,-0.8)$) {\sicd};
    \node[inner sep=1pt] (p2) at ($(b1)+(1.5,0.8)$) {\picd};
    \node[inner sep=1pt] (b2) at ($(b1)+(3.0,0.0)$) {\bscd};

    \node[inner sep=1pt] (s3) at ($(b2)+(1.5,-0.8)$) {\sict};
    \node[inner sep=1pt] (p3) at ($(b2)+(1.5,0.8)$) {\pict};
    \node[inner sep=1pt] (b3) at ($(b2)+(3.0,0.0)$) {\bsct};

    \node[inner sep=1pt] (fo) at ($(b3)+(1.8,0.0)$) {\fo};

    \draw [thick] (b0.-60) to [out=-90,in=180] node[linc2] {\scriptsize $\subseteq$} (s1.west);
    \draw [thick] (b0.60) to [out=90,in=-180] node[linc2] {\scriptsize $\subseteq$} (p1.west);

    \draw [thick] (s1.east) to [out=0,in=-90] node[linc2] {\scriptsize $\subseteq$} (b1.-120);
    \draw [thick] (p1.east) to [out=0,in=90] node[linc2] {\scriptsize $\subseteq$} (b1.120);

    \draw [thick] (b1.-60) to [out=-90,in=180] node[linc2] {\scriptsize $\subseteq$} (s2.west);
    \draw [thick] (b1.60) to [out=90,in=-180] node[linc2] {\scriptsize $\subseteq$} (p2.west);
    \draw [thick] (s2.east) to [out=0,in=-90] node[linc2] {\scriptsize $\subseteq$} (b2.-120);
    \draw [thick] (p2.east) to [out=0,in=90] node[linc2] {\scriptsize $\subseteq$} (b2.120);

    \draw [thick] (b2.-60) to [out=-90,in=180] node[linc2] {\scriptsize $\subseteq$} (s3.west);
    \draw [thick] (b2.60) to [out=90,in=-180] node[linc2] {\scriptsize $\subseteq$} (p3.west);
    \draw [thick] (s3.east) to [out=0,in=-90] node[linc2] {\scriptsize $\subseteq$} (b3.-120);
    \draw [thick] (p3.east) to [out=0,in=90] node[linc2] {\scriptsize $\subseteq$} (b3.120);

    \draw[thick,dotted] (b3) to [out=0,in=180] node[linc2,dotted] {\scriptsize
      $\subseteq$} (fo);
  \end{tikzpicture}
  \caption{Quantifier alternation hierarchy of first-order logic}
  \label{fig:qalt:hiera}
\end{figure}
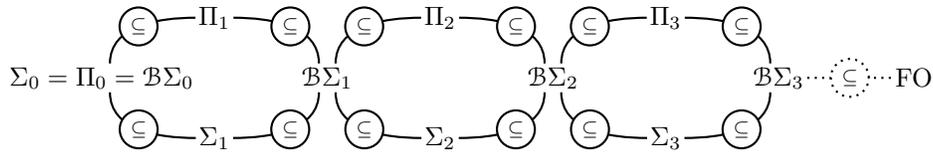

Whether a particular hierarchy is strict depends on its signature \Ss. The two most prominent hierarchies in the literature are known to be strict. These are:
\begin{itemize}
\item The \emph{order hierarchy} is the one associated to the logic \fow.
\item The \emph{enriched hierarchy} is the one associated to the logic \fows.
\end{itemize}

\begin{remark}
  It is a classical result that \fow and \fows have the same expressive power: all predicates available in \fows can be defined from the linear order. However, this is not the case for levels in their respective quantifier alternation hierarchies. Intuitively, the reason is that defining the predicates ``$+1$'', ``$\min$'' and ``$\max$'' from ``$<$'' costs quantifier alternations.
\end{remark}

Finally, a useful lemma is that we can bypass $\bsc{n}(\Ss)$ formulas in the definition of quantifier alternation~hierarchies.

\begin{lemma}\label{lem:qalt:bypassbool}
  For any $n \geq 0$, any $\sic{n+1}(\Ss)$ formula is equivalent to a formula of the form $\psi\vee\forall x\bot$ or $\psi\wedge\exists x\top$, where $\psi$ belongs to the closure of $\pic{n}(\Ss)$ under
  existential quantification.
\end{lemma}

\begin{proof}
  On nonempty words, any formula from $\sic{n+1}(\Ss)$ is equivalent to its prenex normal form, which by definition either belongs itself to $\pic{n}(\Ss)$, or is of the form $\exists x_1\ldots\exists x_k\psi$ where $\psi$ is a $\pic{n}(\Ss)$ formula. The disjunction with $\forall x\bot$ and the conjunction with $\exists x\top$ are used to add or remove the empty word from the language of the formula. This concludes the proof.
\end{proof}

\medskip\noindent\textbf{Main theorem.}
We are ready to present and prove the generic correspondence existing between quantifier alternation and concatenation hierarchies. More precisely, we show that for any basis \Cs, one may define an appropriate signature (also denoted by \Cs) such that the concatenation hierarchy of basis \Cs and the quantifier alternation hierarchy within $\fo(\Cs)$ are identical.

Consider an arbitrary basis \Cs. We associate a signature to \Cs and consider the variant of first-order logic equipped with this signature. As usual, the signature associated to \Cs contains all label predicates: for any $a \in A$, we have a unary predicate (also denoted by ``$a$'') which is interpreted as the unary relation selecting all positions whose label is $a$. Moreover, for any language $L \in \Cs$, we add four predicates:
\begin{itemize}
\item A binary predicate $I_L$ interpreted as follows: given a word $w$ and two positions $i,j$ in $w$, $I_L(i,j)$ holds when $i < j$ and the infix $w]i,j[ $ belongs to $ L$.
\item A unary predicate $P_L$ interpreted as follows: given a word $w$ and a position  $i$ in~$w$, $P_L(i)$ holds when the prefix $w[1,i[$ belongs to $L$.
\item A unary predicate $S_L$ interpreted  as follows: given a word $w$ and a position $i$ in~$w$, $S_L(i)$ holds when the suffix $w]i,|w|]$ belongs to $L$.
\item A nullary predicate $N_L$ interpreted as follows: given a word $w$, $N_L$ holds when~$w$ belongs to $L$.
\end{itemize}
Abusing notation, we denote by $\fo(\Cs)$ the associated variant of first-order logic.

\begin{remark}
  Observe that these signatures always contain the label predicates and the linear order ``$<$''. Indeed, by definition, ``$<$'' is the binary predicate $I_{A^*}$, and $A^*$  belongs to \Cs since it is a \vari. Thus, all variants of first-order logic that we consider here are at least as expressive as \fow. In fact, $\fow = \fo(\Cs)$ when $\Cs = \{\emptyset,A^*\}$. We shall detail this point in the next section.
\end{remark}

We now state the  theorem establishing an exact correspondence between the concatenation hierarchy of basis \Cs and the quantifier alternation hierarchy within~$\fo(\Cs)$.

\begin{theorem} \label{thm:qalt:maintheo}
  Let \Cs be a \vari. For any alphabet $A$, any $n \in \nat$ and any language $L \subseteq A^*$, the two following properties hold:
  \begin{enumerate}
  \item $L \in \Cs[n]$ if and only if $L$ can be defined by a $\bsc{n}(\Cs)$ sentence.
  \item $L \in \Cs[n + \frac{1}{2}]$ if and only if $L$ can be defined by a $\sic{n+1}(\Cs)$ sentence.
  \end{enumerate}
\end{theorem}

The rest of the section is devoted to proving Theorem~\ref{thm:qalt:maintheo}. A first observation is that we may concentrate on the second item as the first one is a simple corollary. Indeed, for $n = 0$, a $\bsc{0}(\Cs)$ sentence is by definition a Boolean combination of atomic formulas which do not involve variables. This includes $\bot,\top$ and the nullary predicates $N_L$ for $L \in \Cs$. By definition of the predicates $N_L$ and since \Cs is a Boolean algebra, it follows that $\bsc{0}(\Cs) = \Cs = \Cs[0]$. Next, for $n \geq 1$, we have $\bsc{n}(\Cs) = \bool{\sic{n}(\Cs)}$. Hence, the equality $\sic{n}(\Cs) = \Cs[n - \frac{1}{2}]$ immediately yields $\bsc{n}(\Cs) = \Cs[n]$.

We now concentrate on proving the second item in Theorem~\ref{thm:qalt:maintheo}. The proof is divided in two steps, one for each inclusion. We first show the easier one, namely,
\[
  \Cs[n + \tfrac{1}{2}] \subseteq \sic{n+1}(\Cs).
\]
The proof is an induction on $n$. The key ingredient is the following lemma which states that for any $n \in \nat$, $\sic{n}(\Cs)$ is closed under marked concatenation.

\begin{lemma} \label{lem:qalt:concat}
  Let $n \in\nat$ and $L_1,L_2 \subseteq A^*$ be two languages in $\sic{n}(\Cs)$. Then, for any $a \in A$, the marked concatenation $L_1aL_2$ also belong to $\sic{n}(\Cs)$.
\end{lemma}

\begin{proof}
  Let $L_1$ and $L_2$ be languages defined by two $\sic{n}(\Cs)$ sentences $\varphi_1$ and $\varphi_2$. We have to construct a third sentence $\Psi$ that defines $L_1aL_2$. Let $x$ be a fresh variable with respect to both $\varphi_1$ and $\varphi_2$. We build two formulas $\varphi'_1(x)$ and $\varphi'_2(x)$ (each with $x$ as a single free variable) with the following semantics. Given $w \in A^*$ and $\mu$ an assignment for $w$ with domain $\{x\}$:
  \begin{itemize}
  \item $w,\mu \models \varphi'_1(x)$ if and only if the prefix $w[1,\mu(x)[$ belongs to $L_1$ (that is, iff $w[1,\mu(x)[ \models \varphi_1$). Observe that this prefix may be empty when $\mu(x) = 1$.
  \item $w,\mu \models \varphi'_2(x)$ if and only if the suffix $w]\mu(x),|w|]$ belongs to $L_2$ (that is, iff $w]\mu(x),|w|] \models \varphi_2$). Observe that this suffix may be empty when $\mu(x) = |w|$.
  \end{itemize}
  The constructions of $\varphi'_1(x)$ and $\varphi'_2(x)$ are symmetrical. Let us describe that of $\varphi'_1(x)$. We build it from $\varphi_1$ as follows:
  \begin{enumerate}
  \item We relativize quantifications to positions that are to the left of $x$. That is, we replace every sub-formula of the form $\exists y\ \Gamma$ (resp. $\forall y\ \Gamma$) by $\exists y\ y < x \wedge \Gamma$ (resp. $\forall y\ \neg(y < x) \vee \Gamma$).
  \item We replace atomic formulas of the form $N_L$ for some $L \in \Cs$ by $P_L(x)$.
  \item We replace atomic formulas of the form $S_L(y)$ for some $L \in \Cs$ by $I_L(y,x)$.
  \end{enumerate}
  Clearly, $\varphi'_1(x)$ is also a $\sic{n}(\Cs)$ formula and one may verify that it satisfies the above property. We can now define $\Psi$ for $L_1aL_2$ as follows,  \[
    \Psi  = \exists x\ a(x) \wedge \varphi'_1(x) \wedge \varphi'_2(x).
  \]
  It is obvious that  $\Psi$ is a $\sic{n}(\Cs)$ sentence defining the language $L_1aL_2$.
\end{proof}

We may now prove that $\Cs[n + \frac{1}{2}] \subseteq \sic{n+1}(\Cs)$ for any $n \in \nat$. We proceed by induction on $n$. When $n = 0$, we first note that $\Cs \subseteq \sic{1}(\Cs)$. Indeed, any language $L$ of $\Cs$ is defined by the atomic sentence $N_L$. Therefore, $\Cs[\frac{1}{2}] = \pol{\Cs}\subseteq\sic{1}(\Cs)$, since
$\sic{1}(\Cs)$ is closed under union and marked concatenation.

When $n \geq 1$, we know that $\Cs[n + \frac{1}{2}] = \pol{\overline{\Cs[n - \frac{1}{2}]}}$ by Proposition~\ref{prop:hintro:bypassfull}. By induction hypothesis, we have $\Cs[n - \frac{1}{2}] \subseteq \sic{n}(\Cs)$ and therefore,
\[
  \overline{\Cs[n - \tfrac{1}{2}]} \subseteq \bsc{n}(\Cs) \subseteq \sic{n+1}(\Cs).
\]
Hence, since $\sic{n+1}(\Cs)$ is closed under union and marked concatenation, we obtain as desired that $\Cs[n+\frac{1}{2}] \subseteq \sic{n+1}(\Cs)$, finishing the proof for this direction.

\medskip
It remains to establish the converse inclusion, \emph{i.e.}, that for any $n \in \nat$:
\begin{equation}\label{eq:concatlog}
  \sic{n+1}(\Cs) \subseteq \Cs[n + \tfrac{1}{2}].
\end{equation}
Since the proof works inductively on the formulas, we have to explain how we handle free variables. We do this using Büchi's classical idea, \emph{i.e.}, by encoding a word and a assignment of first-order variables as a single word over an extended alphabet.

Let $\Xs = \{x_1,x_2,x_3,\dots\}$ be an infinite linearly ordered set of first-order variables. One may assume that all $\fo(\Cs)$ formulas that we consider only use variables from~\Xs. Given $\ell \in \nat$, we use the alphabet $A_{\ell} = {\{0,1\}}^{\ell} \times A$ to represent pairs $(w,\mu)$ with $w \in A^*$ and $\mu$ an assignment of $\{x_1,\dots,x_{\ell}\}$ in the positions of $w$.

For any $h \leq \ell$, we denote by $\pi_h : A_\ell \to \{0,1\}$ the projection on component $h$. Similarly, we denote by $\pi_A : A_\ell \to A$ the projection on the rightmost component (component $\ell+1$). Note that there are actually several mapping $\pi_A$, one for each value of $\ell$, and similarly for $\pi_h$. Which mapping we use will be clear from the context.

 We can now present the encoding. Let $w = a_1 \cdots a_n \in A^*$ and let $\mu$ be an assignment of $\{x_1,\dots,x_{\ell}\}$ in $w$. We encode the pair $(w,\mu)$ by the word ${[w]}_{\mu} = \overline{b_1} \cdots \overline{b_n} \in {(A_{\ell})}^*$ such that for all $i \leq n$, $\overline{b_i} \in A_\ell$ is defined as follows,

\begin{itemize}
\item $\pi_{A}(\overline{b_i}) = a_i$.
\item For all $h \leq \ell$,
  \begin{itemize}
  \item If $i = \mu(x_h)$, $\pi_h(\overline{b_{i}}) = 1$.
  \item If $i \neq \mu(x_h)$, $\pi_h(\overline{b_i}) = 0$.
  \end{itemize}
\end{itemize}

Note that when $\ell = 0$, we have $A_0 = A$ and ${[w]}_{\mu} = w$ ($\mu$ is the empty assignment). Clearly, the map $(w,\mu) \mapsto {[w]}_\mu$ is injective (however, it is not surjective since for any $h \leq \ell$, there is  exactly one position $i$ such that $\pi_h(\overline{b_i}) = 1$).
For $\ell \geq 1$, we define the following class of languages over the alphabet $A_\ell$:
\[
  \Cs_\ell \stackrel{\text{def}}= \{\pi_A^{-1}(L)  \subseteq A_\ell^* \mid L \in \Cs\}.
\]
It straightforward to verify that $\Cs_\ell$ is a \vari of regular languages. Moreover, for any $\ell \in \nat$, we define a morphism $\alpha_\ell: A_\ell^* \to A_{\ell+1}^*$ as follows: given $(i_1,\dots,i_\ell,a) \in A_\ell$, we let $\alpha_\ell(i_1,\dots,i_\ell,a) = (i_1,\dots,i_\ell,0,a) \in A_{\ell+1}$. We now state a connection between the concatenation hierarchies of bases $\Cs_\ell$ and~$\Cs_{\ell+1}$.

\begin{fact} \label{fct:qalt:invmorph}
  For any $\ell,n\in\nat$ and any $K \in \Cs_{\ell+1}[n+\frac{1}{2}]$, we have $\alpha_\ell^{-1}(K) \in \Cs_\ell[n+\frac{1}{2}]$.
\end{fact}

\begin{proof}
  This is immediate by induction on $n$ and the definition of concatenation hierarchies using the following properties. For any $K_1,K_2 \subseteq A_{\ell+1}^*$, we have,
  \begin{enumerate}
  \item By definition of \Cs, when $K_1 \in \Cs_{\ell+1}$, we have $\alpha_\ell^{-1}(K_1) \in \Cs_\ell$.
  \item $\alpha_\ell^{-1}(K_1 \cup K_2) = \alpha_\ell^{-1}(K_1) \cup \alpha_\ell^{-1}(K_2)$.
  \item $\alpha_\ell^{-1}(A_{\ell+1}^* \setminus K_1) = A_{\ell}^* \setminus \alpha_\ell^{-1}(K_1)$.
  \item For any $\overline{b} \in A_{\ell+1}$, we have $\alpha_\ell^{-1}(K_1aK_2) = \alpha_\ell^{-1}(K_1) \alpha_\ell^{-1}(a)\alpha_\ell^{-1}(K_2)$.
  \end{enumerate}
  This concludes the proof of Fact~\ref{fct:qalt:invmorph}.
\end{proof}

The proof of the remaining inclusion~\eqref{eq:concatlog} relies on the following proposition.

\begin{proposition}\label{prop:qalt:qaltsfr}
  Let $\ell,n \in \nat$ and let $\varphi$ be a $\sic{n+1}(\Cs)$ formula whose set of free variables is included in $\{x_1,\dots,x_{\ell}\}$. Then, there exists a language $L_{\ell,\varphi} \in \Cs_\ell[n+\frac{1}{2}]$ such that for any $w \in A^*$ and any assignment $\mu$ of $\{x_1,\dots,x_{\ell}\}$ in $w$, we have,
  \begin{equation}\label{eq:w_mu-in-l_ell}
    {[w]}_\mu \in L_{\ell,\varphi} \quad \text{if and only if} \quad w,\mu \models \varphi.
\end{equation}

\end{proposition}

Note that the special case $\ell = 0$ of Proposition~\ref{prop:qalt:qaltsfr} yields the following corollary.

\begin{corollary}\label{cor:qalt:qaltsfr}
  Let $n \in \nat$ and let $\varphi$ be a $\sic{n+1}(\Cs)$ sentence. Then, there exists a language $L \in \Cs[n+\frac{1}{2}]$ such that for any $w \in A^*$, we have,
  \[
    w \in L \quad \text{if and only if} \quad w \models \varphi.
  \]
\end{corollary}

Corollary~\ref{cor:qalt:qaltsfr} implies that for all $n \in \nat$, we have $\sic{n+1}(\Cs) \subseteq \Cs[n + \tfrac{1}{2}]$, which is the inclusion~\eqref{eq:concatlog} that remained to be proved, concluding the proof of Theorem~\ref{thm:qalt:maintheo}.

\medskip

It remains to prove Proposition~\ref{prop:qalt:qaltsfr}. Let $\ell,n \in \nat$ and let $\varphi$ be a $\sic{n+1}(\Cs)$ formula whose set of free variables is included in $\{x_1,\dots,x_{\ell}\}$. We construct $L_{\ell,\varphi} \in \Cs_\ell[n+\frac{1}{2}]$ satisfying the conditions in Proposition~\ref{prop:qalt:qaltsfr} by induction on $n$.

Recall that we showed in Lemma~\ref{lem:qalt:bypassbool} that we may assume without loss of generality that $\varphi$ is built from negations of $\sic{n}(\Cs)$ formulas using
existential quantifications. We use a sub-induction on this construction. We start with the base case which is different depending on whether $n = 0$ or $n \geq 1$ (essentially the former amounts to treating atomic formulas while the later is immediate by induction on $n$).

\medskip\noindent
\textbf{Base case.} $\varphi$ is the negation $\varphi = \neg \psi$ of some $\sic{n}(\Cs)$ formula $\psi$.

We first treat the case $n \geq 1$, which is where we use induction on $n$. Indeed, induction yields $L_{\ell,\psi} \in \Cs_\ell[n-\frac{1}{2}]$ such that for any $w \in A^*$ and any assignment $\mu$ of $\{x_1,\dots,x_{\ell}\}$ in $w$, we have,
\[
  {[w]}_\mu \in L_{\ell,\psi} \quad \text{if and only if} \quad w,\mu \models \psi.
\]
Hence, it suffices to choose $L_{\ell,\varphi} = A_\ell^* \setminus L_{\ell,\psi} \in \Cs_\ell[n] \subseteq \Cs_\ell[n+\frac{1}{2}]$, which clearly meets the conditions in Proposition~\ref{prop:qalt:qaltsfr}.

It remains to treat the case $n = 0$. By definition, the $\sic{0}(\Cs)$ formulas are the quantifier-free formulas. Thus, $\varphi = \neg \psi$ is itself a $\sic{0}(\Cs)$ formula. In other words $\varphi$ is a Boolean combination of atomic formulas. Moreover, if we allow the equality predicate in the signature, we may eliminate all negations in $\varphi$. Indeed, using DeMorgan's laws, one may push all negations to atomic formulas. Furthermore, given any atomic formula, its negation is equivalent to a $\sic{0}(\Cs)$ formula without negation (this is where we need equality). Indeed, given $a \in A$, $\neg a(x)$ is equivalent to $\bigvee_{c \neq a} c(x)$. Finally, for any $K \in \Cs$, we have the following (recall that since \Cs is a \vari, $A^* \setminus H $ belongs to $\Cs$ as well),
\begin{itemize}
\item $\neg I_K(x,y)$ is equivalent to $I_{A^*}(y,x) \vee x = y \vee I_{A^*\setminus K}(x,y)$.
\item $\neg P_K(x)$ is equivalent to $P_{A^*\setminus K}(x)$.
\item $\neg S_K(x)$ is equivalent to $S_{A^*\setminus K}(x)$.
\item $\neg N_K$ is equivalent to $N_{A^*\setminus K}$.
\end{itemize}

Hence, we may assume without loss of generality that there are no negation in $\varphi$, which is therefore in $\sic{0}(\Cs)$. Hence, $\varphi$ is built from atomic formulas using conjunctions and disjunctions. We may handle disjunctions and conjunctions in the obvious way. Hence, it suffices to treat the cases when $\varphi$ is atomic.

There are two kinds of atomic formulas: those involving the label predicates and those which are specific to \Cs. Moreover, we also need to treat equality since we used it above to eliminate negations. Let us first assume that $\varphi =a(x_h)$ for some $h \leq \ell$ and some $a \in A$. Consider the set $B$ of all letters in $A_\ell$ whose component $h$ is equal to $1$ and whose component $\ell+1$ is equal to $a$:
\[
  B = \{\overline{b} \in A_\ell \mid \pi_{h}(\overline{b}) = 1 \text{ and } \pi_A(\overline{b}) = a\}.
\]
It now suffices to define $L_{\ell,\varphi} = A_\ell^* B A_\ell^* \in \pol{\Cs_\ell} = \Cs_\ell[\frac{1}{2}]$. It is then immediate from the definitions that $L_{\ell,\varphi}$ satisfies the conditions in Proposition~\ref{prop:qalt:qaltsfr}.

\smallskip

We now assume that $\varphi :=$ ``$x_g = x_h$'' for some $g,h \leq \ell$. We now let $B$ as the set of all letters in $A_\ell$ whose components $g$ and $h$ are both equal to $1$.
\[
  B = \{\overline{b} \in A_\ell \mid \pi_{g}(\overline{b}) = 1 \text{ and } \pi_h(\overline{b}) = 1\}.
\]
It now suffices to define $L_{\ell,\varphi} = A_\ell^* B A_\ell^* \in \pol{\Cs_\ell} = \Cs_\ell[\frac{1}{2}]$. It is then immediate from the definitions that $L_{\ell,\varphi}$ satisfies the conditions in Proposition~\ref{prop:qalt:qaltsfr}.

\smallskip

It remains to treat the predicates given by \Cs. Since the argument is the same for all four kinds, we only treat the case when $\varphi =I_K(x_i,x_j)$, for some $K \in \Cs$ and $i,j \leq \ell$. We may assume that $i\not=j$, since $I_K(x_i,x_i)$ is equivalent to $\bot$. By symmetry, we may then assume that $i<j$. Let $B_i$ and $B_j$ be the following sub-alphabets of $A_\ell$:
\[
  \left\{
  \begin{array}{lll}
    B_i & = & \{\overline{b} \in A_\ell \mid \pi_i(\overline{b}) = 1\},\\
    B_j  & = & \{\overline{b} \in A_\ell \mid \pi_j(\overline{b}) = 1\}.
  \end{array}
  \right.
\]
We define $L_{\ell,\varphi} = A_\ell^* B_i \pi_A^{-1}(K) B_j A_\ell^*$. Recall that the language $\pi_A^{-1}(K) \subseteq A_\ell^*$ belongs to $\Cs_\ell$ (by definition of~$\Cs_\ell$). Hence, we have $L_{\ell,\varphi} \in \pol{\Cs} = \Cs[\frac{1}{2}]$. One may then verify that $L_{\ell,\varphi}$ satisfies the conditions in Proposition~\ref{prop:qalt:qaltsfr}.

\medskip

This concludes the base case of our structural induction on the formula $\varphi$. We now consider the inductive case which are handled uniformly for $n = 0$ and $n \geq 1$.

\medskip\noindent
\textbf{Inductive case: First-order quantification}. Assume that $\varphi$ is of the form $\exists x\ \psi$. Since variables can be renamed, we may assume without loss of generality that $x = x_{\ell+1}$, \emph{i.e.}, $\varphi=\exists x_{\ell+1}\ \psi$. This means that all free variables of $\psi$ belong to $\{x_1,\dots,x_{\ell+1}\}$. Applying induction to $\psi$ yields a language $L_{\ell+1,\psi} \in \Cs_{\ell+1}[n+\frac{1}{2}]$ such that for any $w \in A^*$ and any assignment $\gamma$ of $\{x_1,\dots,x_{\ell+1}\}$ in $w$, we have,
\[
  {[w]}_\gamma \in L_{\ell+1,\psi} \quad \text{if and only if} \quad w,\gamma \models \psi.
\]
We first define $L_{\ell,\varphi} \in  \Cs_{\ell}[n+\frac{1}{2}]$ and then prove that it satisfies~\eqref{eq:w_mu-in-l_ell}. Given any word $u \in A_{\ell+1}^*$, we say that $u$ is \emph{good} when there exists exactly one position in $u$ whose label $\overline{b}$ satisfies $\pi_{\ell+1}(\overline{b}) = 1$ (which implies that the labels~$\overline{c}$ of all other positions satisfy $\pi_{\ell+1}(\overline{c}) = 0$).  Let $\pi_{1,\dots,\ell,A}: A_{\ell+1}^* \to A_{\ell}^*$ be the projection which discards component $\ell+1$ in words belonging to $A_{\ell+1}^*$. More precisely,
\[
  \pi_{1,\dots,\ell,A}(i_1,\dots,i_{\ell+1},a) = (i_1,\dots,i_\ell,a).
\]
We now define $L_{\ell,\varphi} \subseteq A_\ell^*$ as the following language:
\[
  L_{\ell,\varphi} = \{\pi_{1,\dots,\ell,A}(u) \mid \text{$u \in L_{\ell+1,\psi}$ and $u$ is good}\}.
\]
It remains to prove that $L_{\ell,\varphi} \in \Cs_{\ell}[n+\frac{1}{2}]$ and that it satisfies Property~\eqref{eq:w_mu-in-l_ell} from Proposition~\ref{prop:qalt:qaltsfr}.
We first deal with Property~\eqref{eq:w_mu-in-l_ell}.

\begin{lemma} \label{lem:qalt:exist1}
  Let $w \in A^*$ and let $\mu$ be an assignment of $\{x_1,\dots,x_{\ell}\}$ in the positions of $w$. Then, we have,
  \[
    {[w]}_\mu \in L_{\ell,\varphi} \quad \text{if and only if} \quad w,\mu \models \varphi.
  \]
\end{lemma}

\begin{proof}
  Assume first that ${[w]}_\mu \in L_{\ell,\varphi}$. By definition, there exists $u \in L_{\ell+1,\psi}$ which is good and such that $\pi_{1,\dots,\ell,A}(u) = {[w]}_\mu$. Since $u$ is good, there exists exactly one position in $u$ whose label $\overline{b}$ satisfies $\pi_{\ell+1}(\overline{b}) = 1$. Let $i$ be this position and let $\gamma$ be the assignment $\mu[x_{\ell+1} \mapsto i]$ of $\{x_1,\dots,x_{\ell+1}\}$ in $w$. It follows immediately from the definitions that $u = {[w]}_{\gamma}$.
  Since $u \in L_{\ell+1,\psi}$, it follows that $w,\gamma \models \psi$, which exactly says that $w,\mu \models \varphi$ since $\varphi=\exists x_{\ell+1}\ \psi$ and $\gamma = \mu[x_{\ell+1} \mapsto i]$.

  Conversely, assume that $w,\mu \models \varphi$. It follows that there exists a position $i$ in $w$ such that $w,\mu[x_{\ell+1} \mapsto i] \models \psi$. Let $\gamma = \mu[x_{\ell+1} \mapsto i]$. By definition of $L_{\ell+1,\psi}$, we have, ${[w]}_{\gamma} \in L_{\ell+1,\psi}$.  Clearly, ${[w]}_{\gamma}$ is good and therefore, we have, ${[w]}_{\mu} = \pi_{1,\dots,\ell,A}({[w]}_{\gamma}) \in L_{\ell,\varphi}$.  This concludes the proof.
\end{proof}

\noindent
It remains to prove that $L_{\ell,\varphi} \in \Cs_{\ell}[n+\frac{1}{2}]$. The argument is based on the next lemma.

\begin{lemma}[Splitting lemma]\label{lem:fo:split}
  Let $\Cs$ be a \pvari of regular languages over $A$ and let $B \subseteq A$. Consider a language $L \in \Cs$. Then, $L \cap A^*BA^*$ is a finite union of languages of the form $PbS$ where $b \in B$ and $P,S \in \Cs$.
\end{lemma}

\begin{proof}
  First observe that we may assume without loss of generality that $B$ is a singleton $\{b\}$. Indeed, we have
  $L \cap A^*BA^* = \bigcup_{b \in B} L \cap A^*bA^*$.
  Hence, it suffices to apply the lemma in the singleton case for each language $L \cap A^*bA^*$. Therefore, we now assume that $B$ is a singleton $\{b\}$.

  For any $u \in A^*$, let $Q_u = (ub)^{-1}L = \{v \in A^* \mid ubv \in L\}$. Consider the following language $L'$:
  \begin{equation} \label{eq:fo:split}
    L' = \bigcup_{u \in A^*} \left(\bigcap_{v \in Q_u} L(bv)^{-1}\right) \cdot b \cdot (ub)^{-1}L
  \end{equation}
  We claim that $L \cap A^*bA^* = L'$. Before we prove this equality, let us explain why it concludes the proof. Since $L \in \Cs$, we know by hypothesis on~\Cs that $L$ is regular. Hence, it follows from Myhill-Nerode Theorem (Theorem~\ref{thm:auto:nerode}) that there are finitely many quotients of $L$. In particular, this means that in~\eqref{eq:fo:split}, the union over all $u \in A^*$ and the intersections over all $v \in Q_u$ are actually finite. Moreover, since $\Cs$ is a \pvari, we obtain that for any $u \in A^*$,
  \[
    \bigcap_{v \in Q_u} L(bv)^{-1} \in \Cs \quad \text{and} \quad (ub)^{-1}L \in \Cs.
  \]
  Hence, this conclude the proof of Lemma~\ref{lem:fo:split}: $L \cap A^*bA^*$ is a finite union of languages of the form $PbS$ where $P,S \in \Cs$. It remains to prove that $L \cap A^*bA^* = L'$.

  To prove that $L \cap A^*bA^* = L'$, assume first that $w \in L \cap A^*bA^*$. It follows that $w = ubv' \in L$ for some $u,v' \in A^*$. Hence, $v' \in (ub)^{-1}L$. Moreover, $u \in L(bv)^{-1}$ for any $v \in Q_u$ by definition. Hence, $u \in \bigcap_{v \in Q_u} L(bv)^{-1}$. We now conclude that,
  \[
    w \in  \left(\bigcap_{v \in Q_u} L(bv)^{-1}\right) \cdot b \cdot (ub)^{-1}L.
  \]
  Therefore, $w \in L'$. We have proved that $L \cap A^*bA^* \subseteq L'$.

  Conversely, assume that $w \in L'$. We obtain $u \in A^*$ such that $w$ admits a decomposition $w = u'bv'$ with $u' \in \bigcap_{v \in Q_u} L(bv)^{-1}$ and $v' \in (ub)^{-1}L$. In particular, since $v' \in (ub)^{-1}L$, we have $ubv' \in L$ which means that $v' \in Q_u$ by definition. Combined with the fact that $u' \in \bigcap_{v \in Q_u} L(bv)^{-1}$, this yields $u' \in L(bv')^{-1}$, which exactly says that $w = u'bv' \in L \cap A^*bA^*$.
\end{proof}

Let $B \subseteq A_{\ell + 1}$ be the set of all letters in $A_{\ell+1}$ whose component $\ell+1$ is $1$:
\[
  B = \{\overline{b} \in A_{\ell+1} \mid \pi_{\ell+1}(\overline{b}) = 1\}.
\]
Note that by definition, any \emph{good} word $u \in A_{\ell+1}^*$ belongs to $A_{\ell+1}^*BA_{\ell+1}^*$. Since by Proposition~\ref{prop:hintro:concatvari}, $\Cs[n+\frac{1}{2}]$ is a \pvari of regular languages, we may apply Lemma~\ref{lem:fo:split} to $L_{\ell+1,\psi} \in \Cs_{\ell+1}[n+\frac{1}{2}]$:
\begin{equation}
  \label{eq:fo:papert}
  L_{\ell+1,\psi} \cap A_{\ell+1}^*BA_{\ell+1}^* = \bigcup_{j \leq m} P_j\overline{b_j}S_j
\end{equation}
where for all $j \leq m$, $\overline{b_j} \in B$ and $P_j,S_j \in \Cs_{\ell+1}[n+\frac{1}{2}]$. For all $j \leq m$, let $\overline{c_j} = \pi_{1,\dots,\ell,A}(\overline{b_j}) \in A_\ell$. Recall that $\alpha_\ell: A_\ell^* \to A_{\ell+1}^*$ is defined as the following morphism. For any letter $(i_1,\dots,i_\ell,a) \in A_\ell$, we have $\alpha(i_1,\dots,i_\ell,a) = (i_1,\dots,i_\ell,0,a) \in A_{\ell+1}$. We have the following fact.

\begin{fact} \label{fct:qalt:finalarg}
  $L_{\ell,\varphi} = \bigcup_{j \leq m} \alpha_\ell^{-1}(P_j) \overline{c_j} \alpha_\ell^{-1}(S_j)$
\end{fact}

\begin{proof}
  We first consider $v \in  L_{\ell,\varphi}$. We have to find $j \leq m$ such that $v\in\alpha^{-1}(P_j) \overline{c_j} \alpha^{-1}(S_j)$. By definition of $L_{\ell,\varphi}$, we get $u \in L_{\ell+1,\psi}$ which is good and such that $v = \pi_{1,\dots,\ell,A}(u)$. Since $u$ is good, we have,
  \[
    u \in L_{\ell+1,\psi} \cap A_{\ell+1}^*BA_{\ell+1}^*
  \]
  It then follows from~\eqref{eq:fo:papert} that we have $u \in P_j\overline{b_j}S_j$ for some $j \leq m$. Hence, we may decompose $u$ as $u = u_1\overline{b_j} u_2$ with $u_1 \in P_j$ and $u_2 \in S_j$. Therefore, we have,
  \[
    v = \pi_{1,\dots,\ell,A}(u) = \pi_{1,\dots,\ell,A}(u_1 \overline{b_j} u_2) = \pi_{1,\dots,\ell,A}(u_1) \overline{c_j} \pi_{1,\dots,\ell,A}(u_2)
  \]
  Finally, since $u = u_1\overline{b_j} u_2$ is good and $\overline{b_j} \in B$, we know that $\overline{b_j}$ is the only letter in $u$ whose component $\ell+1$ is equal to $1$. Hence, the component $\ell+1$ of any letter in $u_1$ or $u_2$ is $0$. By definition of $\alpha$, it follows that $\alpha(\pi_{1,\dots,\ell,A}(u_1)) = u_1$ and $\alpha(\pi_{1,\dots,\ell,A}(u_2)) = u_2$. Thus, since $u_1 \in P_j$ and $u_2 \in S_j$, we get $\pi_{1,\dots,\ell,A}(u_1) \in \alpha^{-1}(P_j)$ and $\pi_{1,\dots,\ell,A}(u_2) \in \alpha^{-1}(S_j)$. Finally, this yields $v \in \alpha^{-1}(P_j) \overline{c_j} \alpha^{-1}(S_j)$ which concludes this direction of the proof.

  Conversely, assume that $v \in \alpha^{-1}(P_j) \overline{c_j} \alpha^{-1}(S_j)$ for some $j \leq m$. We have to prove that $v \in  L_{\ell,\varphi}$. By hypothesis, we have $v = v_1 \overline{c_j} v_2$ with $v_1 \in \alpha^{-1}(P_j)$ and $v_2 \in \alpha^{-1}(S_j)$. Consider the following word $u \in A_{\ell+1}^*$:
  \[
    u = \alpha(v_1) \overline{b_j} \alpha(v_2) \in P_j \overline{b_j} S_j.
  \]
  Observe that by definition, $u$ is good and $v = \pi_{1,\dots,\ell,A}(u)$. Moreover, it follows from~\eqref{eq:fo:papert} that $u \in L_{\ell+1,\psi}$. Thus, $v \in L_{\ell,\varphi}$ by definition of $L_{\ell,\varphi}$.
\end{proof}

Fact~\ref{fct:qalt:finalarg} concludes the proof since it is immediate from Fact~\ref{fct:qalt:invmorph} that for all $j \leq m$, $\alpha_\ell^{-1}(P_j)$ and $\alpha_\ell^{-1}(S_j)$ both belong to $\Cs_{\ell}[n+\frac{1}{2}]$.  Thus, we obtain Fact~\ref{fct:qalt:finalarg} that the language $L_{\ell,\varphi}$ is a finite union of marked concatenations of languages in $\Cs_{\ell}[n+\frac{1}{2}]$ and therefore belongs to $\Cs_{\ell}[n+\frac{1}{2}]$ itself.

\medskip\noindent\textbf{Back to the dot-depth and Straubing-Thérien hierarchies.}
We now apply Theorem~\ref{thm:qalt:maintheo} to the two classical examples: the dot-depth and Straubing-Thérien hierarchies. Note that the logical characterization of the dot-depth hierarchy is historically the first result of this kind which was discovered by Thomas~\cite{ThomEqu}. Therefore, Theorem~\ref{thm:qalt:maintheo} is a generalization of this original result.

Recall that the basis of the dot-depth hierarchy is $\dotzer =\{\emptyset,\{\varepsilon\},A^+,A^*\}$. It turns out that the associated variant of first-order logic ($\fo(\dotzer)$) is exactly  $\fo(\sigenr)$. Indeed, according to Theorem~\ref{thm:qalt:maintheo} the predicates available in $\fo(\dotzer)$ are as follows:
\begin{enumerate}
\item The label predicates.
\item The predicate $I_{A^*}$, which by definition is equivalent to the order predicate~$<$.
\item The predicates $P_{A^*},S_{A^*},N_{A^*}$ which always hold, hence they are equivalent to~$\top$.
\item The predicates $I_\emptyset,P_\emptyset,S_\emptyset,N_\emptyset$ which never hold, hence they are equivalent to  $\bot$.
\item The predicate $I_{\varepsilon}$, which by definition is equivalent to the successor predicate~$+1$.
\item The predicate $I_{A^+}$, and by definition $I_{A^+}$ is equivalent to $\neg I_{\varepsilon}$.
\end{enumerate}
Thus, the only useful predicates in $\fo(\dotdp{0})$ are exactly those that are available in $\fo(\sigenr)$: the label predicates, the linear order predicate, and the successor predicate. Therefore, we re-obtain Theorem~\ref{thm:thomas-citethomequ} as a corollary of Theorem~\ref{thm:qalt:maintheo}.

On the other hand, the basis of the Straubing-Thérien hierarchy is $\dotzer =\{\emptyset,A^*\}$, so that we miss the predicates $I_{\varepsilon}$ and $I_{A^+}$. Therefore, Theorem~\ref{thm:perr-pin-citepp} is also a corollary of Theorem~\ref{thm:qalt:maintheo}.

\section{Conclusion}\label{sec:conc}
In this paper, we surveyed 50 years of progress in the understanding of concatenation hierarchies. We presented a new proof that closure under intersection is implied by polynomial closure if the class we start from is a \pvari. We then established that if level 0 is a finite \vari, then the corresponding hierarchy is strict and we stated that its levels $\frac12$, 1, $\frac32$ have decidable separation. We stated a result transferring decidability of separation for some level to decidability of membership for the next half level, entailing that level $\frac52$ has decidable membership for finitely based hierarchies. We also observed that in the Straubing-Thérien hierarchy, level $q\geq1$ ($q\in \nat$ or $q\in\frac12+\nat$) coincides with level $q-1$ in a concatenation hierarchy whose basis is finite, hence we got decidability of separation for levels  $\frac12$,~1,~$\frac32$,~2 and~$\frac52$ in the Straubing-Thérien hierarchy, and decidability of membership for level~$\frac72$. We then transferred all these results to the dot-depth hierarchy via a generic construction. Finally, we proved a generic logical definition of concatenation hierarchies encompassing the ones established for the dot-depth and the Straubing-Thérien hierarchies.

\smallskip
Some of the research directions following this work are obvious: it is desirable to generalize this approach to capture all levels of such a concatenation hierarchy. This seems however to be difficult. We would also want to test such techniques for other structures than words, for instance, trees. Another short-term interesting topic is to reprove and generalize results that were obtained in the particular case of the Straubing-Thérien hierarchy regarding \emph{unambiguous closure}. We leave this question for a future~work.
 
\printbibliography

\end{document}